\documentclass{amsart}

\usepackage{amsmath}
\usepackage{amsthm}
\usepackage{amsfonts}
\usepackage{amssymb}
\usepackage{bbm}
\usepackage{chngcntr}
\usepackage{enumerate}
\usepackage{float}
\usepackage[bottom]{footmisc}
\usepackage[margin=1.5in]{geometry}
\usepackage{ifthen}
\usepackage{mathrsfs}
\usepackage{mathtools}
\usepackage[new]{old-arrows}
\usepackage{rotating}
\usepackage{thmtools}
\usepackage{tikz}
\usetikzlibrary{calc}
\usepackage{tikz-cd}
\usepackage{xcolor}
\usepackage[pdftex,pdfpagelabels,bookmarks,hyperindex,hyperfigures]{hyperref}
\hypersetup{
  colorlinks,
  linkcolor={red!50!black},
  urlcolor={blue!80!black},
  citecolor={green!50!black}
}

\newcommand{\diagscale}{1}

\newcommand{\M}{\mathbb{M}}
\newcommand{\N}{\mathbb{N}}

\newcommand{\R}{\mathbb{R}}
\newcommand{\Z}[1][]{\ifthenelse{\equal{#1}{}}{\mathbb{Z}}
                                              {\mathbb{Z}/#1\mathbb{Z}}}
\newcommand{\C}{\mathbb{C}}

\newcommand{\K}{\mathbbm{k}}

\newcommand{\s}{\mathscr}

\newcommand{\wh}{\widehat}
\newcommand{\tensor}{\otimes}
\newcommand{\cwedge}{\curlywedge}
\newcommand{\cvee}{\curlyvee}

\renewcommand{\to}[1][]{\stackrel{#1}{\longrightarrow}}
\newcommand{\hto}[1][]{\stackrel{#1}{\longhookrightarrow}}
\newcommand{\hot}[1][]{\stackrel{#1}{\longhookleftarrow}}
\newcommand{\To}[1][]{\stackrel{#1}{\Longrightarrow}}
\newcommand{\ot}[1][]{\stackrel{#1}{\longleftarrow}}

\renewcommand{\mapsto}{\longmapsto}

\newcommand{\codom}{\text{codom }}
\newcommand{\dom}{\text{dom }}
\newcommand{\id}{\text{id}}

\newcommand{\br}[1]{\left( #1 \right)}
\newcommand{\curly}[1]{\left\{ #1 \right\}}
\newcommand{\set}[2][]{\ifthenelse{\equal{#1}{}}
                                  {\curly{#2}}
                                  {\curly{#1\ \textbf{:}\ #2}}}
\newcommand{\sbr}[1]{\left[ #1 \right]}

\newcommand{\bmat}[1]{\begin{bmatrix} #1 \end{bmatrix}}

\newcommand{\Aut}{\text{Aut}}
\newcommand{\ket}[1]{\left| #1 \right\rangle}
\newcommand{\GL}{\text{GL}}
\newcommand{\Lin}{\textrm{Lin}}

\newcommand{\Set}{\textbf{Set}}

\newcommand{\Cob}{\textbf{Cob}}
\newcommand{\CCob}{\mathbb{C}\textbf{ob}}
\newcommand{\Man}{\textbf{Man}}
\newcommand{\Vect}{\textbf{Vect}}
\newcommand{\FVect}{\textbf{FVect}}
\newcommand{\FFVect}{\mathbb{F}\textbf{Vect}}
\newcommand{\Thick}{\textbf{Thick}}

\newcommand{\CPS}[1][]{\text{CP}^*\ifthenelse{\equal{#1}{}}{}{\sbr{#1}}}

\newcommand{\Hom}{\text{Hom}}
\newcommand{\Ob}{\textbf{ob }}

\newcommand{\DThick}{2\mathbb{T}\mathbf{hick}}
\newcommand{\PTT}{\mathcal{PTT}}
\newcommand{\Diff}{\text{Diff}}
\newcommand{\Bun}{\textbf{Bun}}
\newcommand{\BBun}{\mathbb{B}\textbf{un}}
\newcommand{\Conn}{\mathfrak{Conn}}

\newcommand{\TG}{\mathfrak{TG}}

\newcommand{\Exp}[1]{\mathfrak{E}^{#1}}
\newcommand{\Cinf}{C^{\infty}}
\newcommand{\CConn}{\mathbb{C}\mathfrak{onn}}

\newcommand{\li}[1][]{\ifthenelse{\equal{#1}{}}{\item}{\item \label{#1}}}
\newenvironment{enmrt}{
  \enumerate[(i)]
  \setlength{\itemsep}{0pt}
}{
  \endenumerate
}

\newcommand{\pants}[1]{
\coordinate (center) at (#1);
\coordinate (ld) at ($(center) + (-2, -3)$);
\coordinate (lu) at ($(ld) + (0, 6)$);
\draw (ld) -- ($(ld) + (0, 2)$);
\draw (lu) -- ($(lu) + (0, -2)$);
\draw ($(ld) + (4, 2)$) -- ($(ld) + (4, 4)$);
\draw (ld)
   .. controls ($(ld) + (2, 0)$) and ($(ld) + (2, 2)$)
   .. ($(ld) + (4, 2)$);
\draw (lu)
   .. controls ($(lu) + (2, 0)$) and ($(lu) + (2, -2)$)
   .. ($(lu) + (4, -2)$);
\draw ($(ld) + (0, 4)$) arc(90:-90:1);
}

\newcommand{\copants}[1]{
\coordinate (copantscenter) at (#1);
\begin{scope}[rotate around={180:(copantscenter)}]
\pants{copantscenter}
\end{scope}
}

\newcommand{\idcob}[1]{
\coordinate (center) at (#1);
\coordinate (lu) at ($(center) + (-2, 1)$);
\coordinate (ld) at ($(center) + (-2, -1)$);
\coordinate (ru) at ($(center) + (2, 1)$);
\coordinate (rd) at ($(center) + (2, -1)$);
\draw (lu) -- (ru) -- (rd) -- (ld) -- cycle;
}

\newcommand{\idcobup}[1]{
\coordinate (center) at (#1);
\coordinate (ld) at ($(center) + (-2, -3)$);
\coordinate (ldu) at ($(ld) + (0, 2)$);
\draw (ld) -- (ldu);
\draw ($(ld) + (4, 2)$) -- ($(ld) + (4, 4)$);
\draw (ld)
   .. controls ($(ld) + (2, 0)$) and ($(ld) + (2, 2)$)
   .. ($(ld) + (4, 2)$);
\draw (ldu)
   .. controls ($(ldu) + (2, 0)$) and ($(ldu) + (2, 2)$)
   .. ($(ldu) + (4, 2)$);
}

\newcommand{\cupcob}[1]{
\coordinate (center) at (#1);
\coordinate (ru) at ($(center) + (2, 1)$);
\draw (ru) arc(90:270:1) -- cycle;
}

\newcommand{\capcob}[1]{
\coordinate (center) at (#1);
\coordinate (lu) at ($(center) + (-2, 1)$);
\draw (lu) arc(90:-90:1) -- cycle;
}

\newcommand{\midarrow}[3][0.5]{
\coordinate (s) at (#2);
\coordinate (t) at (#3);
\coordinate (midarrowmid) at ($(s) + #1*(t) - #1*(s)$);
\draw[->] (#2)          -- (midarrowmid);
\draw     (midarrowmid) -- (#3);
}

\newcommand{\arrowIn}{\tikz \draw[->] (-1pt, 0) -- (1pt, 0);}
\newcommand{\midarrowc}[5][0.5]{
\coordinate (s) at (#2);
\coordinate (c1) at (#3);
\coordinate (c2) at (#4);
\coordinate (t) at (#5);
\draw (s) .. controls (c1) and (c2) .. (t)
  node[sloped, pos=#1, allow upside down]{\arrowIn};
}

\newcommand{\vertinnersep}{1pt}

\newcommand{\colvert}[3]{
\node[circle, fill, inner sep=\vertinnersep, #1] at (#2) (#3) {};
}

\newcommand{\lblvert}[3]{
\node at (#1) (#2) {#3};
}

\newcommand{\ppants}[1]{
\coordinate (center) at (#1);
\coordinate (ld) at ($(center) + (-2, -3)$);
\coordinate (lu) at ($(ld) + (0, 6)$);
\draw (ld) -- ($(ld) + (0, 2)$);
\draw (lu) -- ($(lu) + (0, -2)$);
\draw ($(ld) + (4, 2)$) -- ($(ld) + (4, 4)$);
\draw (ld)
   .. controls ($(ld) + (2, 0)$) and ($(ld) + (2, 2)$)
   .. ($(ld) + (4, 2)$);
\draw (lu)
   .. controls ($(lu) + (2, 0)$) and ($(lu) + (2, -2)$)
   .. ($(lu) + (4, -2)$);
\draw ($(ld) + (0, 4)$) arc(90:-90:1);
\colvert{black}{$(lu) + (0, -1)$}{s1}
\colvert{black}{$(ld) + (0, 1)$}{s2}
\colvert{black}{$(lu) + (4, -3)$}{t}
\midarrowc{s2}{$(s2) + (2, 0)$}{$(s2) + (2, 2)$}{t}
\midarrowc{s1}{$(s1) + (2, 0)$}{$(s1) + (2, -2)$}{t}
}

\newcommand{\pcopants}[1]{
\coordinate (copantscenter) at (#1);
\begin{scope}[rotate around={180:(copantscenter)}]
\ppants{copantscenter}
\end{scope}
}

\newcommand{\pidcobup}[1]{
\coordinate (center) at (#1);
\coordinate (ld) at ($(center) + (-2, -3)$);
\coordinate (ldu) at ($(ld) + (0, 2)$);
\draw (ld) -- (ldu);
\draw ($(ld) + (4, 2)$) -- ($(ld) + (4, 4)$);
\draw (ld)
   .. controls ($(ld) + (2, 0)$) and ($(ld) + (2, 2)$)
   .. ($(ld) + (4, 2)$);
\draw (ldu)
   .. controls ($(ldu) + (2, 0)$) and ($(ldu) + (2, 2)$)
   .. ($(ldu) + (4, 2)$);
\colvert{black}{$(ld) + (0, 1)$}{s}
\colvert{black}{$(ld) + (4, 3)$}{t}
\midarrowc{s}{$(s) + (2, 0)$}{$(s) + (2, 2)$}{t}
}

\newtheorem{thm}{Theorem}
\numberwithin{thm}{section}
\newtheorem{lem}[thm]{Lemma}
\newtheorem{cor}[thm]{Corollary}

\theoremstyle{definition}
\newtheorem{defn}[thm]{Definition}
\newtheorem{rmk}[thm]{Remark}
\newtheorem{alg}[thm]{Algorithm}
\newtheorem{exm}[thm]{Example}
\newtheorem{qstn}[thm]{Question}

\begin{document}

\title[TQFTs and quantum computing]{TQFTs and quantum computing}
\author{Mahmud Azam and Steven Rayan}
\date{\today}

\address{Centre for Quantum Topology and Its Applications (quanTA) and Department of Mathematics \& Statistics, University of Saskatchewan, McLean Hall, 106 Wiggins Road, Saskatoon, SK, CANADA S7N 5E6}
\email{mahmud.azam@usask.ca, rayan@math.usask.ca}

\begin{abstract}
Quantum computing is captured in the formalism of the monoidal subcategory of
$\Vect_{\C}$ generated by $\C^2$ --- in particular, quantum circuits are diagrams
in $\Vect_{\C}$ --- while topological quantum field theories, in the sense of
Atiyah, are diagrams in $\Vect_{\C}$ indexed by cobordisms. We initiate a program
that formalizes this connection. In doing so, we equip cobordisms with
machinery for producing linear maps by parallel transport along curves under a
connection and then assemble these structures into a double category. Finite-dimensional complex vector spaces and linear maps between them are given a
suitable double categorical structure which we call $\FFVect_{\C}$. We
realize quantum circuits as images of cobordisms under monoidal double functors
from these modified cobordisms to $\FFVect_{\C}$, which are computed by taking
parallel transports of vectors and then combining the results in a pattern
encoded in the domain double category.
\end{abstract}

\maketitle

\tableofcontents


\section{Introduction}

Perhaps unsurprisingly, quantum field theories and quantum information enjoy natural points of intersection as two sides of modern quantum theory.  However, the essential purposes and formalisms inherent to these subjects are rather different, and many of these observed intersections are coincidental or speculative in nature.  Here, we capitalize on shared aspects of the categorical frameworks associated to the two theories in order make efforts to close the gaps between them.  To be precise, we consider topological quantum field theories (TQFTs) of thick tangled type.  On the quantum information side, we impose no restrictions.

The monoidal category $2\Thick$ of thick tangles is the monoidal category freely
generated by the composition of the following morphisms \cite{NonCommTQFT}:
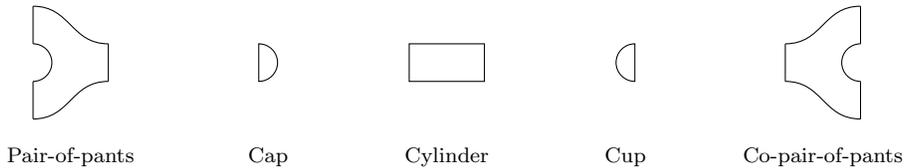
\[\begin{tikzpicture}[scale=0.25]
\pants{0, 0}
\node at (0, -5) {\footnotesize Pair-of-pants};
\capcob{12, 0}
\node at (10.5, -5) {\footnotesize Cap};
\idcob{20, 0}
\node at (20, -5) {\footnotesize Cylinder};
\cupcob{28, 0}
\node at (29.5, -5) {\footnotesize Cup};
\copants{40, 0}
\node at (40, -5) {\footnotesize Co-pair-of-pants};
\end{tikzpicture}\]
Consider a planar open string topological quantum field theory
$F : 2\Thick \to \Vect_{\C}$ as defined in \cite{NonCommTQFT}. $F$ is
determined by the images of the above generating structures under $F$. One
possible to way to connect quantum information to topological quantum
field theories is to assume that the image of the interval $I$ under $F$
is some finite dimensional $C^*$--algebra of operators on some Hilbert space of
states. Coecke, Heunen, and Kissinger \cite{CatQChan} have shown that every such
algebra is, in fact, a dagger Frobenius algebra, with some additional structure,
so that this assumption on $F$ is valid. In fact, every finite dimensional
$C^*$--algebra arises as the image of the interval under such a planar open
string field theory.

It is well-known that finite dimensional $C^*$--algebras, up to
$*$--isomorphism, can be realized as finite direct sums of square matrix algebras $\bigoplus_i
\M_{n_i}$ where $\M_n$ is the set of $n \times n$ matrices with complex entries
equipped with the usual multiplication. For simplicity, we first assume that
$F(I)$ is the $C^*$--algebra $\M_n$. Then, we can consider quantum gates to be
elements of $\M_n$ and circuits to be composites (products) of these
elements. While $\M_n$ has all gates necessary for quantum computing for a finite quantum register, there is a major constraint on the elements of $\M_n$
that are in the image of $F$, as we describe below.

Elements $a \in \M_n$ can be seen as maps $\C \to \M_n : z \mapsto za$. The
elements that are accounted for by $F$ are images of thick tangles
$\varnothing \to I$. However, these ``element'' thick tangles are determined by
their genus, since in $2\Thick$ we identify morphisms up to diffeomorphism, so
that all element thick tangles must decompose as follows:
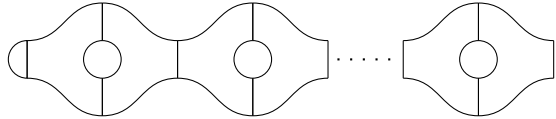
\[\begin{tikzpicture}[scale=0.25]
\cupcob{0, 0}
\copants{4, 0}
\pants{8, 0}
\copants{12, 0}
\pants{16, 0}
\draw[loosely dotted, thick] (18.5, 0) -- (21.5, 0);
\copants{24, 0}
\pants{28, 0}
\end{tikzpicture}\]
where we take the domain of the thick tangle to be on the left and the codomain,
on the right. Call this thick tangle $R$. Let $m : \M_n \tensor \M_n \to \M_n$
be the multiplication of $\M_n$.  Then the multiplicative unit is
$e : \C \to \M_n : z \mapsto zI_n$ and, by the definition of dagger Frobenius
algebra, the comultiplication of $\M_n$ is $m^{\dagger}$. Hence, under $F$, the
element thick tangle above yields a map $(mm^{\dagger})^k e$, where
the superscript $k$ denotes a $k$--fold composite for some non-negative integer
$k$. Now, $mm^{\dagger}$ is the map $a \mapsto na$
\cite[p. 5189 --- 5190]{CatQChan} so that
\[
  F(R)(z) = z \cdot F(R)(1) = z \cdot (mm^{\dagger})^ke(1) = zn^{k}I_n
\]

This shows that the only elements, in the usual sense of elements of $\M_n$,
that are accessible through $F$ are multiples of the identity matrix by powers
of $n$. It is then easy to see that for direct sums such as
$\bigoplus_{i = 1}^{q} \M_{n_i}$, the only accessible elements are
\[
  (n_1^k I_{n_1}, \dots, n_q^k I_{n_q})
\]
so that up to $*$--isomorphism, there are some major constraints on the quantum
gates that are accessible through planar open string field theories.
We note that this situation is brought about by the identification of thick
tangles up to diffeomorphisms. This motivates us to look for methods in a
setting where we drop this identification --- higher categories of thick tangles
(for instance, $\PTT$ as defined in \cite{NonSemiSimp}) (or cobordisms)
where gluing is associative and unital up to higher isomorphisms. We will
see that it suffices to consider double categories for obtaining a reasonable
method for formulating quantum information in terms of topological field
theories.

Recall that a monoidal double category consists of a $1$--category of objects
and a $1$--category of morphisms with source, target and unit functors, a notion
of horizontal composition of morphisms, and a monoidal structure on the object
and morphism categories. In addition, horizontal composition and monoidal
products need to be associative and unital up to isomorphism with several
coherence and compatibility properties \cite{SymMonBicat}. In this work,
however, by ``monoidal double category'' we will mean only the data of such a
structure. Nevertheless, wherever possible, we have commented on how the data of
our constructions inherit most of the necessary properties from the usual
categories of sets, manifolds, vector spaces, and so on. Thus, we will construct
several monoidal double categories in this work but they should be seen as
monoidal double categories in a somewhat relaxed sense --- they consist of all
the required data but satisfy the required axioms with a few possible
exceptions. We will treat monoidal double functors in the same loose sense.

\subsubsection*{Acknowledgements.}

The second-named author is partially supported by a Natural Sciences and Engineering Research Council of Canada (NSERC) Discovery Grant, a Canadian Tri-Agency New Frontiers in Research (Exploration Stream) Grant, and a Pacific Institute for the Mathematical Sciences (PIMS) Collaborative Research Group (CRG) Award. The first-named author was supported by an NSERC Undergraduate Student Research (USRA) Award and the funding of the second-named author. The code for figures \ref{fig:geomegraph} and \ref{fig:cobaddexm} were generated with the help of the Mathcha editor (mathcha.io).


\section{A Double Categorical Approach}

It is well known that taking $d$--dimensional manifolds without boundary as
objects, diffeomorphisms between them as vertical $1$--morphisms, $(d +
1)$--dimensional cobordisms between them as horizontal $1$--morphisms and
boundary preserving diffeomorphisms between cobordisms as $2$--morphisms yields
a fibrant monoidal double category $\Cob_{d + 1}$ under the disjoint union of
manifolds \cite{SymMonBicat}. Shulman gives a trifunctor $\mathcal{H}$ from the
tricategory of fibrant double categories to the tricategory of bicategories that
takes $\Cob_{d + 1}$ to a monoidal bicategory --- in fact, Shulman proves that
$\mathcal{H}$ takes any fibrant monoidal double category to a bicategory.

On the other hand, the monoidal category of thick tangles $2\Thick$ as defined
in \cite{NonCommTQFT} is a decategorification of a monoidal bicategory of thick
tangles $\PTT$ defined in \cite{NonSemiSimp}. Taking inspiration from this
situation, we assume that there is a fibrant monoidal double category
$\DThick$ which, under $\mathcal{H}$, yields $\PTT$. The structures defining
$\DThick$ are the ones analogous to $\Cob_{2}$:
\begin{itemize}
\setlength{\itemsep}{0pt}
\item objects are diffeomorphism classes of disjoint unions of the interval
$I = [0, 1]$\footnote{The justification for equating disjoint unions of the
interval up to diffeomorphism is that the monoidal product on the objects of
$\PTT$ (and $2\Thick$) is strict given that the objects are taken to be the
integers $n$ as opposed to $n$--fold disjoint unions $I^{\amalg n}$ with
different bracketings.}
\item vertical $1$--morphisms are only the identity morphisms
\item horizontal $1$--morphisms $I^{\amalg n} \to I^{\amalg m}$ are surfaces
with boundary $I^{\amalg n} \amalg I^{\amalg m}$ along with an embedding $d$
into $\R \times I$ such that $d^{-1}(\R \times \set{0}) = I^{\amalg n}$ and
$d^{-1}(\R \times \set{1}) = I^{\amalg m}$\footnote{There are finer details here
which will be unimportant for our purposes.}
\item $2$--morphisms are diffeomorphisms between cobordisms
(horizontal $1$--morphisms) that preserve the boundary
\end{itemize}
We then attempt to use this notion to concretely define the data of a monoidal
double functor from $\DThick$ to a suitable monoidal double category of complex
vector spaces that yields enough unitary linear transformations in the image to
facilitate quantum computing.

We define the object function $F_0$ of such a functor by assigning to each
disjoint union $X = I^{\amalg n}$ the $n$--th tensor power $F_0(X) = A^{\tensor
n}$ of some fixed algebra $A$ --- we may be specific enough to pick a consistent
bracketing pattern for $A^{\tensor n}$. It is easy to see that this assignment
is well-defined. Since the vertical $1$--morphisms are only identities, the
object category of $\DThick$ is discrete and, hence, the vertical $1$--morphism
function is the unique, obvious one: $F_0(\id_X) = \id_{F_0(X)}$.

Next, we consider the horizontal $1$--morphisms or the cobordisms
$Z : I^{\amalg n} \to I^{\amalg m}$, which are determined up to diffeomorphism
by their genus. For positive $m$ and $n$, we first consider some ``canonical''
genus $k$ cobordism $Z : I \to I$ where the holes are circles with centers along
a straight line from one boundary interval to another. This cobordism then
decomposes into a $k$--fold composition of $M * W : I \to I$ with itself, where
$M$ is the ``canonical'' pair-of-pants and $W$ is the ``canonical''
co-pair-of-pants. An example is shown below:

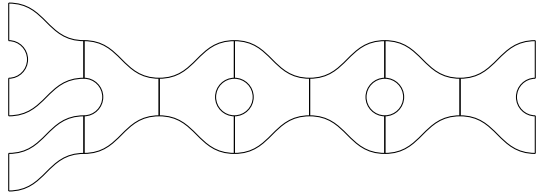
\[\begin{tikzpicture}[scale=0.25]
\idcobup{1, 0}
\pants{1, 4}
\pants{5, 2}
\copants{9, 2}
\pants{13, 2}
\copants{17, 2}
\pants{21, 2}
\copants{25, 2}
\end{tikzpicture}\]
In associating a linear map ``functorially'' to $Z$, it
suffices to associate linear maps to $M$ --- we can then associate a linear map
to $W$ by duality and get a linear map for $Z$ by composition.

For associating a linear map to $M$, we take a complex bundle
$\pi : E \to M$ with fibre $A$ along with a connection $\nabla$. Then, we choose
two paths: one from the mid-point of each in-boundary interval to the mid-point
of the out-boudnary interval. By parallel transport along each curve, we get two
linear maps $l_1, l_2 : A \to A$. We then have a linear map
$l : A \tensor A \to A$ given by $l(x \tensor y) = l_1(x)l_2(y)$ where the
product in the right is the algebra product in $A$. Then, the linear map
associated to $W$ is the conjugate transpose $l^{\dagger}$.
We then obtain a linear map associated to $Z$: the $k$--fold composite
$(ll^{\dagger})^k$. We set $F_1(Z) := (ll^{\dagger})^k$.

Now, consider a cobordism $Y : I^{\amalg n} \to I^{\amalg m}$ of genus $k$ such
that $Y = RZQ$, where $Q$ is a cobordism $I^{\amalg n} \to I$ formed in some
``canonical'' way by a gluing of pairs of pants and cylinders (``identity''
cobordisms) and $R$ is a cobordism $I \to I^{\amalg m}$ formed again in a
``canonical'' way from co-pairs of pants and cylinders. Of course, we associate
a linear map to the cylinder by another parallel transport along a curve between
the mid-points of its two boundaries. Then, again by composition, we get a
linear map $F_1(Y) : A^{\tensor n} \to A^{\tensor m}$.

This gives an assignment $F_1(Y)$ for a representative $Y$ of each
diffeomorphism class of cobordisms in $\DThick$. For the assignment of a linear
map to every cobordism in $\DThick$, we take the following approach. Let $Y'$ be
an arbitrary cobordism in the class of some $Y$ for which $F_1(Y)$ has been
defined as above. Then we pick a boundary preserving diffeomorphism $f \in
\Aut_{\Man}(Y)$ such that $f(Y) = Y'$\footnote{Of course, we take the underlying
topological space for each manifold in a diffeomorphism class to be the same.}.
Let $\set{\gamma_i}$ be the family of curves which along which parallel
transport gave us the linear maps $F_1(Y)$. Then, $\set{f \circ \gamma_i}$ is a
family of curves in $Y'$ which yield linear maps by parallel transport under a
connection $f\nabla$, to be made precise later. Combining these using the same
``pattern'' or ``expression'' of algebra multiplications, tensor products and
compositions as we had for $Y$, we can obtain a linear map $F_1(Y')$. Continuing
the example, we have:

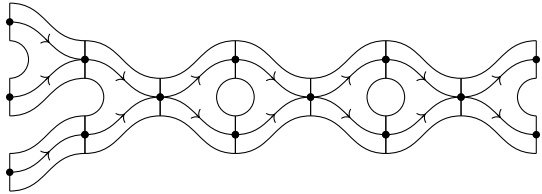
\[\begin{tikzpicture}[scale=0.25]
\pidcobup{1, 0}
\ppants{1, 4}
\ppants{5, 2}
\pcopants{9, 2}
\ppants{13, 2}
\pcopants{17, 2}
\ppants{21, 2}
\pcopants{25, 2}
\end{tikzpicture}\]

We now turn our attention to the case when $m$ or $n$ is zero. We can treat the
case $n = 0$ and obtain the other case by duality. Let
$Y : \varnothing \to I^{\amalg m}$ with genus $k$. Then $Y$ decomposes as
$R * Z * Z'$ where $Z$ and $R$ are as before and $Z'$ is a genus zero cobordism
$\varnothing \to I$. We call cobordisms $\varnothing \to I$ elements and we call
elements of genus zero, atomic elements because they will not decompose into any
simpler structures. $Z'$ has a boundary preserving diffeomorphism
$f : Z'' \to Z'$ for some atomic element $Z''$ deemed ``canonical''. We
associate a linear map $a : \C \to A$ to $Z''$ as follows.  Take a loop $\gamma$
in $Z''$ on the mid-point of its only boundary interval and obtain an element
$a \in A$, or equivalently a linear map $a : \C \to A$ by parallel transport of
some fixed element $a_0 \in A$. We set $F_1(Z'') = a$ and get $F_1(Z')$ by
parallel transport of $a_0$ along $f\gamma$. $F_1(Y)$ is then obtained by
composition.

This completes the definition of an object function $F_1$ for the morphism
category of $\DThick$. Now, we turn our attention to $2$--morphisms --- boundary
preserving diffeomorphisms between cobordisms. Let $f : Y \to Y'$ be one such
diffeomorphism. We must assign to $f$ some object that functions as roughly a
``morphism of morphisms of vector spaces''. One natural choice is homotopy
classes of paths between linear functions in spaces of linear functions under
some suitable norm. In this case, we can take the following approach. There
exists a path $\psi$ in the diffeomorphism group $\Diff(Y)$ of $Y$ from $\id_Y$
to $f$ such that we can ``move'' the parallel transport machinery ``along''
$\psi$ --- that is, taking connections $\psi(t)\nabla$ and paths
$\set{\psi(t)\gamma_i}$, for $t \in [0, 1]$ --- to get linear maps for each
$\psi(t)(Y)$, which constitute a path in the space $\Hom(\dom F_1(Y), \codom
F_1(Y))$. Note that, implicit in this is the assumption that
$\psi(1)\nabla = f\nabla$ and $\set{\psi(1)\gamma_i} = \set{f\gamma_i}$ yield
the linear map $F_1(Y')$ by parallel transport, which need not hold in general.

The issue is that the $2$--morphisms are not guaranteed to ``preserve parallel
transport''. In order for this picture to make sense, we need to consider
$2$--morphisms that do this. Another way of viewing the scenario is that we are
dealing with cobordisms equipped with parallel transport machinery and these
should be the objects of our morphism category.


\section{Connections on Cobordisms}

For a vector bundle $\pi_E : E \to M$, we write $\Gamma(E)$ to denote the set of
smooth sections of the bundle --- we will not be using the sheaf structure unless
necessary. For $\K = \R$ when $E$ is a smooth bundle with a real vector space as
fibres, or $\K = \C$ when $E$ is a complex bundle, we recall that a connection
on $E$ is a $\K$--linear map
\[
  \nabla : \Gamma(E) \to \Gamma(E \tensor T^*M)
\]
such that for all $s \in \Gamma(E)$ and $r \in \Cinf(M, \K)$, the following
Leibniz property is satisfied:
\[
  \nabla(r \cdot s) = r \cdot \nabla(s) + s \tensor dr
\]

The right summand requires some clarification.
First, if $r : M \to \R$ is a smooth function, then the derivative of $r$ is a
map $dr : TM \to T\R$ such that $(r, dr)$ is a bundle morphism making the
following diagram in the category of manifolds commute:
\[\begin{tikzcd}
TM \ar[d, "\pi_{TM}" left] \ar[r, "dr" above] & T\R \ar[d, "\pi_{T\R}" right]\\
M \ar[r, "r" below] & \R
\end{tikzcd}\]
By the definition of bundle morphism, $dr$ is linear on fibres so that for each
$x \in M$, $dr$ restricts to a linear map $T_xM \to \R$ and these maps vary
smoothly with $x$, so that $dr$ is a section of $T^*M$. Hence, by an abuse of
notation, we can view $dr$ as the following map:
\[
  dr : M \to T^*M : x \mapsto dr|_{T_xM}
\]

Next, we are treating $s \tensor dr$ as a map $M \to E \tensor T^*M$ which is
not strictly a tensor product of two maps because its domain is not a tensor
product. By $s \tensor dr$, what we really mean is the map:
\begin{equation}\label{eqn:conn_tensor}
  x \mapsto s(x) \tensor dr|_{T_xM}
\end{equation}

We will now see that isomorphisms of vector bundles have an action on the
connections on these bundles, leading to a notion of morphism for connections
--- a first step in developing a double category of
``parallel transport machinery''.

\subsection{Gauge Transformations}

Suppose we have $\pi_E : E \to M$ and $\nabla$ as before as well as another
bundle $\pi_{E'} : E' \to M$ with a bundle isomorphism $f = (u, v) : E \to E'$
--- a pair of maps $u : M \to M'$ and $v : E \to E'$ with $v$ linear on each
fibre of $E$, making the following diagram commute:
\[\begin{tikzcd}
E \ar[d, "\pi_E" left] \ar[r, "v" above] & E' \ar[d, "\pi_{E'}" right]\\
M \ar[r, "u" below] & M'
\end{tikzcd}\]

Let $s : M' \to E' \in \Gamma(E')$. Then we have a section
$\wh{f}(s) \in \Gamma(E)$ defined by
\[
  \wh{f}(s) = v^{-1} \circ s \circ u
\]
Noting that $u : M \to M'$ is a diffeomorphism, it is easy to verify $du$
is also a bundle isomorphism, from the definition of differentials. Furthermore,
there is a bundle isomorphism $d^*u : T^*M \to T^*M'$ corresponding to $du$,
defined, for each $g : T_xM \to \R \in T^*_xM$, by the composite
\[
  (d^*u)(g) := T_{u(x)}M' \to[(du)^{-1}] T_xM \to[g] \R \in T^*_{u(x)}M'
\]

Denoting $\tilde{f}(x \tensor g) := v(x) \tensor (d^*u)(g)$, we then define:
\[\begin{array}{ccccc}
f \diamond \nabla
&:& \Gamma(E') &\to    & \Gamma(E' \tensor T^*M') \\
&:& s &\mapsto& \tilde{f} \circ \nabla(\wh{f}(s)) \circ u^{-1} \\
&&& = &
  \tilde{f} \circ \nabla(v^{-1} \circ s \circ u) \circ u^{-1}
\end{array}\]
We wish to show that $f \diamond \nabla$ is a connection.
Let $c \in \K$. Then, for a section $s \in \Gamma(E')$, have:
\begin{align*}
   & (f \diamond \nabla)(c \cdot s) \\
  =& \tilde{f}\nabla(v^{-1}(c \cdot s)u)u^{-1} \\
  =& \tilde{f}\nabla(c \cdot v^{-1}su)u{-1}
      && \text{fibre-wise linearity of } v^{-1} \\
  =& \tilde{f} (c \cdot  \nabla(v^{-1}su)u{-1})
      && \text{linearity of } \nabla \\
  =& c \cdot \tilde{f}\nabla(v^{-1}su)u^{-1}
      && \text{fibre-wise linearity of } \tilde{f} \\
  =& c \cdot (f \diamond \nabla)(s)
\end{align*}
We also observe that, for sections $s_1, s_2 \in \Gamma(E')$, we have
\begin{align*}
   & (f \diamond \nabla)(s_1 + s_2)\\
  =& \tilde{f}\nabla(v^{-1}(s_1 + s_2)u)u^{-1} \\
  =& \tilde{f}\nabla(v^{-1}(s_1u + s_2u))u^{-1}
    && \text{definition of pointwise addition} \\
  =& \tilde{f}\nabla(v^{-1}s_1u + v^{-1}s_2u)u^{-1}
    && \text{fibre-wise linearity of } v^{-1} \\
  =& \tilde{f}(\nabla(v^{-1}s_1u)
    + \nabla(v^{-1}s_2u))u^{-1}
    && \text{linearity of } \nabla \\
  =& \tilde{f}(\nabla(v^{-1}s_1u)u^{-1}
    + \nabla(v^{-1}s_2u)u^{-1})
    && \text{definition of pointwise addition} \\
  =& \tilde{f}\nabla(v^{-1}s_1u)u^{-1}
    + \tilde{f}\nabla(v^{-1}s_2u)u^{-1}
    && \text{fibre-wise linearity of } \tilde{f} \\
  =& (f \diamond \nabla)(s_1) + (f \diamond \nabla)(s_2)
\end{align*}
Thus, $f \diamond \nabla$ is $\K$--linear. Now, for $r \in \Cinf(M', \K)$, we
again observe:
\begin{align*}
   & (f \diamond \nabla)(r \cdot s) \\
  =& \tilde{f}\nabla(v^{-1}(r \cdot s)u)u^{-1} \\
  =& \tilde{f}\nabla(v^{-1}(ru \cdot su))u^{-1}
    && \text{pointwise multiplication} \\
  =& \tilde{f}\nabla(ru \cdot v^{-1}su)u^{-1}
    && \text{fibre-wise linearity of } v^{-1} \\
  =& \tilde{f}(ru \cdot \nabla(v^{-1}su) + (v^{-1}su) \tensor d(ru))u^{-1}
    && \text{Leibniz property} \\
  =& \tilde{f}(ru \cdot \nabla(v^{-1}su))u^{-1}
      + \tilde{f}((v^{-1}su) \tensor d(ru))u^{-1}
    && \text{distribute over $+$ as before} \\
  =& \tilde{f}(ruu^{-1} \cdot \nabla(v^{-1}su)u^{-1})
      + \tilde{f}((v^{-1}su) \tensor d(ru))u^{-1}
    && \text{pointwise multiplication} \\
  =& r \cdot \tilde{f}\nabla(v^{-1}su)u^{-1}
      + \tilde{f}((v^{-1}su) \tensor d(ru))u^{-1}
    && \text{fibre-wise linearity of } \tilde{f} \\
  =& r \cdot (f \diamond \nabla)(s)
      + \tilde{f}((v^{-1}su) \tensor d(ru))u^{-1} \\
  =& r \cdot (f \diamond \nabla)(s)
      + (v \tensor (d^*u))(v^{-1}su \tensor d(ru))u^{-1} \\
  =& r \cdot (f \diamond \nabla)(s)
      + (v \tensor (d^*u))(v^{-1}suu^{-1} \tensor d(ru)u^{-1})
    && \text{by definition \eqref{eqn:conn_tensor}} \\
  =& r \cdot (f \diamond \nabla)(s)
      + (v \tensor (d^*u))(v^{-1}s \tensor d(ru)u^{-1}) \\
  =& r \cdot (f \diamond \nabla)(s)
      + s \tensor (d^*u)d(ru)u^{-1}
    && \text{$\tensor$ for sections}
\end{align*}

It now suffices to show that $(d^*u)d(ru)u^{-1} = dr$. Pointwise, we have:
\[
  ((d^*u) \circ d(ru) \circ u^{-1})(x) = d(ru)|_{T_{u^{-1}(x)}M} \circ (du)^{-1}
\]
We then observe a useful property of the derivative
operator $d-$. If $a : L \to M$ and $b : M \to N$ are smooth maps, then it is
easy to verify, from the definition of the differential, that
\[
  d(b \circ a) = db \circ da
\]
Then, we observe that
\[
  d(ru)|_{T_{u^{-1}(x)}M} \circ (du)^{-1}
  = dr|_{T_xM'} \circ du|_{T_{u^{-1}(x)}M} \circ (du)^{-1}
  = dr|_{T_xM'}
  = dr(x)
\]
as required. This completes the proof of the following theorem.

\begin{thm}
Let $\pi : E \to M$, $\pi' : E' \to M'$ be bundles with a bundle isomorphism
$f = (u, v) : E \to E'$. If $\nabla$ is a connection on $\pi$, then
$f \diamond \nabla$ is a connection on $\pi'$.
\end{thm}

The following theorem shows that the operation $- \diamond -$ commutes with
composition of diffeomorphisms. This result ultimately provides a notion of
morphism of connections from bundle isomorphisms.
\begin{thm}
Let $\pi_i : E_i \to M_i$ be bundles
for $i \in \set{1, 2, 3}$ along
bundle isomorphisms
$f_{j, j + 1} = (u_{j, j+ 1}, v_{j, j + 1}) : \pi_{j} \to \pi_{j + 1}$
for $j \in \set{1, 2}$. Then, if $\nabla$ is a connection on $\pi_1$, we have:
\[
  (f_{2, 3}f_{1, 2}) \diamond \nabla
  = f_{2, 3} \diamond (f_{1, 2} \diamond \nabla)
\]
\end{thm}
\begin{proof}
By expanding expressions, we obtain:
\begin{align*}
((f_{2, 3}f_{1, 2}) \diamond \nabla)(s)
=& (v_{2,3}v_{1,2} \tensor d^*(u_{2,3}u_{1,2}))
   \nabla((v_{2, 3}v_{1, 2})^{-1}s(u_{2, 3}u_{1,3}))
   (u_{2,3}u_{1,2})^{-1}
\end{align*}
We observe that for any suitable maps $p, q, w$, we have:
\[
  d^*(pq)(w) = w \circ (d(pq))^{-1}
  = w \circ (dq)^{-1} \circ (dp)^{-1}
  = d^*p(w \circ (dq)^{-1})
  = d^*p(d^*q(w))
\]
so that the first expression becomes:
\begin{align*}
((f_{2, 3}f_{1, 2}) \diamond \nabla)(s)
=& (v_{2,3}v_{1,2} \tensor d^*u_{2,3}d^*u_{1,2})
   \nabla((v_{2, 3}v_{1, 2})^{-1}s(u_{2, 3}u_{1,3}))
   (u_{2,3}u_{1,2})^{-1} \\
=& (v_{2,3} \tensor d^*u_{2,3})\sbr{(v_{1,2} \tensor d^*u_{1,2})
   \nabla(v_{1, 2}^{-1}(v_{2, 3}^{-1}su_{2, 3})u_{1,2})
   u_{1,2}^{-1}}u_{2,3}^{-1} \\
=& (v_{2,3} \tensor d^*u_{2,3})
   (f_{1,2} \diamond \nabla)(v_{2, 3}^{-1}su_{2, 3})
   u_{2,3}^{-1} \\
=& (f_{2,3} \diamond (f_{1,2} \diamond \nabla))(s)
\end{align*}
as required.
\end{proof}

From this point, we will write $f\nabla$ as opposed to $f \diamond \nabla$,
as long as the action is clear from context. Now, Let $\pi_i : E_i \to M_i$ be
bundles equipped with connections
$\nabla_i$ for $i \in \set{1, 2, 3, 4}$. Let
$f_{i, i + 1} = (u_{i, i + 1}, v_{i, i + 1}) : \pi_{i} \to \pi_{i + 1}$
be bundle isomorphisms for $i \in \set{1, 2, 3}$. Then, the composite
$f_{1, 3} := f_{2, 3}f_{1, 2} = (u_{2, 3}u_{1, 2}, v_{2, 3}v_{1, 2})
=: (u_{1, 3}, v_{1, 3})$ is clearly a bundle isomorphism. We similarly define
composites $f_{i, j}$ for each $i < j \in \set{1, 2, 3, 4}$. Now, suppose
$\nabla_{i + 1} = f_{i, i + 1}\nabla_i$ for each $i \in \set{1, 2, 3}$. Then, we
immediately have, from the previous theorem:
\[
  f_{i, k}\nabla_i = f_{j, k}f_{i, j}\nabla_i = f_{j, k}\nabla_j = \nabla_k
\]
for each $i < j < k$ in $\set{1, 2, 3, 4}$. In particular,
\[
  \nabla_4
    = f_{3, 4}(f_{2, 3}f_{1, 2}\nabla_1)
    = (f_{3,4}f_{2,3}f_{1,2}) \nabla_1
    = (f_{3, 4}f_{2, 3})f_{1, 2}\nabla_1
\]

We then observe the action of identity bundle morphisms. The identity bundle
morphism on $\pi_1$ is the pair $\id_{\pi_1} = (\id_{E_1}, \id_{M_1})$. Then,
\[
  \id_{\pi_1}\nabla_1(s)
    = (\id \tensor \id) \nabla_1(\id \circ s \circ \id) \id
    = \nabla_1(s)
\]
so that $\id_{\pi_1}\nabla = \nabla$. We finally observe that
$f(f^{-1}\nabla) = (ff^{-1})\nabla = \id\nabla = \nabla$ for any connection
$\nabla$ and any compatible bundle morphism $f$.

These observations motivate the following definition.
\begin{defn}[Gauge Transformation]
Let $\pi_1$ and $\pi_2$ be bundles equipped with connections $\nabla_1$ and
$\nabla_2$ respectively. Then, a bundle isomorphism
$f = (u, v) : \pi_1 \to \pi_2$ satisfying $f\nabla_1 = \nabla_2$ is called an
isomorphism, or simply morphism, of connections or a gauge transformation.
\end{defn}

From the work above, we have established the following results.
\begin{thm}
There exists a groupoid whose objects are connections and
whose morphisms are gauge transformations or isomorphisms of connections.
\end{thm}
\begin{defn}[Category of Connections]
We will call the category of the above theorem the category or groupoid of
connections. We will denote this category $\Conn$.
\end{defn}

We will see that connections on bundles on compact manifolds with boundary, when
specialized slightly, assemble into a double categorical structure compatible
with that of the cobordism double category formed by the underlying manifolds
over which we take the bundles. Before we proceed to this result, we will
require a notion of cobordism double category for bundles, which we develop
next.

\subsection{Bundle Cobordisms}\label{subsec:bund_cob}

Consider smooth bundles $\pi_1 : E_1 \to M_1$ and $\pi_2 : E_2 \to M_2$. We will
consider the coproduct or disjoint union of these bundles in the category of
manifolds. There exists a smooth map
$\pi_1 \amalg \pi_2 : E_1 \amalg E_2 \to M_1 \amalg M_2$ which we will give the
structure of a vector bundle as follows. For this, we additionally assume that
the fibres of $E_1$ and $E_2$ are the same vector space. Let
$U = U_1 \amalg U_2, V = V_1 \amalg V_2$ be open sets in $M_1 \amalg M_2$ with
$U_i, V_i \subset M_i$ open for $i \in \set{1, 2}$, and consider
$(U_1 \amalg U_2) \cap (V_1 \amalg V_2) = (U_1 \cap V_1) \amalg (U_2 \cap V_2)$.
We have a transition function $G_{U_1, V_1}$ on $U_1 \cap V_1$ from the bundle
$\pi_1$ and one $H_{U_2, V_2}$ on $U_2 \cap V_2$ from $\pi_2$. We define a
function $(G \amalg H)_{U, V} : U \cap V \to \GL_n(\C)$ piecewise, as
follows:
\[
  (G \amalg H)_{U, V}(x) := \begin{cases}
    G_{U_1, V_1}(x), & x \in U_1 \cap V_1 \subset M_1 \\
    H_{U_2, V_2}(x), & x \in U_2 \cap V_2 \subset M_2
  \end{cases}
\]
which is smooth since it is a disjoint union of smooth functions. Therefore,
\[
  G \amalg H := \set[(G \amalg H)_{U, V}]
                    {U, V \subset M_1 \amalg M_2 \text{ are open}}
\]
is a vector bundle structure on $\pi_1 \amalg \pi_2$. A section of
$E_1 \amalg E_2$ is a smooth map
\[
  s : M_1 \amalg M_2 \to E_1 \amalg E_2
\]
satisfying $(\pi_1 \amalg \pi_2)s = \id_{M_1 \amalg M_2}$. We note that this
guarantees that the $s$ must be of the form $s_1 \amalg s_2$ where $s_i$ is a
section of $E_i$, $i \in \set{1, 2}$.

Similarly, $TM_1 \amalg TM_2 \to M_1 \amalg M_2$ is a vector bundle when $M_1$
and $M_2$ have the same dimension, and we can take this to be the definition of
the tangent bundle $T(M_1 \amalg M_2)$ on $M_1 \amalg M_2$. Now, let
$\pi_3 : E_3 \to M_3$ be another bundle where all the $E_i$ have the same fibres
and all the $M_i$ are equidimensional.

We can pick a convention for disjoint unions of sets as follows:
\[
  A \amalg B = (A \times \set{0}) \cup (B \times \set{1})
\]
Under this convention,
\[
  E_1 \amalg (E_2 \amalg E_3)
    = \set[(x_1, 0)]{x_1 \in E_1}
      \cup \set[((x_2, 0), 1)]{x_2 \in E_2}
      \cup \set[((x_3, 1), 1)]{x_3 \in E_3}
\]
and
\[
  (E_1 \amalg E_2) \amalg E_3
    = \set[((x_1, 0), 0)]{x_1 \in E_1}
      \cup \set[((x_2, 1), 0)]{x_2 \in E_2}
      \cup \set[(x_3, 1)]{x_3 \in E_3}
\]
We have similar descriptions for the two distinct parenthesizations for
$M_1 \amalg M_2 \amalg M_3$. Now, the map
\[
  \alpha_{E_1, E_2, E_3} : E_1 \amalg (E_2 \amalg E_3)
                           \to (E_1 \amalg E_2) \amalg E_3
\]
defined by
\[
  (x_1, 0) \mapsto ((x_1, 0), 0),
  ((x_2, 0), 1) \mapsto ((x_2, 1), 0),
  ((x_3, 1), 1) \mapsto (x_3, 1)
\]
is easily seen to be bijective and fibre-preserving. Smoothness and naturality
in the subscripts follow from those of associators in $\Man$. We can make a
similar argument for similarly defined unitors $\rho_E$ and $\lambda_E$. We thus
have the following theorem.
\begin{thm}
The subcategory of the category of bundles consisting of bundles with base
spaces of a fixed dimension $d$ and total spaces with isomorphic fibres is
monoidal under the disjoint union of manifolds.
\end{thm}
\begin{defn}[Category of {$(V, d)$--bundles}]
The subcategory of the category of bundles in the above theorem is called the
category of $V$--fibred bundles on $d$--dimensional manifolds or of
$(V, d)$--bundles and is denoted $\Bun^V_d$.
\end{defn}

We will now develop a notion of gluing complex bundles on compact manifolds with
boundary along with connections on these bundles. To accomplish this, we will
first show the following:
\begin{lem}\label{thm:bundle_gluing}
Let $M$ be a smooth compact manifold such that $\partial M$ has
a collar $C_0$ whose connected components are each contractible. For any complex
vector bundle $\pi : E \to M$ with fibre $P$,
there exists a complex bundle $\wh{\pi} : \wh{E} \to M$ which restricts to
the trivial bundle on a collar $C \subset C_0$ of $\partial M$ and to $E$
on $M \setminus C_0$.
\end{lem}
\begin{proof}
Let $U$ and $V$ be any two open sets of $M$ over which $E$ trivializes and
$G_{U, V} : U \cap V \to \Aut(P)$, the assignment of transition functions to
their intersection.
By the smooth collar theorem, there exists a nieghbourhood $C_0$ of $\partial M$
diffeomorphic to the cylinder $\partial M \times I$ on $\partial M$,
with $\partial M$ identified with $\partial M \times \set{1}$. By hypothesis, we
can take $C_0$ to have contractible components such that $E|_{C_0}$ is
isomorphic to the trivial bundle. Therefore, $G_{U, V}|_{C_0}$ is smoothly
homotopic to the constant map $U \cap V \cap C_0 \to \Aut(P) : x \mapsto \id_P$.
Let $H_{U, V} : I \times U \cap V \cap C_0 \to \Aut(P)$ be this homotopy so that
$H_{U, V}(1, -) = G_{U, V}|_{C_0}$ and $H_{U, V}(0, x) = \id_P$ for all
$x \in U \cap V \cap C_0$.

We can then cut $C_0$ into pieces
$C' \cong \partial M \times \sbr{0, \frac{1}{2}}$ and
$C \cong \partial M \times \sbr{\frac{1}{2}, 1}$ that are each diffeomorphic to
$\partial M \times I$. There exists a smooth bump function $f : M \to \R$ such
that $f$ is $1$ on $M \setminus C_0$, decreases to $0$ on $C'$ as we move
towards $\partial M \times \set{\frac{1}{2}}$ and vanishes on $C$:
\[
  f(x) = \begin{cases}
    1 & x \in M \setminus C_0 \\
    \frac{1}{2}(1 - \text{erf}(at + b))
      & x = (x', t) \in C', x' \in \partial M, t \in \sbr{0, \frac{1}{2}} \\
    0 & x \in C
  \end{cases}
\]
where $a$ and $b$ are appropriately chosen constants.

We then take the bundle $\wh{E} \to M$ with the same trivializations as $E$
and transition functions
\[
  K_{U, V}(x) = \begin{cases}
    H_{U, V}(f(x), x) & x \in C_0 \\
    G_{U, V}(x)       & x \in M \setminus C_0
  \end{cases}
\]
We observe that away from the collar $C_0$, the bundle is the same as $E$ and
inside $C$, it is trivial, as required.
\end{proof}

It is straightforward to verify that for any cospan $M \ot[f] X \to[g] N$ and
any finite dimensional vector space $V$ seen as an object in $\Man$, the
following holds:
\[
  V \times (M \amalg_X N) \cong (V \times M) \amalg_{X \times V} (V \times N)
\]
such that the isomorphism is fibre-preserving and linear on fibres. Hence,
trivial bundles always glue at boundaries to give trivial bundles. This
observation yields a gluing operation $- * -$ for the following collection of
complex bundles:
\[
  \set[\wh{E}]{E \text{ is a complex bundle with fibre } V}
\]
since the bundles $\wh{E}$ are trivial near their boundaries. We observe that
gluing fibres at the boundaries is associative up to diffeomorphism by the same
argument for the associativity of gluing manifolds along boundaries. It is also
not hard to verify that the associator diffeomorphisms are fibre-preserving and
linear on the fibres. Furthermore, given a bundle $\wh{E} \to M$ where
$\partial M = W_0 \amalg W_1$, we take the trivial bundles
$W_0 \times I \times V \to W_0 \times I$ and
$W_1 \times I \times V \to W_1 \times I$, and observe that they act as gluing
identities for $\wh{E}$ on either side by a simple reparametrization. This
establishes a notion of cobordism of bundles. That is,
\begin{thm}
Given a double category of cobordisms $\s{C}$ (e.g. $\CCob_d$ or $\DThick$)
and a complex vector space $V$, we have a double category $\BBun^V_{\s{C}}$
consisting of the following data:
\begin{enmrt}
\li Object category: objects are trivial $V$--bundles on the objects of $\s{C}$
and morphisms are bundle isomorphisms
\li Morphism category: objects are complex bundles $\wh{E} \to M$, for $M$ in
the morphism category of $\s{C}$ and complex bundles $E \to M$; morphisms are
bundle isomorphisms
\li Source functor: each bundle $\wh{E} \to M$ is sent to the trivial bundle on
the source of $M$; action on morphisms is by restriction to appropriate boundary
components
\li Target functor: defined analogously as the source functor
\li Unit functors: each bundle $\wh{E} \to M$ is sent to the trivial bundle on
the cylinder on the appropriate boundary components
\li Horizontal composition: gluing corresponding fibres at common boundary
\li Horizontal composition associators: inherited from the category of manifolds
\li Horizontal composition unitors: inherited like the associators
\li Monoidal product: disjoint union
\li Monoidal unit(s): empty bundle(s)
\end{enmrt}
\end{thm}

We also notice that the above constructions apply to smooth (real) bundles as
long as the transition functions at the points in some collar of the boundary
can be connected to the identity function by paths in the automorphism group of
the fibre. This is possible if these transition functions all have positive
determinant. We may guess that the
structure on the category of bundles developed here transfers over to the
category of connections on the bundles involved. We next show that this is
indeed the case.

\subsection{Monoidal Double Category of Connections}

For $i \in \set{1, 2, 3}$ and connections $\nabla_i$ on smooth bundles
$\pi_i : E_i \to M_i$, we define a function
\[
  \nabla_1 \amalg \nabla_2
    : \Gamma(E_1 \amalg E_2)
    \to \Gamma((E_1 \tensor T^*M_1) \amalg (E_2 \tensor T^*M_2))
\]
as follows, for $j \in \set{0, 1}$:
\[\begin{array}{crcl}
       & (\nabla_1 \amalg \nabla_2)(s_1 \amalg s_2)(x, j)
       & = & \nabla_{j + 1}(s_{j + 1})(x) \\
  \iff & (\nabla_1 \amalg \nabla_2)(s_1 \amalg s_2)
       & = & \nabla_1(s_1) \amalg \nabla_2(s_2)
\end{array}\]
It is easy to see that this function satisfies the connection identities
piecewise so that it satisfies these identities on its entire domain. Thus,
$\nabla_1 \amalg \nabla_2$ is a connection. Letting
$f = (u, v) = (\alpha_{M_1, M_2, M_3}, \alpha_{E_1, E_2, E_3})$, we now wish to
verify that
\begin{equation}\label{eqn:conn_assoc}
  f \diamond (\nabla_1 \amalg (\nabla_2 \amalg \nabla_3))
    = (\nabla_1 \amalg \nabla_2) \amalg \nabla_3
\end{equation}
For this, we will need to inspect the expression:
\begin{align*}
   & f \diamond (\nabla_1 \amalg (\nabla_2 \amalg \nabla_2))(
        (s_1 \amalg s_2) \amalg s_3
     ) \\
  =& \tilde{f}(\nabla_1 \amalg (\nabla_2 \amalg \nabla_2))(
      v^{-1}((s_1 \amalg s_2) \amalg s_3)u
     )u^{-1} \\
  =& \tilde{f}(\nabla_1 \amalg (\nabla_2 \amalg \nabla_2))(
      s_1 \amalg (s_2 \amalg s_3)
     )u^{-1} \\
  =& (v \tensor d^*u)(\nabla_1 \amalg (\nabla_2 \amalg \nabla_2))(
      s_1 \amalg (s_2 \amalg s_3)
     )u^{-1} \\
  =& (v \tensor d^*u)(
      \nabla_1(s_1) \amalg (\nabla_2(s_2) \amalg \nabla_3(s_3))
     )u^{-1}
\end{align*}
To reach our goal \eqref{eqn:conn_assoc}, we observe the following basic facts.
\begin{cor}
For tangent bundles $\pi_i : TM_i \to M_i$, $i \in \set{1, 2, 3}$, we have
$d\alpha_{M_1, M_2, M_3} = \alpha_{TM_1, TM_2, TM_3}$.
\end{cor}
\begin{proof}
$d\alpha_{M_1, M_2, M_3}$ is the essentially the identity on fibres as is
$\alpha_{TM_1, TM_2, TM_2}$. To see this, pick explicit elements of
the relevant spaces and expand the definitions.
\end{proof}

\begin{cor}
For tangent bundles $\pi_i$ as in the previous lemma, we have
\[
  (d^*\alpha_{M_1, M_2, M_3})(g) = g \circ \alpha_{TM_1, TM_2, TM_3}^{-1}
    = \br{\alpha^{-1}_{TM_1, TM_2, TM_3}}^*(g)
\]
\end{cor}
The above corollary yields:
\begin{align*}
   & f \diamond (\nabla_1 \amalg (\nabla_2 \amalg \nabla_2))(
        (s_1 \amalg s_2) \amalg s_3
     ) \\
  =& (v \tensor d^*u)(
      \nabla_1(s_1) \amalg (\nabla_2(s_2) \amalg \nabla_3(s_3))
     )u^{-1} \\
  =& \br{\alpha_{M_1, M_2, M_3} \tensor \br{\alpha^{-1}_{TM_1, TM_2, TM_3}}^*}(
      \nabla_1(s_1) \amalg (\nabla_2(s_2) \amalg \nabla_3(s_3))
     )\alpha_{M_1, M_2, M_3}^{-1}
\end{align*}
where the last expression is easily seen to be
\[
  (\nabla_1(s) \amalg \nabla_2(s_2)) \amalg \nabla_3(s_3)
  = ((\nabla_1 \amalg \nabla_2) \amalg \nabla_3)((s_1 \amalg s_2) \amalg s_3)
\]

We can similarly show that the unitors in $\Man$ yield unitors for disjoint
unions of connections of bundles with equal fibres and equidimensional base
spaces. We have thus proved the following theorem.
\begin{thm}\label{thm:bundle_cat}
For a vector space $V$ and a non-negative integer $d$, the subcategory of the
category of connections consisting of all connections on objects in
$\Bun_d^{V}$ and all morphisms of connections between them is a monoidal
category under disjoint union.
\end{thm}
\begin{defn}[Category of {$(V, d)$}--Connections]
We call the subcategory of the category of connections in the above theorem the
category of connections on $V$--fibred bundles on $d$--dimensional manifolds or
of $(V, d)$--connections. We denote this category $\Conn^V_d$.
\end{defn}

As we have a notion of gluing for (a subset of) complex bundles with a fixed
fibre, we will develop a notion of gluing for connections defined on these
bundles. We will require some basic algebraic observations to achieve this.
Recall that for $\K = \R$ or $\C$ and any
$\K$--vector bundle $B \to M$ --- that is, smooth vector bundle with $\R$
or $\C$ as fibres --- its set $\Gamma(B)$ of sections is a
$\Cinf(M, \K)$--module with addition and scaling defined pointwise.
The set $\Lin\br{\Gamma(B), \Gamma(B')}$ of $\Cinf(M, \K)$--linear maps
$\Gamma(B) \to \Gamma(B')$ for another bundle $B' \to M$ is again a
$\Cinf(M, \K)$--module.  Identifying $\K$ with the subset of $\Cinf(M, \K)$ consisting of the constant
functions, this means that set of connections over the vector bundle $E \to M$
is contained $\Lin\br{\Gamma(E), \Gamma(E \tensor T^*M)}$.
We will now show that the set of all connections on $E$ is a coset of
a submodule of $\Lin\br{\Gamma(E), \Gamma(E \tensor T^*M)}$. It will follow from
the following elementary fact about modules over a unital ring.

\begin{defn}[Affine Linear Combinations]
For some unital ring $R$ and (left) $R$--module $V$, we will call
an $R$--linear combination $\sum_{i = 1}^{n} r_iv_i$ for $r_i \in R, v_i \in V$
affine if $\sum_{i = 1}^{n} r_i = 1$.
\end{defn}

\begin{lem}
For any ring $R$, a non-empty subset of a left $R$--module $V$ that is closed
under affine $R$--linear combinations is a coset of a submodule of $V$.
\end{lem}
\begin{proof}
Let $S$ be a subset of $V$ closed under affine $R$--linear combinations.
Since $S$ is non-empty, we may pick some $s \in S$ and define:
\[
S - s = V_S = \set[s' - s]{s' \in S}
\]
We verify that $V_S$ is a submodule. Let $s' - s, s'' - s \in V_S$ and
$r \in R$. Observe that $0 = s - s \in V_S$ while
$-(s' - s) = (-s' + s) - s$ where $-s' + s \in S$ by closure of $S$ under affine
combinations so that $-(s' - s) \in V_S$. We also have that:
\[
(s' - s) + (s'' - s) = (s' + s'' - s) - s \in V_S
\text{ since } s' + s'' - s \in S
\]
by closure of $S$ under affine combinations again. This shows that $V_S$ is a
subgroup of $V$. Then, we observe that
\[
r(s' - s) = (rs' + (1 - r)s) - s \in V_S \text{ since } rs' + (1 - r)s \in S
\]
again by the closure of $S$ under affine combination. This shows that
$V_S$ is a subgroup of $V$ closed under $R$--scaling and thus a submodule.
\end{proof}

\begin{thm}
The set of connections on a bundle $E \to M$ is a coset of a submodule of the
$\Cinf(M, \K)$--module $\Lin\br{\Gamma(E), \Gamma(E \tensor T^*M)}$.
\end{thm}
\begin{proof}
By the previous lemma, it suffices to show that the set of connections
is closed under affine $\Cinf(M, \K)$--linear combinations.
It is well known that the set of connections are closed under affine
$\K$--linear combinations \cite[10.5, p. 73]{LT17} --- although, the argument
there is for $\K = \R$, the same argument applies for $\K = \C$. We will show
that the same holds for affine $\Cinf(M, \K)$--linear combinations
of connections. For any $s \in \Gamma(E)$,
$r, r_1, \dots, r_n \in \Cinf(M, \K)$ with $\sum_{i = 1}^n r_i = 1$,
and connections $\nabla^1, \dots \nabla^n$ on $E$, we consider the following
expression:
\begin{align*}
\br{\sum_{i = 1}^n r_i\nabla^i}(r \cdot s)
&= \sum_{i = 1}^n (r_i\nabla^i)(r \cdot s) \\
&= \sum_{i = 1}^n r_i (\nabla^i(r \cdot s))
  && \text{scaling in $\Hom$ module} \\
&= \sum_{i = 1}^n r_i \br{r \cdot \nabla^i(s) + s \tensor dr}
  && \text{Leibniz property of $\nabla^i$} \\
&= \sum_{i = 1}^n r_i (r \nabla^i(s))
    + \br{\sum_{i = 1}^{n} r_i}(s \tensor dr) \\
&= \sum_{i = 1}^n r_i (r \nabla^i(s))
    + 1(s \tensor dr) \\
&= \sum_{i = 1}^n (r_i r) \nabla^i(s)
    + s \tensor dr \\
&= r \sum_{i = 1}^n r_i \nabla^i(s)
    + s \tensor dr
  && \text{commutativity of $\Cinf(M, \K)$} \\
&= r \br{\sum_{i = 1}^n r_i \nabla^i}(s)
    + s \tensor dr
\end{align*}
This shows that $\sum_{i = 1}^{n} r_i \nabla^i$ satisfies the Leibniz property.
Since $r_i, \nabla^i$ were arbitrary, we have the desired result.
\end{proof}

This allows us to write each connection $\nabla$ on
$E$ as a sum
\[
  \nabla = \varphi + A_{\nabla}
\]
where $\varphi$ is a fixed element of $\Lin(\Gamma(E), \Gamma(E \tensor T^*M))$
and $A_{\nabla}$ is an element of some fixed submodule of
$\Lin(\Gamma(E), \Gamma(E \tensor T^*M))$ varying with $\nabla$.
It is well known that
$\varphi$ can be taken as the exterior derivative operator $d$ but this will not
be important for our purposes. In particular,
this allows us to define for each $r \in \Cinf(\R)$, a mapping
on the set of connections as follows:
\[
  r \cdot \nabla := d + r \cdot A_{\nabla} =: d + A_{r \cdot \nabla}
\]
From this, a gluing operation for (a class of) connections is immediate. This
is made precise in the following theorem.

\begin{thm}
Let $E \to M$ and $E' \to M'$ be horizontal $1$--morphisms
(or bundle cobordisms) in $\BBun^V_{\s{C}}$ for some cobordism category $\s{C}$,
such that the gluing $E' * E \to M' * M$ exists. Then, for any two connections
$\nabla$ and $\nabla'$ on $E$ and $E'$ respectively, there exists a connection
$\nabla' * \nabla$ that has the same output as $\nabla$ on local sections of $E$
defined away from $E$'s trivialized boundary collar and to $\nabla'$ for local
sections of $E'$ defined away from $E'$'s trivialized boundary collar.
\end{thm}
\begin{proof}
Let the bump functions on $M$ and $M'$ used to trivialize $E$ and $E'$
respectively near boundaries, as defined in the proof of theorem
\ref{thm:bundle_gluing}, be $f$ and $f'$ respectively. Let
\[
  \nabla = d + A_{\nabla} \text{ and } \nabla' = d + A_{\nabla'}.
\]
We define:
\[
  \wh{\nabla} := f \cdot \nabla = d + f \cdot A_{\nabla}
  \text{ and } \wh{\nabla'} := f' \cdot \nabla' = d + f' \cdot A_{\nabla'}
\]
We can then extend the domain of $f$ to $M' * M$ by defining it to be zero on
$M'$. Similarly, we define $f'$ to be zero on $M$. This allows us to extend the
domain of $f \cdot A_{\nabla}$ and $f' \cdot A_{\nabla'}$ to $M' * M$ in the
same way. We can thus define:
\[
  \nabla' * \nabla := \wh{\nabla'} * \wh{\nabla}
    := d + f \cdot A_{\nabla} + f' \cdot A_{\nabla'}
\]
It is immediate that $\nabla' * \nabla$ satisfies the conditions of being a
connection pointwise. Hence, $\nabla' * \nabla$ is a connection on $E' * E$. We
also observe that where $f$ is $1$, $\nabla' * \nabla$ is equal to $\nabla$ and
likewise with $f', \nabla'$.
\end{proof}

\begin{defn}[Gluable Connection]
Connections of the form $\wh{\nabla}$ as in the previous theorem are called
gluable. The set of gluable connections on bundle cobordisms in
$\BBun^V_{\s{C}}$ is denoted $\Conn^V_{\s{C}}$.
\end{defn}

\begin{cor}
Gluing of connections on bundle cobordisms is associative and unital upto gauge
isomorphisms.
\end{cor}
\begin{proof}
For associativity, it suffices to observe that extending domains of bump
functions by zeros is associative and that associators from $\BBun^V_{\s{C}}$
are associator gauge transformations --- the proof of the latter claim is similar
to that of the associativity of the disjoint union of connections.

For unitality, we first observe that on boundary collars $f \cdot \nabla$ is
equal to $d$ for all connections $\nabla = d + A_{\nabla}$. The exterior
derivative operators, that are themselves connections, on the gluing
units of $\BBun^V_{\s{C}}$ --- that is, the cylinders on boundaries --- suffice
as the gluing units for connections. The unitors carry over like associators.
\end{proof}

Finally, one has the following fundamental fact.

\begin{thm}
For any cobordism double category $\s{C}$, $\BBun^V_{\s{C}}$ can be promoted to
a monoidal double category of connections by taking pairs $(E, \nabla)$ for each
bundle cobordism $E$ in $\BBun^V_{\s{C}}$ and a gluable connection $\nabla$ on
$E$ as the horizontal $1$--morphisms. Vertical $1$--morphisms are taken to be
bundle isomorphisms and $2$--morphisms are taken to be gauge transformations.
The rest of the structure is modified in the obvious ways.
\end{thm}

\begin{defn}
We denote the monoidal double category of connections in the last theorem as
$\CConn^V_{\s{C}}$.
\end{defn}

At this point, one might think that we can immediately define a notion of TQFT
by choosing suitable paths for a representative for each cobordism class,
taking the linear maps obtained by parallel transport along these paths and
combining them in a suitable way. However, the problem remains that two gauge
transformations need not take a collection of paths in the domain to the same
collection of paths in the codomain. We will build the machinery to handle this
matter in the next section.


\section{Paths for Parallel Transport}

In order to obtain linear maps by parallel transport on manifolds, we need
additional structure on top of connections. These are collections of paths on
manifolds along which we will parallel transport vectors in the fibres of a
bundle with connection. We shall now formalize this apparatus in terms of
categories. We will require a notion of graphs on manifolds whose vertices are
points, possibly repeated, on the manifold and whose edges are paths on the
manifold.

\subsection{Graphs Encoding Algebraic Expressions}\label{subsec:alg_graph_exp}

We will now describe a method of encoding expressions involving tensor products,
point-wise algebra products and composition of linear maps $A \to A$ for some
algebra $A$, using directed graphs. As a matter of convention, we will take all
graphs to mean directed acyclic graphs with possibly multiple copies of the same
vertex but where we do not allow more than one edge
with the same source and target nor self-loops, unless stated otherwise.
We should clarify that by multiple copies of the same vertex, we mean a
labelling of the vertices where multiple vertices can have the same label and
vertices with the same label are regarded as ``copies of the same vertex''.
We should also distinguish the word ``label'' from ``colour'' as we will also
require graphs to eventually be vertex--$2$--coloured in addition to the
aforementioned vertex labelling.
We also allow the underlying
undirected graph of any directed graph to be a forest. Given a graph
$G = (V, E)$, we will write $V = V(G)$ and $E = E(G)$. We now see a motivating
example.

\begin{exm}\label{exm:egraph1}
Consider a graph consisting of nine vertices $1, \dots, 9$ with edges:
\[\begin{array}{ccccc}
  (1, 3) &,& (1, 4) &,& (2, 3),\\
  (3, 5) &,& (3, 6) &,& (4, 7),\\
  (5, 8) &,& (6, 9) &,& (7, 9)
\end{array}\]
We visualize this graph as follows:
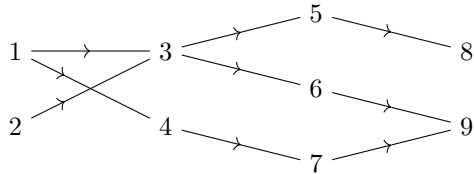
\[\begin{tikzpicture}[xscale=2,yscale=0.5]
\node at (0, 3) (v1) {1};
\node at (0, 1) (v2) {2};
\node at (1, 3) (v3) {3};
\node at (1, 1) (v4) {4};
\node at (2, 4) (v5) {5};
\node at (2, 2) (v6) {6};
\node at (2, 0) (v7) {7};
\node at (3, 3) (v8) {8};
\node at (3, 1) (v9) {9};
\midarrow{v1}{v3}
\midarrow[0.33]{v1}{v4}
\midarrow[0.33]{v2}{v3}
\midarrow{v3}{v5}
\midarrow{v3}{v6}
\midarrow{v4}{v7}
\midarrow{v5}{v8}
\midarrow{v6}{v9}
\midarrow{v7}{v9}
\end{tikzpicture}\]
Thinking of individual edges in the above graph as mappings $A \to A$ of some
bimonoid $A$ with product $m : A \otimes A \to A$ and $c : A \to A \otimes A$
in some monoidal category, we can think of multiple
incoming edges on a vertex --- for example, the incoming edges on $3$ --- as a
multiplication $\cwedge$ of maps and multiple outgoing edges --- for example, the
outgoing edges of $1$ --- as a comultiplication $\cvee$ of maps.
A bit of clarification is in order: for maps
$f : B \to A, g : C \to A, h : A \to B, k : A \to C$ in the
ambient monoidal category, by the multiplication $f \cwedge g$, we mean the
composite:
\[\begin{tikzcd}
 B \otimes C \ar[r, "f \otimes g"] & A \otimes A \ar[r, "m"] & A
\end{tikzcd}\]
and by the comultiplication $h \cvee k$, we mean the composite
\[\begin{tikzcd}
A \ar[r, "c"] & A \otimes A \ar[r, "h \otimes k"] & B \otimes C
\end{tikzcd}\]
Parallel edges,
multiplications or comultiplications with disjoint sources and targets can then
be thought of as a tensor product of maps. We can treat a single vertex as an
identity mapping.

We can then modify the graph as follows:
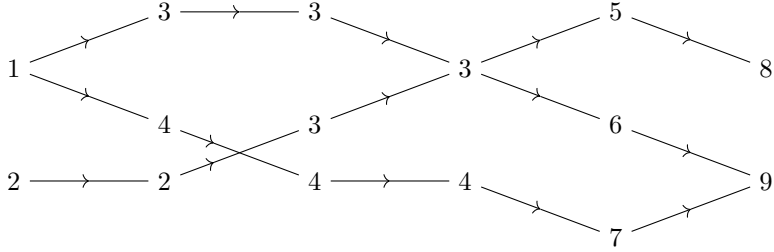
\[\begin{tikzpicture}[xscale=2,yscale=0.75]
\node at (-2, 3) (v1) {1};
\node at (-2, 1) (v2) {2};
\node at (-1, 4) (v33) {3};
\node at (-1, 2) (v44) {4};
\node at (-1, 1) (v22) {2};
\node at (0, 4) (v333) {3};
\node at (0, 2) (v3333) {3};
\node at (0, 1) (v444) {4};
\node at (1, 3) (v3) {3};
\node at (1, 1) (v4) {4};
\node at (2, 4) (v5) {5};
\node at (2, 2) (v6) {6};
\node at (2, 0) (v7) {7};
\node at (3, 3) (v8) {8};
\node at (3, 1) (v9) {9};
\midarrow{v1}{v33}
\midarrow{v1}{v44}
\midarrow{v2}{v22}
\midarrow{v33}{v333}
\midarrow[0.33]{v22}{v3333}
\midarrow[0.33]{v44}{v444}
\midarrow{v333}{v3}
\midarrow{v3333}{v3}
\midarrow{v444}{v4}
\midarrow{v3}{v5}
\midarrow{v3}{v6}
\midarrow{v4}{v7}
\midarrow{v5}{v8}
\midarrow{v6}{v9}
\midarrow{v7}{v9}
\end{tikzpicture}\]
Thinking of edges between copies of the same vertex, which did not originally
exist, as the identity map and crossing edges as a twist operation --- denoted
$\boxtimes$ --- of maps, we can capture the graph into an algebraic expression of
the following form:
\begin{align*}
        & ((5, 8) \tensor ((6, 9) \cwedge (7, 9))) \\
  \circ & (((3, 5) \cvee (3, 6)) \tensor (4, 7)) \\
  \circ & (((3, 3) \cwedge (3, 3)) \tensor (4, 4)) \\
  \circ & ((3, 3) \tensor ((4, 4) \boxtimes (2, 3))) \\
  \circ & (((1, 3) \cvee (1, 4)) \tensor (2, 2))
\end{align*}
Again, some clarification is needed here: we are now assuming a braiding
isomorphism $\tau_{B, C} : B \otimes C \to C \otimes B$ in the ambient monoidal
category, for each pair of objects $B, C$, and given maps $f : W \to X$
and $g : Y \to Z$, the twist $f \boxtimes g$ is defined to be the composite:
\[\begin{tikzcd}
W \otimes Y \ar[r, "f \otimes g"] & X \otimes Z \ar[r, "\tau_{X, Z}"] &
  Z \otimes X
\end{tikzcd}\]

Observe that we have taken a directed acyclic graph and converted it to a
diagram in $\Cob_{2}$ --- the symmetric monoidal category of cobordisms of
dimension $2$. This, in turn, yields an expression involving operations on
endomorphisms of a monoid in some monoidal category. However, we should note
that this is just one possible interpretation of the graph above.
\end{exm}

This motivates us to define an algorithm for extracting algebraic expressions
from a directed acyclic graph such as the one above.

\begin{alg}[Graph Reduction Algorithm]\label{alg:expr-construction}
Let $G = (V, E)$ be any graph. We make the following modifications to $G$:
\begin{enumerate}
\setlength{\itemsep}{0pt}

\item For each vertex $v$, choose an ordering $(v, w_1), \dots, (v, w_n)$ of its
outgoing edges.

\item Let $S(G)$ be the set consisting of vertices with no incoming
edges --- called the source vertices of $G$ --- and $T(G)$, the set consisting of
vertices with no outgoing edges --- called the target vertices of $G$. We then
choose an ordering of $S(G)$. Note that $S(G) \cap T(G)$ might be non-empty
because of vertices with no edges, incoming or outgoing --- these will be called
the edgeless vertices. These choices induce more structure on $G$:

\begin{lem}\label{lem:level-ordering}
Let $G$ be as above with the chosen ordering of $S(G)$ and, for each vertex,
the chosen ordering of its outgoing edges. Then, the vertices of $G$ can be
written as a disjoint union $V_1 \amalg V_2 \amalg \cdots \amalg V_n$ such that:
\begin{enmrt}
\li each $V_i$ is non-empty when $V(G)$ is non-empty,
\li $V_1 = S(G)$,
\li for each $1 < i \leq n$, the sources of the incoming edges of the vertices
in $V_i$ are all contained in $\coprod_{j = 1}^{ i - 1} V_{j}$, and
\li there is an induced ordering of the vertices within $V_i$ for each
$1 < i \leq n$.
\end{enmrt}
\end{lem}
\begin{proof}
We construct the sets $V_i$ as follows:
\begin{enmrt}
\li Let $V_1$ be the set of vertices with no incoming edges --- i.e.,
$V_1 := S(G)$
\li Let $V_{k + 1}$ be the set of vertices in
$G \setminus (V_1 \cup \cdots \cup V_k)$ with incoming edges only from
$V_1, \dots, V_k$.
\end{enmrt}
The first three properties are easy to check.
For the last property,
we first note that the ordering on $S(G) = V_1$ induces an ordering of $V_2$ as
follows. Let $S(G)$ be ordered as $u_1, \dots, u_n$. Let the ordering of the
outgoing edges of $u_i$ be $(u_i, v_{i, 1}), \dots, (u_i, v_{i, k_i})$. We let
$v_i := v_{1, i}$ for $1 \leq i \leq k_1$. For $1 \leq j < n$,
$k_j < i \leq k_{j + 1}$, we let $v_i := v_{j + 1, i}$. Then, the $v_i$ are an
ordering of $V_2$. We can repeat this process with $V_2$ in place of $S(G)$ and
so on to obtain an ordering for each level.
\end{proof}

\begin{defn}[Level Ordering]\label{defn:level-ordering}
We shall call this partition of $V(G)$ a level ordering of $G$ and the subsets
$V_i$, its levels, with respect to the choices. We will call the process
described in the above proof as the level ordering algorithm.
\end{defn}

\begin{cor}\label{cor:inedge-ordering}
Let $G$ be as above with the chosen ordering of $S(G)$ and the outgoing edges
of each vertex. For each vertex $v$ of $G$, this induces an ordering $(u_1, v),
\dots, (u_k, v)$ of the incoming edges of $v$.
\end{cor}
\begin{proof}
$v$ must lie in $V_i$ for some $i$. If $i = 1$, then $v$ has no incoming edges.
Otherwise, $i > 1$ and the source vertices of the incoming edges of
$v$ are in $V_{i - 1}$ and inherit an ordering from $V_{i - 1}$, giving an
ordering of the incoming edges of $v$.
\end{proof}

\begin{cor}\label{cor:lvltolvl}
Let $V_1, V_2, \dots, V_k$ be the induced level ordering of a graph
$G$ as above. Then for each $i \in \set{2, \dots, k}$ and each
$v \in V_i$ that is not edgeless, there is an edge $(u, v)$ with
$u \in V_{i - 1}$.
\end{cor}
\begin{proof}
If not, the level-ordering algorithm would place $v$ at a lower level.
\end{proof}

\item Copy vertices and add edges to make the following modification, where the
incoming edges are in order from the lowest at the top to the highest at the
bottom:
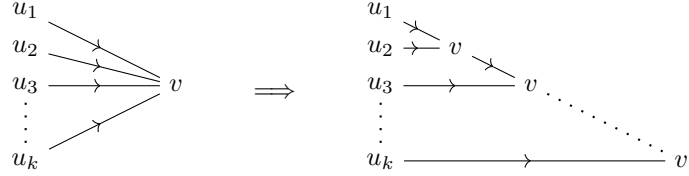
\[\begin{tikzpicture}
\node at (2, 0)     (v) {$v$};
\node at (0, 1)     (u1) {$u_1$};
\node at (0, 0.5)   (u2) {$u_2$};
\node at (0, 0)     (u3) {$u_3$};
\node at (0, -1)    (uk) {$u_k$};
\midarrow{u1}{v}
\midarrow{u2}{v}
\midarrow{u3}{v}
\midarrow{uk}{v}
\draw[thick, loosely dotted] (u3) -- (uk);
\end{tikzpicture}
\qquad
\begin{tikzpicture}
\node at (0, 1)   (TOP)     {};
\node at (0, 0)   (TO)      {$\Longrightarrow$};
\node at (0, -1)  (BOTTOM)  {};
\end{tikzpicture}
\qquad
\begin{tikzpicture}
\node at (4, -1)     (v) {$v$};
\node at (0, 1)     (u1) {$u_1$};
\node at (0, 0.5)   (u2) {$u_2$};
\node at (0, 0)     (u3) {$u_3$};
\node at (0, -1)    (uk) {$u_k$};
\node at (1, 0.5) (a)  {$v$};
\node at (2, 0)   (b)  {$v$};
\midarrow{u1}{a}
\midarrow{u2}{a}
\midarrow{a}{b}
\midarrow{u3}{b}
\draw[thick, loosely dotted] (b) -- (v);
\draw[thick, loosely dotted] (u3)   -- (uk);
\midarrow{uk}{v};
\end{tikzpicture}
\]

\item Copy vertices and edges to make the following modification similar to the
previous step:
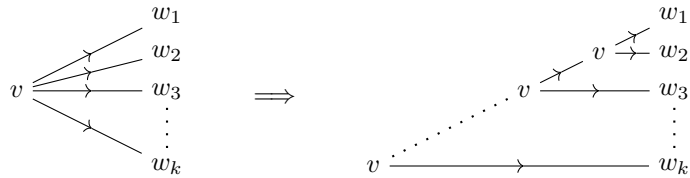
\[
\begin{tikzpicture}
\node at (-2, 0)     (v) {$v$};
\node at (0, 1)     (w1) {$w_1$};
\node at (0, 0.5)   (w2) {$w_2$};
\node at (0, 0)     (w3) {$w_3$};
\node at (0, -1)    (wk) {$w_k$};
\midarrow{v}{w1}
\midarrow{v}{w2}
\midarrow{v}{w3}
\midarrow{v}{wk}
\draw[thick, loosely dotted] (u3) -- (uk);
\end{tikzpicture}
\qquad
\begin{tikzpicture}
\node at (0, 1)   (TOP)     {};
\node at (0, 0)   (TO)      {$\Longrightarrow$};
\node at (0, -1)  (BOTTOM)  {};
\end{tikzpicture}
\qquad
\begin{tikzpicture}
\node at (-4, -1)   (v)  {$v$};
\node at (0, 1)     (w1) {$w_1$};
\node at (0, 0.5)   (w2) {$w_2$};
\node at (0, 0)     (w3) {$w_3$};
\node at (0, -1)    (wk) {$w_k$};
\node at (-1, 0.5) (a)  {$v$};
\node at (-2, 0)   (b)  {$v$};
\midarrow{a}{w1}
\midarrow{a}{w2}
\midarrow{b}{a}
\midarrow{b}{w3}
\draw[thick, loosely dotted] (b) -- (v);
\draw[thick, loosely dotted] (w3)   -- (wk);
\midarrow{v}{wk};
\end{tikzpicture}
\]

\item At this point, every vertex has both indegree and outdegree at most $2$.
The chosen edge orderings induces a local orientation --- in an
informal sense --- on the graph. By this we mean, that this allows us to
distinguish the following diagrams in a precise sense:
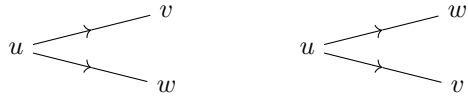
\[
\begin{tikzpicture}
\node at (0, 0) (u) {$u$};
\node at (2, 0.5) (v) {$v$};
\node at (2, -0.5) (w) {$w$};
\midarrow{u}{v}
\midarrow{u}{w}
\end{tikzpicture}
\qquad
\qquad
\begin{tikzpicture}
\node at (0, 0) (u) {$u$};
\node at (2, -0.5) (v) {$v$};
\node at (2, 0.5) (w) {$w$};
\midarrow{u}{v}
\midarrow{u}{w}
\end{tikzpicture}
\]
The distinction is that if we take $(u, v) < (u, w)$ in the left picture, say,
then we can take $(u, v) > (u, w)$ in the right picture.

Then, for every edge that is shared between a ``multiplication'' and a
``comultiplication'', considering the edge orderings chosen before, we have the
following possibilities for common edges and we make the modifications shown:
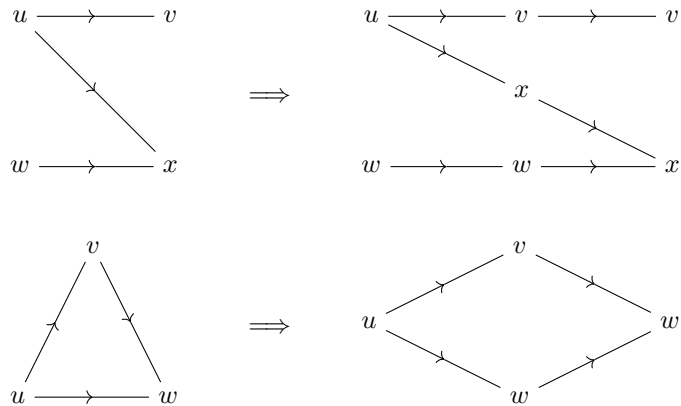
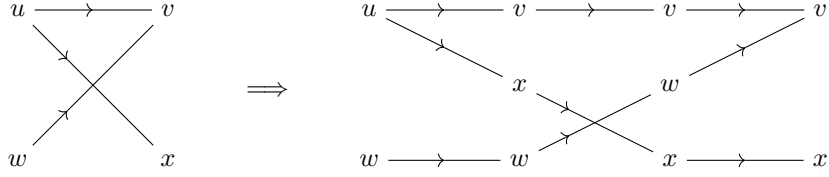
\[
\begin{tikzpicture}
\node at (0, 2) (u) {$u$};
\node at (2, 2) (v) {$v$};
\node at (0, 0) (w) {$w$};
\node at (2, 0) (x) {$x$};
\midarrow{u}{v}
\midarrow{u}{x}
\midarrow{w}{x}
\end{tikzpicture}
\qquad
\begin{tikzpicture}
\node at (0, 1)   (TOP)     {};
\node at (0, 0)   (TO)      {$\Longrightarrow$};
\node at (0, -1)  (BOTTOM)  {};
\end{tikzpicture}
\qquad
\begin{tikzpicture}
\node at (0, 2) (u) {$u$};
\node at (2, 2) (a) {$v$};
\node at (4, 2) (v) {$v$};
\node at (2, 1) (b) {$x$};
\node at (0, 0) (w) {$w$};
\node at (2, 0) (c) {$w$};
\node at (4, 0) (x) {$x$};
\midarrow{u}{a}
\midarrow{a}{v}
\midarrow{u}{b}
\midarrow{b}{x}
\midarrow{w}{c}
\midarrow{c}{x}
\end{tikzpicture}
\]
\[
\begin{tikzpicture}
\node at (0, 0) (u) {$u$};
\node at (1, 2) (v) {$v$};
\node at (2, 0) (w) {$w$};
\midarrow{u}{v}
\midarrow{v}{w}
\midarrow{u}{w}
\end{tikzpicture}
\qquad
\begin{tikzpicture}
\node at (0, 1)   (TOP)     {};
\node at (0, 0)   (TO)      {$\Longrightarrow$};
\node at (0, -1)  (BOTTOM)  {};
\end{tikzpicture}
\qquad
\begin{tikzpicture}
\node at (0, 1) (u) {$u$};
\node at (2, 2) (v) {$v$};
\node at (2, 0) (a) {$w$};
\node at (4, 1) (w) {$w$};
\midarrow{u}{v}
\midarrow{u}{a}
\midarrow{v}{w}
\midarrow{a}{w}
\end{tikzpicture}
\]

\[
\begin{tikzpicture}
\node at (0, 2) (u) {$u$};
\node at (2, 2) (v) {$v$};
\node at (0, 0) (w) {$w$};
\node at (2, 0) (x) {$x$};
\midarrow{u}{v}
\midarrow[0.33]{w}{v}
\midarrow[0.33]{u}{x}
\end{tikzpicture}
\qquad
\begin{tikzpicture}
\node at (0, 1)   (TOP)     {};
\node at (0, 0)   (TO)      {$\Longrightarrow$};
\node at (0, -1)  (BOTTOM)  {};
\end{tikzpicture}
\qquad
\begin{tikzpicture}
\node at (0, 2) (u) {$u$};
\node at (2, 2) (a) {$v$};
\node at (4, 2) (aa) {$v$};
\node at (6, 2) (v) {$v$};
\node at (2, 1) (b) {$x$};
\node at (4, 1) (bb) {$w$};
\node at (0, 0) (w) {$w$};
\node at (2, 0) (c) {$w$};
\node at (4, 0) (cc) {$x$};
\node at (6, 0) (x) {$x$};
\midarrow{u}{a}
\midarrow{u}{b}
\midarrow{w}{c}
\midarrow{a}{aa}
\midarrow[0.33]{b}{cc}
\midarrow[0.33]{c}{bb}
\midarrow{aa}{v}
\midarrow{bb}{v}
\midarrow{cc}{x}
\end{tikzpicture}
\]

We also make the modifications obtained from the rotations of the above diagrams
about a horizontal edge. After these modifications have been applied, there are
no edges shared between ``multiplications'' and ``comultiplications''.

\item Note that $G$ remains acyclic even after the modifications, the vertices
with no incoming edges remain so, the ordering
of the vertices with no incoming edges still applies, and that the
orderings of outgoing edges for each vertex can be modified accordingly, with no
further choice involved. We then construct a
level ordering of this modified $G$.

\begin{rmk}
The vertices with no edges, incoming or outgoing, are always in the first level.
However, vertices with no outgoing edges need not always be in the last level.
\end{rmk}

\item\label{alg:edgeless}
Consider the vertices with no outgoing edges but not at the last level. For each
such vertex $u$ in $V_i$, we add a copy $u'$, to $V_{i + 1}$. Its insertion
order needs to be made precise. Call the vertices in $V_{i + 1}$ with an
incoming edge from some vertex above $u \in V_{i}$, the vertices above $u$ in
$V_{i + 1}$. Similarly, call the vertices in $V_{i + 1}$ with no incoming edges
from $u$ or a vertex above $u \in V_i$, the vertices below $u$ in $V_{i + 1}$.
A copy of $u$, say $u'$, is inserted
in the position right after the vertices above $u$ and before the vertices below
$u$ in $V_{i + 1}$. We then add in the edge $(u, u')$. We continue this process
with $V_{i + 1}$ in place of $V_i$ and so on, until there is a copy
of $u$ in each level after $V_i$ with a path connecting them.

\begin{rmk}
Each vertex without any edges, incoming or outgoing, are also copied in this
way, noting that these vertices are placed in the first level during the level
ordering.
\end{rmk}

\item For each level-skipping edge $(u, v)$ --- that is, with $u \in V_k$ and
$v \in V_{k'}$ for some $k' > k + 1$ --- we insert a copy of $u$ in $V_{i}$ for
each $k < i < k'$ in positions similar to the insertions in the last step. We
call this copy $u'$. We then add an edge $(u, u')$. We repeat this with
$V_{i + 1}$ in place of $V_{i}$ and $u'$ in place of $u$. After this process
completes, we delete the edge $(u, v)$ and add an edge from the copy of $u$
in $V_{k' - 1}$ to $v$. After we complete this process for every
level-skipping edge, level-skipping edges are replaced by paths connecting the
source vertex to the target vertex of these edges.

\item Observe that we can now identify the graph $G$ with a cobordism in
$\Cob_2$ from $|S(G)|$ copies of $S^1$ to $|T(G)|$ copies of $S^1$! The
identifications of the generating structures are the obvious ones:
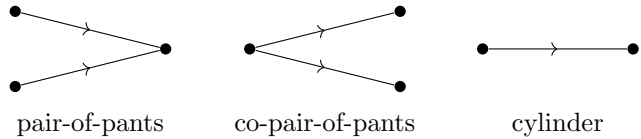
\[
\begin{tikzpicture}
\node[circle, fill, inner sep=1.5pt] at (0, 0) (u) {};
\node[circle, fill, inner sep=1.5pt] at (-2, 0.5) (v) {};
\node[circle, fill, inner sep=1.5pt] at (-2, -0.5) (w) {};
\midarrow{v}{u}
\midarrow{w}{u}
\node at (-1, -1) (lbl) {pair-of-pants};
\end{tikzpicture}
\qquad
\begin{tikzpicture}
\node[circle, fill, inner sep=1.5pt] at (0, 0) (u) {};
\node[circle, fill, inner sep=1.5pt] at (2, -0.5) (v) {};
\node[circle, fill, inner sep=1.5pt] at (2, 0.5) (w) {};
\midarrow{u}{v}
\midarrow{u}{w}
\node at (1, -1) (lbl) {co-pair-of-pants};
\end{tikzpicture}
\qquad
\begin{tikzpicture}
\node[circle, fill, inner sep=1.5pt] at (0, 0) (u) {};
\node at (2, -0.5) (v) {};
\node at (2, 0.5) (w) {};
\node[circle, fill, inner sep=1.5pt] at (2, 0) (x) {};
\midarrow{u}{x}
\node at (1, -1) (lbl) {cylinder};
\end{tikzpicture}
\]
Note that the cap or cup were not used. Furthermore, there are also crossings of
various kinds such as:
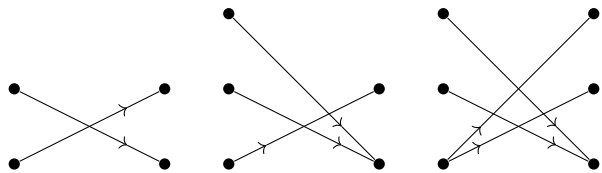
\[
\begin{tikzpicture}
\node[circle, fill, inner sep=1.5pt] at (0, 0.5) (y) {};
\node[circle, fill, inner sep=1.5pt] at (0, -0.5) (z) {};
\node[circle, fill, inner sep=1.5pt] at (-2, 0.5) (w) {};
\node[circle, fill, inner sep=1.5pt] at (-2, -0.5) (x) {};
\midarrow[0.75]{w}{z}
\midarrow[0.75]{x}{y}
\end{tikzpicture}
\qquad
\begin{tikzpicture}
\node[circle, fill, inner sep=1.5pt] at (0, 0.5) (x) {};
\node[circle, fill, inner sep=1.5pt] at (0, -0.5) (y) {};
\node[circle, fill, inner sep=1.5pt] at (-2, 1.5) (u) {};
\node[circle, fill, inner sep=1.5pt] at (-2, 0.5) (v) {};
\node[circle, fill, inner sep=1.5pt] at (-2, -0.5) (w) {};
\midarrow[0.75]{u}{y}
\midarrow[0.75]{v}{y}
\midarrow[0.25]{w}{x}
\end{tikzpicture}
\qquad
\begin{tikzpicture}
\node[circle, fill, inner sep=1.5pt] at (0, 1.5) (x) {};
\node[circle, fill, inner sep=1.5pt] at (0, 0.5) (y) {};
\node[circle, fill, inner sep=1.5pt] at (0, -0.5) (z) {};
\node[circle, fill, inner sep=1.5pt] at (-2, 1.5) (u) {};
\node[circle, fill, inner sep=1.5pt] at (-2, 0.5) (v) {};
\node[circle, fill, inner sep=1.5pt] at (-2, -0.5) (w) {};
\midarrow[0.75]{u}{z}
\midarrow[0.75]{v}{z}
\midarrow[0.25]{w}{x}
\midarrow[0.25]{w}{y}
\end{tikzpicture}
\]
but these can be identified with manifolds representing crossing morphisms in
$\Cob_2$.
The resulting cobordism or, more precisely, the manifold representing said
cobordism yields an algebraic expression involving the composition and monoidal
product of $\Cob_2$ and the generating morphisms. This algebraic expression is
taken to be the expression of the original graph $G$.

\end{enumerate}
\end{alg}

\begin{defn}[Reduced Graphs]
Given a graph $G$, the modified graph resulting from the above algorithm will be
called a reduction of $G$ and will be denoted $G'$.
Any graph with the properties that were introduced in $G'$ through the
algorithm --- we will skip listing them again --- is called reduced.
The expression resulting from the algorithm above is called
an expression of $G$ and denoted $\Exp{G}$. In particular, given a reduced
graph $H$, we may apply the last step to directly obtain an expression
$\Exp{H}$.
\end{defn}

\begin{rmk}
Note that the possible expressions of a graph are by no means unique as
we made a large number of arbitrary choices in our process. Nevertheless, the
algorithm above is deterministic after the second step.
That is, if we fix a choice for the ordering of the source vertices and for the
orderings of the outgoing edges for each vertex, then the algorithm gives a
fixed expression.
\end{rmk}

Even though we made many choices in the above algorithm, the only
choices that cannot be thought of as canonical in any way are the chosen edge
orderings and the ordering of the source vertices. We are thus motivated to make
the following definition.

\begin{defn}[Expression Graph]
A graph with a chosen ordering of its source vertices and,
for each vertex, a chosen orderings of the outgoing edges of the vertex 
is called an expression graph.
\end{defn}

Thus, given an expression graph, we have an algorithm to extract its expression
consistently as long as we keep the other choices in the above algorithm fixed.
We will later see that $\Exp{G}$ is functorial in a suitable sense.

\subsection{Morphisms of Expression Graphs}

It is of interest to define morphisms of expression graphs. They will be
necessary to define interesting constructions on expression graphs later on.

\begin{defn}[Expression Homomorphism]
An expression graph homomorphism or expression homomorphism is a graph
homomorphism $f : G \to H$ between expression graphs such that $f(S(G))$ is
contained in some level set $L$ in $H$, $f$ preserves the ordering of $S(G)$ in
$L$ and $f$ preserves the edge orderings for each vertex.
\end{defn}

We then note some useful facts concerning level orderings and expression
isomorphisms. First, every expression graph $G$ has a level ordering as
defined in \ref{defn:level-ordering} so that the vertex
set is a disjoint union of the levels. Then, we have a level function $l_G$ for
$G$ which assigns to each vertex $v$ the integer $n$ for which $v$ is in the
$n$--th level set.

\begin{lem}\label{thm:expiso_lvlpres}
Expression isomorphisms $f : G \to H$ are level preserving, i.e.
\[
  l_G(v) = l_H(f(v))
\]
\end{lem}
\begin{proof}
Let
$\set{V_i}_{i = 1}^{N}$ and
$\set{W_j}_{i = 1}^{M}$ be the level sets of $G$ and $H$
respectively. Let $v \in V_k$ for some $1 \leq k \leq N$.
By corollary \ref{cor:lvltolvl}, there is a path $v_1, \dots, v_k = v$ with each
$v_i \in V_i$. Then, there is a path $f(v_1), \dots, f(v_k) = f(v)$ in $H$ with
$l_H(f(v_{i})) < l_H(f(v_{i + 1}))$ (since edges only go forward in levels) and
$l_H(f(v_1)) \geq 1$. Hence, by extending this path in $H$ backwards to $W_1$,
we have that $l_H(f(v)) \geq l_G(v)$.
On the other hand, let $f(v) \in W_m$ so that there is
a path $w_1, \dots, w_m$ in $H$ with each $w_i \in W_i$. Using $f^{-1}$ on this
path and a similar argument as the one before, we can show that
$l_H(f(v)) \leq l_G(v)$.
\end{proof}

\begin{cor}
For any expression isomorphism $f : G \to H$, if $V_i$ is the $i$--th level set
in $G$, then $f(V_i)$ is the $i$--th level set in $H$.
\end{cor}
\begin{proof}
Let $W_i$ be the $i$--th level set of $H$. Then, by the previous lemma
$f(V_i) \subset W_i$. Again, using $f^{-1}$ in place of $f$ and $W_i$ in place
of $V_i$, we have
$f^{-1}(W_i) \subset V_i \implies W_i \subset f(V_i)$. Hence, $f(V_i) = W_i$.
\end{proof}

\begin{lem}
Expression isomorphisms $f : G \to H$ are order-preserving on levels, i.e.
\[
  u \leq v \iff f(u) \leq f(g)
\]
\end{lem}
\begin{proof}
Let $u, v$ be in the same level set in $G$ and $u \leq v$. We proceed by
induction on $l_G(u) = l_G(v)$. When $l_G(u) = l_G(v) = 1$, then
$f(u) \leq f(v)$, by definition. The other direction is obtained similarly with
$f^{-1}$ in place of $f$.

Let $l_G(u) = l_G(v) = k + 1$. By the construction of level sets and their
ordering given in the graph reduction algorithm, there exist
$u', v' \in V(G)$ such that the following hold:
\begin{enmrt}
\li $l_G(u') = l_G(v') = k$
\li $u' \leq v'$
\li there are edges $(u', u)$ and $(v' v)$
\li there are no edges $(u'', u)$ or $(v'', v)$ with $u'' < u'$ or $v'' < u'$
\li if $u' = v'$, $(u', u) \leq (v', v)$
\end{enmrt}
By the previous corollary, $l_H(f(u')) = l_H(f(v')) = k$.
If there is an $x < f(u')$ or a $y < f(v')$ in $V(H)$ with edges
$(x, f(u'))$ or $(y, f(v'))$, then by induction, $f^{-1}(x) < u'$ or
$f^{-1}(y) < v'$ with some edge $(f^{-1}(x), u)$ or $(f^{-1}(y), v)$,
contradicting the conditions on $u'$ and $v'$. Thus, the edges $(f(u'), f(u))$
and $(f(v'), f(v))$ ensure that $f(u) \leq f(v)$. The other direction is again
obtained similarly by replacing $f$ with $f^{-1}$.
\end{proof}

\begin{cor}\label{cor:expiso_unique}
Expression isomorphisms $f : G \to H$ are unique.
\end{cor}
\begin{proof}
Let $g : G \to H$ be another expression isomorphism. Then both $f$ and $g$
restrict to order-preserving bijections on the finite level sets and hence must
agree on the level sets. Thus, $f$ and $g$ agree on $G$.
\end{proof}

In light of the last corollary, it is reasonable to consider expression graphs
up to expression isomorphisms from this point onwards. We will next define some
useful constructs on expression graphs that will facilitate our desired
modification of TQFTs.

\subsection{Constructs on Expression Graphs}

Consider expression graphs $G$ and $H$. We can take the disjoint unions of their
vertex and edge sets. It is clear that the edge orderings of the vertices of
$G$ and $H$ collectively provide an edge ordering for every vertex of
$G \amalg H$. We observe that $S(G \amalg H) = S(G) \amalg S(H)$ so that the
orderings of $S(G)$ and $S(H)$ provide an ordering of $S(G \amalg H)$, where the
vertices of $G$ come before those of $H$. The empty graph is an expression graph
and hence acts as a unit for the disjoint union operation. It is easy to see
that the associators and unitors for the coproduct in the category of sets are
expression isomorphisms.

We then proceed to define a gluing of expression graphs. For an expression
graph $G$, $S(G)$ without any edges is itself an expression graph.
The edgeless vertices are in both $S(G)$ and $T(G)$. Since $S(G)$ is also the
first level of $G$, $S(G)$ is ordered, by definition. We then observe that
$T(G)$ is the union of the last level and the set of edgeless vertices. Thus,
using a method similar to point \ref{alg:edgeless} in the algorithm given in
\S\ref{subsec:alg_graph_exp}, we have an induced ordering of $T(G)$, so that
$T(G)$ is also an expression graph without any edges.
There are obvious order-preserving expression homomorphisms
$S(G) \hto G \hot T(G)$ that are isomorphisms onto their images. We note
that if there is a an expression isomorphism $\psi : S(H) \to T(G)$ for
expression graphs $G$ and $H$, then $\psi$ is unique by \ref{cor:expiso_unique}.
This allows us to define the following notion of gluing.

\begin{defn}[Gluing of Expression Graphs]
Let $G$ and $H$ be expression graphs such that there is a unique expression
isomorphism\footnote{Note that, in this case, this is simply an order-preserving
bijection of sets.}
\[
  \psi_{G, H} : S(H) \to T(G)
\]
Then, we say that $G$ and $H$ are gluable at $S(H) \cong T(G)$. We denote the
pushout of the following span in $\Set$ as $V(H * G)$:
\[
  H \hot S(H) \to[\psi_{G, H}] T(G) \hto G
\]
The graph with vertex set $V(H * G)$ and edge set $E(G) \amalg E(H)$ is denoted
$H * G$ and called the gluing of $H$ with $G$ at $S(H) \cong T(G)$.
\end{defn}

\begin{thm}
Gluings $H * G$ of expression graphs $G$ and $H$ at $S(H) \cong T(G)$ are
expression graphs.
\end{thm}
\begin{proof}
First we show that $S(H * G) = S(G)$. If $v \in S(H * G) \setminus S(G)$, then
we consider the case that $v \in G \subset H * G$. Since $v \not\in S(G)$ and
gluing does not delete edges, $v$ must have an incoming edge in $G$. We then
consider the case that $v \in H \subset H * G$. Here, we can assume
$v \not\in S(H)$ because $S(H) = T(G)$ in $H * G$ so that in this case, $v$ has
an incoming edge in $H$. In either case, $v \not\in S(H * G)$ --- a contradiction
showing that $S(H * G) \setminus S(G)$ is empty.
Thus, $S(H * G) \subseteq S(G)$.
If $v \in S(G)$, then $v$ has no incoming edges in $G$. It is clear that gluing
cannot introduce incoming edges to the source vertices of $G$ so that
$S(G) \subset S(H * G)$. Thus, $H * G$ inherits the ordering of its source
vertices (including edgeless vertices which are also in $S(H) \cong T(G)$) from
$G$.

We observe that only the vertices in $T(G)$ are the sites of new edges and these
are all outgoing while the vertices in $T(G)$ have no outgoing edges in $G$.
Thus, $H * G$ inherits edge orderings unambiguously from $G$ and $H$,
collectively. Therefore, $H * G$ is an expression graph.
\end{proof}

\begin{exm}\label{exm:expression_gluing}
Consider $G$ from example \ref{exm:egraph1} and $H$ as follows:
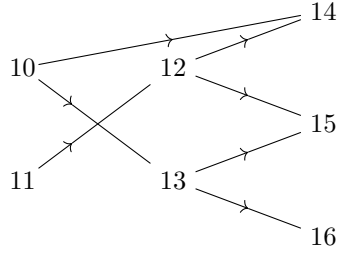
\[\begin{tikzpicture}[xscale=2,yscale=0.75]
\node at (3, 3) (v10) {$10$};
\node at (3, 1) (v11) {$11$};
\node at (4, 3) (v12) {$12$};
\node at (4, 1) (v13) {$13$};
\node at (5, 4) (v14) {$14$};
\node at (5, 2) (v15) {$15$};
\node at (5, 0) (v16) {$16$};
\midarrow[0.33]{v8}{v13}
\midarrow[0.33]{v9}{v12}
\midarrow{v8}{v14}
\midarrow{v12}{v14}
\midarrow{v12}{v15}
\midarrow{v13}{v15}
\midarrow{v13}{v16}
\end{tikzpicture}\]
We have the following diagram of $H * G$:
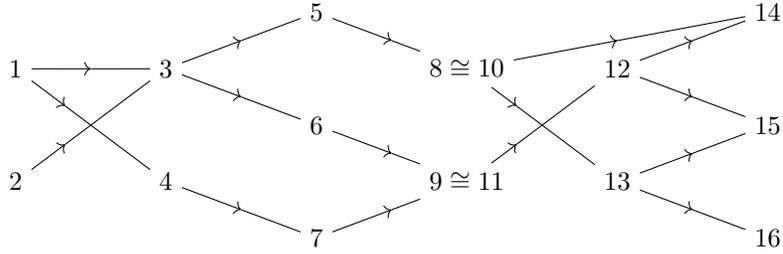
\[\begin{tikzpicture}[xscale=2,yscale=0.75]
\node at (0, 3) (v1) {$1$};
\node at (0, 1) (v2) {$2$};
\node at (1, 3) (v3) {$3$};
\node at (1, 1) (v4) {$4$};
\node at (2, 4) (v5) {$5$};
\node at (2, 2) (v6) {$6$};
\node at (2, 0) (v7) {$7$};
\node at (3, 3) (v8) {$8 \cong 10$};
\node at (3, 1) (v9) {$9 \cong 11$};
\node at (4, 3) (v12) {$12$};
\node at (4, 1) (v13) {$13$};
\node at (5, 4) (v14) {$14$};
\node at (5, 2) (v15) {$15$};
\node at (5, 0) (v16) {$16$};
\midarrow{v1}{v3}
\midarrow[0.33]{v1}{v4}
\midarrow[0.33]{v2}{v3}
\midarrow{v3}{v5}
\midarrow{v3}{v6}
\midarrow{v4}{v7}
\midarrow{v5}{v8}
\midarrow{v6}{v9}
\midarrow{v7}{v9}
\midarrow[0.33]{v8}{v13}
\midarrow[0.33]{v9}{v12}
\midarrow{v8}{v14}
\midarrow{v12}{v14}
\midarrow{v12}{v15}
\midarrow{v13}{v15}
\midarrow{v13}{v16}
\end{tikzpicture}\]
\end{exm}

It is then easy to verify that gluing of expression graphs is associative and
unital up to expression isomorphism much like the disjoint union. We can further
verify that the data of expression graphs defined so far form a monoidal double
category whose objects are finite, ordered sets, vertical $1$--morphisms are
unique order isomorphisms, horizontal $1$--morphisms $G : U \to V$ are
expression graphs $G$ with $S(G) \cong U$ and $T(G) \cong V$, and $2$--morphisms
are expression isomorphisms, with horizontal composition given by gluing and
monoidal product given by disjoint union.

We now observe some properties of the graph reduction algorithm
\ref{alg:expr-construction}.

\begin{exm}
We observe that, in general, the reduction does not result in the
same graph when we apply it before gluing as opposed to after gluing. Let
$G$ and $H$ be as follows:
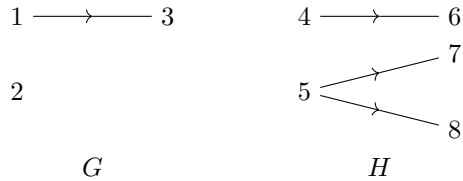
\[
\begin{tikzpicture}[yscale=0.5]
\node at (0, 3) (v1) {$1$};
\node at (0, 1) (v2) {$2$};
\node at (2, 3) (v3) {$3$};
\midarrow{v1}{v3}
\node at (1, -1) (G) {$G$};
\end{tikzpicture}
\qquad
\qquad
\begin{tikzpicture}[yscale=0.5]
\node at (0, 3) (v4) {$4$};
\node at (0, 1) (v5) {$5$};
\node at (2, 3) (v6) {$6$};
\node at (2, 2) (v7) {$7$};
\node at (2, 0) (v8) {$8$};
\midarrow{v4}{v6}
\midarrow{v5}{v7}
\midarrow{v5}{v8}
\node at (1, -1) (H) {$H$};
\end{tikzpicture}
\]
Applying the reduction on $G$ and $H$ separately and then gluing
the results yields:
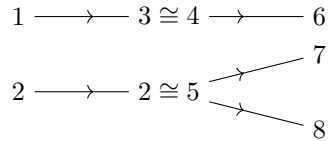
\[
\begin{tikzpicture}[yscale=0.5]
\node at (0, 3) (v1) {$1$};
\node at (0, 1) (v2) {$2$};
\node at (2, 3) (v4) {$3 \cong 4$};
\node at (2, 1) (v5) {$2 \cong 5$};
\node at (4, 3) (v6) {$6$};
\node at (4, 2) (v7) {$7$};
\node at (4, 0) (v8) {$8$};
\midarrow{v1}{v4}
\midarrow{v2}{v5}
\midarrow{v4}{v6}
\midarrow{v5}{v7}
\midarrow{v5}{v8}
\end{tikzpicture}
\]
Applying the reduction on $H * G$ results in:
\[
\begin{tikzpicture}[yscale=0.5]
\node at (0, 3) (v1) {$1$};
\node at (2, 3) (v4) {$3 \cong 4$};
\node at (0, 1) (v5) {$2 \cong 5$};
\node at (4, 3) (v6) {$6$};
\node at (2, 2) (v7) {$7$};
\node at (2, 0) (v8) {$8$};
\node at (4, 2) (v77) {$7$};
\node at (4, 0) (v88) {$8$};
\midarrow{v1}{v4}
\midarrow{v4}{v6}
\midarrow{v5}{v7}
\midarrow{v5}{v8}
\midarrow{v7}{v77}
\midarrow{v8}{v88}
\end{tikzpicture}
\]
However, in this case, we observe that the expressions $\Exp{H * G}$ and
$\Exp{H} \circ \Exp{G}$ differ only by ``identity'' edges and hence, when
interpreted as a morphism in a monoidal category would give morphisms that are
equal.
\end{exm}

In light of the previous example, it seems useful to decide when we should
consider two expressions to be the same and so we make the following definition.

\begin{defn}[Equivalence of Expressions]
Given graphs $G$ and $H$, we will consider $\Exp{G}$ and $\Exp{H}$ equivalent
if their reductions $G'$ and $H'$ differ by insertion or contraction of identity
edges --- edges whose source and target labels are the same and that are not
part of a pair-of-pants or a co-pair-of-pants. In this case, we will simply
write them as equal:
\[
\Exp{G} = \Exp{H}
\]
Given an expression graph $W$ with no edges (a finite ordered set), the identity
expression on $W$ is the exprssion $\id_W = (\id_{w_1}, \dots, \id_{w_n})$,
where $W = \set{w_1 < w_2 < \dots < w_n}$. We will also equate:
\[
  \Exp{G} \circ \id_W = \Exp{G} \text{ and } \id_W \circ \Exp{H} = \Exp{H}
\]
whenever $S(G) \cong W$ and $T(H) \cong W$.
\end{defn}

With no work, we have the following from our definition of equivalence of
expressions:

\begin{thm}[Unitality of Expressions]
For an expression graph $G$, we have
\[
  \Exp{G * S(G)} = \Exp{G} \circ \Exp{S(G)} = \Exp{G}
  = \Exp{T(G)} \circ \Exp{G} = \Exp{T(G) * G}
\]
\end{thm}

\begin{rmk}
In the above theorem, we could go further and take the cylinders on $S(G)$ and
$T(G)$ in place of $S(G)$ and $T(G)$ respectively given by expanding each vertex
into an identity edge.
\end{rmk}

We then have the following results which show that all differences that can
arise between reductions of graphs before and after gluing can be accounted for
by simple cases similar to the last example above.

\begin{lem}
For expression graphs $G$ and $H$ with $n$ and $m$ levels respectively and
gluable at $S(H) \cong T(G)$, the number of levels in $H * G$ is $n + m - 1$.
\end{lem}
\begin{proof}
Let the level sets of $G$ be $V_1, \dots, V_n$ and
those of $H$, be $W_1, \dots, W_m$. By \ref{cor:lvltolvl}, there exists a path
$w_1, \dots, w_m$ with $w_i \in W_i$ in $H$, and also a path
$v_1, \dots, v_n = w_1$ with $v_i \in V_i$ in $G$. The concatenation of these
paths shows that $H * G$ has at least $n + m - 1$ levels.

Let $X_1, \dots, X_k$ be the level sets of $H * G$. For any vertex
$x \in X_k$, by \ref{cor:lvltolvl}, there must be a path $x_1, \dots, x_k = x$
in $H * G$ with $x_i \in X_i$. However, since the edges of $H * G$ are the edges
of $G$ or $H$, we must have that all $x_1, \dots, x_{j}$ are in
$G \setminus H$, $x_{j + 1} \in G \cap H = S(H) \cong T(G)$ and all
$x_{j + 2}, \dots, x_k$ are in $H \setminus G$, for some $j$. If this is
not the case, then we must have an edge from some vertex in $H$ to some vertex
in $G$, which is impossible. Therefore, $H * G$ has at most $n + m - 1$ levels.
\end{proof}

We observe that the first five steps of algorithm \ref{alg:expr-construction}
commute with gluing. One might think that the sixth
step onwards does not. However, we will see that it does. For this, we note the
following result about changes in level sets after gluing.

\begin{thm}
For expression graphs $G$ and $H$ gluable at $S(H) \cong T(G)$,
the level sets of $H * G$ are formed by taking the level sets of $G$ and those
of $H$, identifying the first level of $H$ with the last level of $G$ and moving
vertices between levels in the following pattern:
\begin{enmrt}
\li Move vertices in $S(H)$ which glue to edgeless vertices to the first level.
\li For each such vertex $v$ in $S(H)$ and each edge $(v, v')$ where $(v, v')$
is the first incoming edge of $v'$, move $v'$ to the second level.
\li Repeat the previous step with $v'$ in place of $v$.
\end{enmrt}
\end{thm}
\begin{proof}
By the previous lemma, it makes sense to say that the levels of $H * G$ are
formed by taking the levels of
$G$ and $H$ and moving vertices between them. Let the levels of $G$ be $V_1,
\dots, V_n$ and those of $H$ be $W_1, \dots, W_m$. We will informally say that
$V_n = W_1$ and that the levels of $H * G$ are
$V_1, \dots, V_n $ (or $W_1$)$, \dots, W_m$ for the sake of simplifying the
language, although the movement of vertices makes the equality false. The cases
for movement of vertices in the gluing $H * G$ are as follows:
\begin{enmrt}
\li If $v \in V_i$ for $1 \leq i \leq n$ before gluing, then $v$ cannot move to
another $V_j$ or $W_k$.
We proceed by induction. For $i = 1$, we observe that $S(H * G) = S(G) = V_1$
and hence no $v \in V_1$ can move forward since it has no incoming edges. For
$i = r + 1$, if $v \in V_i$, then there exists an edge $(v', v) \in G$ with
$v' \in V_{r}$ and $v'$ does not move by induction. Hence, $v$ cannot move to
$V_j$ for $j < i$ since $(v', v)$ would then be an edge that does not go forward
in levels. Since all the incoming edges of $v \in V_i$ are from
$V_1, \dots, V_{i - 1}$, $v$ cannot move to $V_j$ for some $j > i$ or $W_k$ for
some $k$, because this would result in $v$ not having edges from the immediate
previous level.

\li If $v \in S(H)$ before gluing, then $v$ can only move when $v$ glues to some
vertex $v_G \in G$ that has no incoming edges in $G$ and, in this case,
$v_G \in V_1$ in $H * G$ and $v$ moves to $V_1$ by gluing to $v_G$. This is true
because if $v_G$ has some incoming edges in $G$, then $v_G \in V_n$ and we are
in the previous case so that $v_G$ stays in its level and $v$ merely to glues to
$v_G$.

\li If $v \in W_k$ for some $1 < k \leq m$ before gluing, then there exists a
path $w_1, \dots, w_k = v$ with $w_i \in W_i$ in $H$. We can additionally assume
that for each $i$, $(w_i, w_{i + 1})$ is the first incoming edge of $w_{i + 1}$.
After gluing, if $v$ moves to some $W_{k'}$, then the first incoming edge of
$v = w_k$ is from $w_{k - 1}$ so that $w_{k - 1}$ must have also moved to
$W_{k' - 1}$, by our definition of level ordering. Repeating this argument, we
see that $w_1$ must move forward in levels which is impossible by the previous
case. If, on the other hand, $v$ moves backwards, we must have each $w_i$ move
backwards as well because otherwise, we will have at least one edge not going
forward in levels. In particular, $w_1$ must move backwards and hence to $V_1$,
by the first case, since $w_1 \in W_1 = V_n$ before gluing.
In this case, $w_2$ moves to $V_2$, $w_3$ to $V_3$ and so on until $w_k$ moves
to $V_k$ if $k \leq n$. If $k > n$, then $w_1, \dots, w_n$ move to
$V_1, \dots, V_n$ and $w_{n + 1}, \dots, w_{n + k - n} = w_k$ move to
$W_1, \dots, W_{k - n}$ respectively.
\end{enmrt}
\end{proof}

Whether we construct level sets before or after gluing, such movements must take
place and hence the sixth step where we construct the levels commutes with
gluing. We then observe that applying the seventh step before gluing results in
no edgeless vertices in both graphs so that there are no movements of vertices
between levels by the above theorem. Applying the seventh step after gluing
results in movement of vertices and then ``extending'' some vertices to the last
level. In either case, even though we do not obtain the same graph, the parts
that differ, do so only by prefixes and suffixes of edges between copies of the
same vertex that are not part of a pair-of-pants or a co-pair-of-pants --- that
is, identity edges.

For the eighth step, observe that a vertex $v \in S(H)$ after the seventh step
can have either one outgoing edge or two. Now, suppose it glues to an edgeless
vertex in $G$ and moves to the first level. Then, if $v$ has one outgoing edge
$(v, v')$, then either $v'$ gets moved to the second level or the edge is now
a level-skipping edge. In this case, it gets broken up into identity edges
that traverese one level at a time, after the eighth step.
Now, consider the case the $v$ has two outgoing edges $(v, w_1), (v, w_2)$.
None of $w_1$ and $w_2$ can have another outgoing edge after the seventh step.
Hence, the given edges are the first and only incoming edges of both $w_1$ and
$w_2$, and both get moved to the second level, without introducing level
skipping edges. We can repeat analogous arguments with $v'$, $w_1$ and $w_2$ in
place of $v$ to show that the resulting graphs after the eighth step, before and
after gluing differ by identity edges.

Thus, the resulting cobordisms in the ninth step differ by identity cobordisms
at some places and hence are equivalent, yielding the following theorem.

\begin{thm}[Compositionality of Expressions]
For expression graphs $G$ and $H$, gluable at $S(H) \cong T(G)$, we have
\[
  \Exp{H * G} = \Exp{H} \circ \Exp{G}
\]
\end{thm}

We have thus shown that the reduction of a graph to produce an expression is
``functorial'' but in a loose sense since we have not specified a codomain for
$\Exp{}$.
We can also define the following constructs unambiguously:

\begin{defn}[Expression Tensor Product]
Let $G$ and $H$ be expression graphs. We define:
\[
  \Exp{G} \tensor \Exp{H} := \Exp{G \amalg H}
\]
\end{defn}

\begin{defn}[Expression Substitution]
Given an expression graph $G$, we write
\[
  \Exp{G}[f] \text{ or } \Exp{G}[f(u, v)] \text{ or } \Exp{G}[(u, v)/f(u, v)]
\]
to denote the expression obtained by replacing each edge $(u, v) \in G$ in
$\Exp{G}$ with some string of symbols $f(u, v)$, depending on $(u, v)$.
\end{defn}

\subsection{Geometric Realization}

We wish to use expression graphs to generate linear maps by viewing the edges as
paths in a manifold equipped with an algebra-fibred bundle with connection,
taking their associated parallel transport maps, and combining these maps in a
pattern encoded in the expression graph. To accomplish this, we use the
following simple notion of geometric realization of graphs.

\begin{defn}[Geometric Graph]
A graph in a manifold $M$ or a geometric graph is a graph $G = (V, E)$ equipped
with a function $\gamma : E \to C^0(I, M)$, called a geometric realization of
$G$. We call $M$ the realizing manifold of $G$ under $\gamma$. For an edge from
$u$ to $v$ we write $\gamma_{u, v}$ to denote the path associated to that edge.
\end{defn}

\begin{rmk}
A geometric graph is essentially a collection of paths in a manifold but we
consider one or more copies of each path and ``identify'' their end-points in a
pattern encoded by a graph, even though these paths need not share end-points
--- that is, we do not strictly require $\gamma_{u, v}(1) = \gamma_{v, w}(0)$.
We will see that this relaxation is essential to defining our notion of TQFTs.
We will ultimately be interested in expression graphs in manifolds, which will
provide us with a way to associate linear maps to manifolds.
\end{rmk}

\begin{exm}\label{exm:geomegraph}
Consider the graph in example \ref{exm:egraph1}. Then, we consider a mapping of
the edges to paths on a surface as shown below. Note that here we have
$\gamma_{u, v}(1) = \gamma_{v, w}(0)$ for most points but not all --- for
instance, $\gamma_{2, 3}$ and $\gamma_{3, 5}$ do not share any end-points.
\begin{figure}[H]\label{fig:geomegraph}
\begin{center}
\begin{tikzpicture}[x=0.75pt,y=0.75pt,yscale=-0.95,xscale=0.95]

\draw  [dash pattern={on 4.5pt off 4.5pt}] (6,138.5) .. controls (6,65.32) and (107.63,6) .. (233,6) .. controls (358.37,6) and (460,65.32) .. (460,138.5) .. controls (460,211.68) and (358.37,271) .. (233,271) .. controls (107.63,271) and (6,211.68) .. (6,138.5) -- cycle ;
\draw    (52.31,161.49) .. controls (119.1,125.59) and (102.4,162.09) .. (169.19,126.19) ;
\draw [shift={(169.19,126.19)}, rotate = 331.74] [color={rgb, 255:red, 0; green, 0; blue, 0 }  ][fill={rgb, 255:red, 0; green, 0; blue, 0 }  ][line width=0.75]      (0, 0) circle [x radius= 3.35, y radius= 3.35]   ;
\draw [shift={(111.13,143.84)}, rotate = 180.1] [color={rgb, 255:red, 0; green, 0; blue, 0 }  ][line width=0.75]    (10.93,-3.29) .. controls (6.95,-1.4) and (3.31,-0.3) .. (0,0) .. controls (3.31,0.3) and (6.95,1.4) .. (10.93,3.29)   ;
\draw [shift={(52.31,161.49)}, rotate = 331.74] [color={rgb, 255:red, 0; green, 0; blue, 0 }  ][fill={rgb, 255:red, 0; green, 0; blue, 0 }  ][line width=0.75]      (0, 0) circle [x radius= 3.35, y radius= 3.35]   ;
\draw    (169.19,126.19) .. controls (233.04,103.33) and (202.45,29.92) .. (275.29,52.51) ;
\draw [shift={(275.29,52.51)}, rotate = 17.23] [color={rgb, 255:red, 0; green, 0; blue, 0 }  ][fill={rgb, 255:red, 0; green, 0; blue, 0 }  ][line width=0.75]      (0, 0) circle [x radius= 3.35, y radius= 3.35]   ;
\draw [shift={(217.18,75.67)}, rotate = 476.11] [color={rgb, 255:red, 0; green, 0; blue, 0 }  ][line width=0.75]    (10.93,-3.29) .. controls (6.95,-1.4) and (3.31,-0.3) .. (0,0) .. controls (3.31,0.3) and (6.95,1.4) .. (10.93,3.29)   ;
\draw [shift={(169.19,126.19)}, rotate = 340.3] [color={rgb, 255:red, 0; green, 0; blue, 0 }  ][fill={rgb, 255:red, 0; green, 0; blue, 0 }  ][line width=0.75]      (0, 0) circle [x radius= 3.35, y radius= 3.35]   ;
\draw    (52.31,161.49) .. controls (178.45,204.05) and (218.84,194.46) .. (229.3,171.66) ;
\draw [shift={(229.3,171.66)}, rotate = 294.65] [color={rgb, 255:red, 0; green, 0; blue, 0 }  ][fill={rgb, 255:red, 0; green, 0; blue, 0 }  ][line width=0.75]      (0, 0) circle [x radius= 3.35, y radius= 3.35]   ;
\draw [shift={(142.99,186.54)}, rotate = 191.16] [color={rgb, 255:red, 0; green, 0; blue, 0 }  ][line width=0.75]    (10.93,-3.29) .. controls (6.95,-1.4) and (3.31,-0.3) .. (0,0) .. controls (3.31,0.3) and (6.95,1.4) .. (10.93,3.29)   ;
\draw [shift={(52.31,161.49)}, rotate = 18.65] [color={rgb, 255:red, 0; green, 0; blue, 0 }  ][fill={rgb, 255:red, 0; green, 0; blue, 0 }  ][line width=0.75]      (0, 0) circle [x radius= 3.35, y radius= 3.35]   ;
\draw    (150.83,241.06) .. controls (275.76,227.54) and (145.96,205.62) .. (189.19,146.19) ;
\draw [shift={(169.19,126.19)}, rotate = 286.3] [color={rgb, 255:red, 0; green, 0; blue, 0 }  ][fill={rgb, 255:red, 0; green, 0; blue, 0 }  ][line width=0.75]      (0, 0) circle [x radius= 3.35, y radius= 3.35]   ;
\draw [shift={(192.755,201.39)}, rotate = 410.89] [color={rgb, 255:red, 0; green, 0; blue, 0 }  ][line width=0.75]    (10.93,-3.29) .. controls (6.95,-1.4) and (3.31,-0.3) .. (0,0) .. controls (3.31,0.3) and (6.95,1.4) .. (10.93,3.29)   ;
\draw [shift={(150.83,241.06)}, rotate = 353.83] [color={rgb, 255:red, 0; green, 0; blue, 0 }  ][fill={rgb, 255:red, 0; green, 0; blue, 0 }  ][line width=0.75]      (0, 0) circle [x radius= 3.35, y radius= 3.35]   ;
\draw [shift={(189.19,146.19)}, rotate = 353.83] [color={rgb, 255:red, 0; green, 0; blue, 0 }  ][fill={rgb, 255:red, 0; green, 0; blue, 0 }  ][line width=0.75]      (0, 0) circle [x radius= 3.35, y radius= 3.35]   ;
\draw    (229.3,171.66) .. controls (243.97,235.64) and (284.76,184.82) .. (314.44,210.62) ;
\draw [shift={(314.44,210.62)}, rotate = 41] [color={rgb, 255:red, 0; green, 0; blue, 0 }  ][fill={rgb, 255:red, 0; green, 0; blue, 0 }  ][line width=0.75]      (0, 0) circle [x radius= 3.35, y radius= 3.35]   ;
\draw [shift={(263.43,205.67)}, rotate = 180.44] [color={rgb, 255:red, 0; green, 0; blue, 0 }  ][line width=0.75]    (10.93,-3.29) .. controls (6.95,-1.4) and (3.31,-0.3) .. (0,0) .. controls (3.31,0.3) and (6.95,1.4) .. (10.93,3.29)   ;
\draw [shift={(229.3,171.66)}, rotate = 77.09] [color={rgb, 255:red, 0; green, 0; blue, 0 }  ][fill={rgb, 255:red, 0; green, 0; blue, 0 }  ][line width=0.75]      (0, 0) circle [x radius= 3.35, y radius= 3.35]   ;
\draw    (169.19,126.19) .. controls (246.58,102.99) and (236.29,142.42) .. (284.82,127.77) ;
\draw [shift={(284.82,127.77)}, rotate = 343.19] [color={rgb, 255:red, 0; green, 0; blue, 0 }  ][fill={rgb, 255:red, 0; green, 0; blue, 0 }  ][line width=0.75]      (0, 0) circle [x radius= 3.35, y radius= 3.35]   ;
\draw [shift={(227.95,120.58)}, rotate = 190.98] [color={rgb, 255:red, 0; green, 0; blue, 0 }  ][line width=0.75]    (10.93,-3.29) .. controls (6.95,-1.4) and (3.31,-0.3) .. (0,0) .. controls (3.31,0.3) and (6.95,1.4) .. (10.93,3.29)   ;
\draw    (314.44,210.62) .. controls (399.23,240.98) and (364.55,150.12) .. (406.29,102.26) ;
\draw [shift={(406.29,102.26)}, rotate = 311.09] [color={rgb, 255:red, 0; green, 0; blue, 0 }  ][fill={rgb, 255:red, 0; green, 0; blue, 0 }  ][line width=0.75]      (0, 0) circle [x radius= 3.35, y radius= 3.35]   ;
\draw [shift={(377.94,180.41)}, rotate = 465.08] [color={rgb, 255:red, 0; green, 0; blue, 0 }  ][line width=0.75]    (10.93,-3.29) .. controls (6.95,-1.4) and (3.31,-0.3) .. (0,0) .. controls (3.31,0.3) and (6.95,1.4) .. (10.93,3.29)   ;
\draw  [color={rgb, 255:red, 0; green, 0; blue, 0 }  ,draw opacity=1 ][dash pattern={on 0.84pt off 2.51pt}] (16.54,136.11) .. controls (25.63,126.66) and (66.01,145.03) .. (106.72,177.14) .. controls (147.44,209.25) and (173.08,242.94) .. (163.99,252.39) .. controls (154.9,261.85) and (114.52,243.48) .. (73.81,211.37) .. controls (33.09,179.26) and (7.45,145.57) .. (16.54,136.11) -- cycle ;
\draw    (275.29,52.51) .. controls (328.72,77.63) and (326.15,20.9) .. (367.89,59.19) ;
\draw [shift={(367.89,59.19)}, rotate = 42.53] [color={rgb, 255:red, 0; green, 0; blue, 0 }  ][fill={rgb, 255:red, 0; green, 0; blue, 0 }  ][line width=0.75]      (0, 0) circle [x radius= 3.35, y radius= 3.35]   ;
\draw [shift={(323.24,52.43)}, rotate = 512.6] [color={rgb, 255:red, 0; green, 0; blue, 0 }  ][line width=0.75]    (10.93,-3.29) .. controls (6.95,-1.4) and (3.31,-0.3) .. (0,0) .. controls (3.31,0.3) and (6.95,1.4) .. (10.93,3.29)   ;
\draw    (290.82,100.77) .. controls (383.33,149.3) and (366.22,80.72) .. (406.29,102.26) ;
\draw [shift={(350.78,118.25)}, rotate = 537.99] [color={rgb, 255:red, 0; green, 0; blue, 0 }  ][line width=0.75]    (10.93,-3.29) .. controls (6.95,-1.4) and (3.31,-0.3) .. (0,0) .. controls (3.31,0.3) and (6.95,1.4) .. (10.93,3.29)   ;
\draw [shift={(290.82,100.77)}, rotate = 42.53] [color={rgb, 255:red, 0; green, 0; blue, 0 }  ][fill={rgb, 255:red, 0; green, 0; blue, 0 }  ][line width=0.75]      (0, 0) circle [x radius= 3.35, y radius= 3.35]   ;
\draw  [color={rgb, 255:red, 0; green, 0; blue, 0 }  ,draw opacity=1 ][dash pattern={on 0.84pt off 2.51pt}] (138.62,104.37) .. controls (148.74,97.58) and (179.13,108.92) .. (206.48,129.69) .. controls (233.84,150.47) and (247.81,172.81) .. (237.69,179.59) .. controls (227.57,186.38) and (197.19,175.04) .. (169.83,154.27) .. controls (142.47,133.49) and (128.5,111.15) .. (138.62,104.37) -- cycle ;
\draw  [dash pattern={on 0.84pt off 2.51pt}] (268.26,15.88) .. controls (281.5,13.99) and (305.14,60.41) .. (321.05,119.55) .. controls (336.96,178.69) and (339.13,228.17) .. (325.89,230.05) .. controls (312.64,231.94) and (289.01,185.52) .. (273.09,126.38) .. controls (257.18,67.24) and (255.02,17.76) .. (268.26,15.88) -- cycle ;
\draw  [dash pattern={on 0.84pt off 2.51pt}] (358.5,48.46) .. controls (367.96,43.55) and (391.4,56.66) .. (410.84,77.74) .. controls (430.28,98.83) and (438.37,119.9) .. (428.91,124.82) .. controls (419.45,129.73) and (396.01,116.62) .. (376.57,95.54) .. controls (357.13,74.45) and (349.04,53.37) .. (358.5,48.46) -- cycle ;

\draw (89.78,125.75) node [anchor=north west][inner sep=0.75pt]    {$\gamma _{1}{}_{,}{}_{3}$};
\draw (122.67,165.43) node [anchor=north west][inner sep=0.75pt]    {$\gamma _{1}{}_{,}{}_{4}$};
\draw (165.08,200.67) node [anchor=north west][inner sep=0.75pt]    {$\gamma _{2}{}_{,}{}_{3}$};
\draw (179.33,67.29) node [anchor=north west][inner sep=0.75pt]    {$\gamma _{3}{}_{,}{}_{5}$};
\draw (208.3,100.52) node [anchor=north west][inner sep=0.75pt]    {$\gamma _{3}{}_{,}{}_{6}$};
\draw (242.33,185.1) node [anchor=north west][inner sep=0.75pt]    {$\gamma _{4}{}_{,}{}_{7}$};
\draw (299.22,38.41) node [anchor=north west][inner sep=0.75pt]    {$\gamma _{5}{}_{,}{}_{8}$};
\draw (326.9,100.09) node [anchor=north west][inner sep=0.75pt]    {$\gamma _{6}{}_{,}{}_{9}$};
\draw (344.63,175.21) node [anchor=north west][inner sep=0.75pt]    {$\gamma _{7}{}_{,}{}_{9}$};

\end{tikzpicture}

\end{center}
\caption{A geometric graph}
\end{figure}
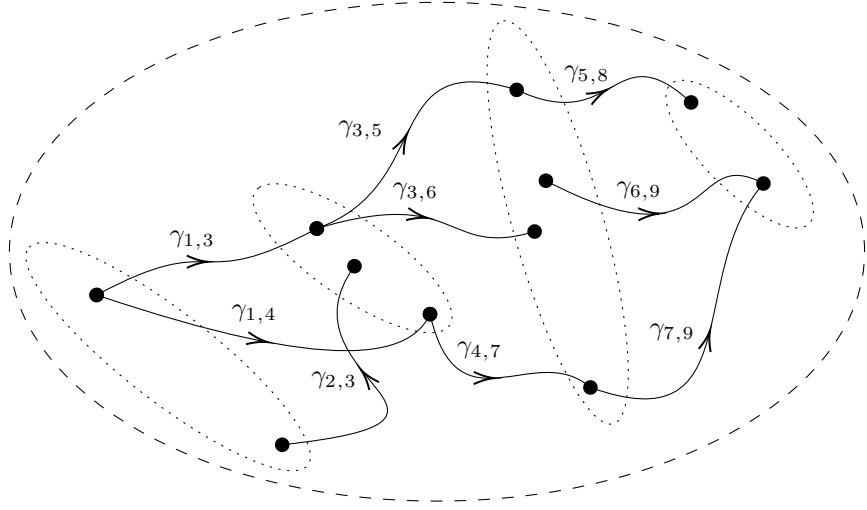
We notice that the same algebraic expression carries over:
\begin{align*}
        & (\gamma_{5, 8} \tensor (\gamma_{6, 9} \cwedge \gamma_{7, 9})) \\
  \circ & ((\gamma_{3, 5} \cvee \gamma_{3, 6}) \tensor \gamma_{4, 7}) \\
  \circ & ((\gamma_{3, 3} \cwedge \gamma_{3, 3}) \tensor \gamma_{4, 4}) \\
  \circ & (\gamma_{3, 3} \tensor (\gamma_{4, 4} \boxtimes \gamma_{2, 3})) \\
  \circ & ((\gamma_{1, 3} \cvee \gamma_{1, 4}) \tensor \gamma_{2, 2})
\end{align*}
\end{exm}

At this point, we make a necessary observation. Consider geometric expression
graphs in a manifold equipped with an $A$--fibred bundle with connection. Then,
consider the expression of an expression graph realized in this manifold. If we
substitute the edges in the expression of the graph with the parallel transport
maps along the paths associated to the edges, we obtain a linear map from a
non-zero tensor power of $A$ to another such tensor power, given the graph has
at least one edge. However, there is no immediate way to obtain a map of the
form $\R \to A$ or $A \to \R$ in this method.

We recall that cobordisms $\varnothing \to S^1$ or
$S^1 \to \varnothing$ in $\Cob_2$ are given by composing cobordisms
$S^1 \to S^1$ with the following structures:
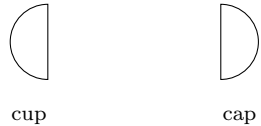
\[\begin{tikzpicture}[scale=0.5]
\cupcob{0, 2}
\node at (1.5, 0) (cuplbl) {{\footnotesize cup}};
\end{tikzpicture}
\qquad \qquad \qquad
\begin{tikzpicture}[scale=0.5]
\capcob{3, 2}
\node at (1.5, 0) (caplbl) {{\footnotesize cap}};
\end{tikzpicture}\]
We would like analogous structures for expression graphs. We can use vertex
colouring to distinguish between copies of the empty manifold and copies of
$S^1$ in the cobordism manifold resulting from the graph reduction.
So, we consider vertex coloured expression graphs.

\begin{defn}[Pretransport Graph]
A pretransport graph is an expression graph $G$ with a colour function
$V(G) \to \set{{\color{green!55!black}\text{green}},
{\color{blue!55!black}\text{blue}}}$ (with no constraints on adjacency of
vertices).
\end{defn}

\begin{defn}[Pretransport Homomorphism]
A pretransport homomorphism is an expression homomorphism that preserves the
colouring of vertices.
\end{defn}

From \ref{cor:expiso_unique}, it is immediate that:

\begin{cor}
Pretransport isomorphisms are unique.
\end{cor}

We now redefine the source and target of a pretransport graph so that
constructions like expressions and gluing can be adapted to this new setting
where we allow empty graphs to be sources and targets.

\begin{defn}
Let $G$ be a pretransport graph with $s_1, \dots, s_n$ the ordering of the
source vertices of its underlying expression graph and $t_1, \dots, t_m$ the
ordering of the target vertices thereof. Let $s_{i_1}, \dots, s_{i_p}$ be the
green source vertices and $t_{j_1}, \dots t_{j_q}$ be the green target vertices
in $G$. Then, the source of $G$ is defined to be the expression graph
\[
  S(G) := \set{s_1, \dots, s_{i_1 - 1}}
       \amalg \varnothing
       \amalg \set{s_{i_1 + 1}, \dots, s_{i_2}}
       \amalg \varnothing
       \amalg \cdots
       \amalg \varnothing
       \amalg \set{s_{i_p + 1}, \dots, s_{n}}
\]
The target of $G$ is defined to be the expression graph
\[
  T(G) := \set{t_{1}, \dots, t_{j_1 - 1}}
       \amalg \varnothing
       \amalg \set{t_{j_1 + 1}, \dots, t_{j_2}}
       \amalg \varnothing
       \amalg \cdots
       \amalg \varnothing
       \amalg \set{t_{j_q + 1}, \dots, t_{m}}
\]
\end{defn}

\begin{defn}
Gluing of expression graphs extends to gluing of pretransport graphs if we
consider pretransport isomorphisms instead of expression isomorphisms. More
precisely, for pretransport graphs $G$ and $H$, if $S(H) \cong T(G)$, then the
gluing is exactly the same as for expression graphs except for positions where
we have instances of $\varnothing$. For this, we have the following cases:
\begin{enmrt}
\li $H$ and $G$ both have green vertices at that position: we glue them
as usual.
\li $H$ has a green vertex but $G$ has $\varnothing$: we simply forget
that $G$ has $\varnothing$.
\li $H$ has $\varnothing$ but $G$ has a green vertex: we forget the copy
of $\varnothing$
\li $H$ and $G$ both have copies of $\varnothing$: we keep a single copy
of $\varnothing$
\end{enmrt}
\end{defn}

\begin{defn}[Expression Construction]
The expression of a pretransport graph $G$ is obtained by running the graph
reduction algorithm and expression construction on $G$, copying colours whenever
vertices are copied, and then identifying green vertices
with the empty manifold in the resulting cobordism in $\Cob_2$, adding caps and
cups where appropriate.
\end{defn}

We then consider the class of graphs that will yield our desired linear maps.

\begin{defn}[Transport Graph]
We call a geometric pretransport graph a transport graph.
\end{defn}

\begin{defn}[Geometric Homomorphism]
Let $G$ and $H$ be graphs in manifolds $M$ and $N$ with geometric realizations
$\gamma^G$ and $\gamma^H$ respectively, then a homomorphism $h : G \to H$
equipped with a smooth map $f : M \to N$ making the following diagram commute is
called a geometric homomorphism:
\[\begin{tikzpicture}[baseline=(a).base]
\node[scale=\diagscale] (a) at (0, 0){
\begin{tikzcd}
E(G) \ar[d, "\gamma^G" left] \ar[r, "h" above] &
E(H) \ar[d, "\gamma^H" right] \\
C^0(I, M) \ar[r, "f_*" below] &
C^0(I, N)
\end{tikzcd}
};
\end{tikzpicture}\]
where $f_*$ is the post-composition map $g \mapsto f \circ g$, and the
homomorphism $h$ is viewed as the function it induces on edge sets.
\end{defn}

\begin{defn}[Transport Homomorphism]
A geometric pretransport homomorphism is called a transport homomorphism.
\end{defn}

Pasting commutative squares as the one above along the $\gamma$ sides, we
observe that geometric homomorphisms compose associatively. Taking
$h$ and $f$ as the identity maps yields an identity morphism of geometric
graphs. Similarly, transport homomorphism also compose associatively and have
units. Furthermore, transport graphs inherit the disjoint union from their
underlying expression graphs and manifolds. For their gluing, we require a
notion of gluing geometric realizations.

We first observe that if $G$ and $H$ are transport graphs with a transport
isomorphism $S(H) \to T(G)$, it is unique by \ref{cor:expiso_unique}, since
there is only one map of path sets --- the empty map which always makes the
rquired diagram commute. We then have the following basic fact.

\begin{cor}
Let $G$ and $H$ be transport graphs realized in manifolds
$M$ and $N$ with geometric realizations $\gamma^G$ and
$\gamma^H$ respectively such that $\psi_{G, H} : S(H) \to T(G)$ is a unique
pretransport homomorphism, and $M$ and $N$ are ``smoothly gluable'' at some
part. Then, there exists a geometric realization
\[
  \gamma^H * \gamma^G : E(H * G) \to C^0(I, N * M)
\]
of the pretransport graph $H * G$ in $N * M$.
\end{cor}
\begin{proof}
We can define $\gamma^H * \gamma^G$ piecewise.
\end{proof}

\begin{defn}[Gluing Transport Graphs]
For transport graphs $(G, \gamma^G)$ and $(H, \gamma^H)$ with $S(H) \cong T(G)$
with a transport homomorphism, we define their gluing in the obvious way:
\[
  (H, \gamma^H) * (G, \gamma^G) := (H * G, \gamma^H * \gamma^G)
\]
\end{defn}

This is enough structure to develop a simple system for doing algebra using
paths on a manifold. In particular, it yields at least two monoidal double
categories of transport graphs realized in manifolds. We develop one of these
notions of TQFTs next.

\subsection{Single Manifold TQFT}\label{subsec:sing_man_tqft}

Notice that the theory we developed so far transfers verbatim from the smooth
setting to the setting of complex bundles in that we can replace all instances
of $\R$ with $\C$, smooth bundles and $\R$--linear connections with smooth
manifolds, complex bundles and $\C$--linear connections on complex bundles
respectively. From this point onwards, when we say manifold, bundle, connection,
etc., we will mean these in either of the two settings. We will write $\K$ to
mean either of $\R$ or $\C$. With this convention, we consider the following
data for a monoidal double category:

\begin{enmrt}
\li Object category: objects are totally ordered, $2$--coloured, finite sets;
morphisms are order-preserving, colour-preserving (unique) bijections
$V \stackrel{!}{\longleftrightarrow} V'$

\li Morphism category: objects (horizontal $1$--morphisms) are transport graphs
$(G, \gamma)$, where for a fixed manifold $M$, a fixed bundle $\pi : E \to M$,
and a fixed connection $\nabla$ on $\pi$, $G$ is a pretransport graph with
$\gamma$ a geometric realization of $G$ in $M$; morphisms are tuples
\[
  (f_0, f_1, h) : (G_1, \gamma^1) \to (G_2, \gamma^2)
\]
where $(f_0, f_1)$ is an automorphism of the connection $\nabla$ and
$(f_0, h)$ is a transport isomorphism $(G_1, \gamma^1) \to (G_2, \gamma^2)$
\footnote{It is possible that this condition forces $(f_0, f_1) = (\id, \id)$.}

\li Source functor: $S : G \mapsto S(G)$; for a $2$--morphism $(f_0, f_1, h)$,
$S(f_0, f_1, h)$ is the unique order-preserving bijection
$S(\dom h) \to S(\codom h)$

\li Target functor: $T : G \mapsto T(G)$, defined similarly as $S$

\li Unit functor: $U : V \mapsto V$; $U : f \mapsto f$ --- each finite,
totally-ordered, $2$--coloured set is a transport graph with no edges and
order- and colour- preserving bijections between finite sets are unique
transport isomorphisms

\li Horizontal composition:
$(G_2, \gamma^2) * (G_1, \gamma^1) = (G_2 * G_1, \gamma^2 * \gamma^1)$

\li Horizontal associators: inherited from the categories of sets and manifolds

\li Horizontal unitors: inherited like associators

\li Monoidal product: disjoint union

\li Monoidal unit: empty set for object category, empty graph for morphism
category
\end{enmrt}

\begin{defn}[Double Category of Transport Graphs in a Manifold]
The above data defines the double category of transport graphs in $M$ and we
denote it as $\TG(M)$. We denote the object category as $\TG(M)_0$ and the
morphism category as $\TG(M)_1$.
\end{defn}

We then consider a double functor defined as follows.

\begin{defn}[Parallel Transport Calculus over a Manifold]\label{defn:sing_man_tqft}
For a finite, totally-ordered, $2$--coloured set
$V = \set{v_1, \dots, v_n}$ in $\TG(M)_0$, we set
\[
  F(V) := \bigotimes_{i = 1}^{n} c(v_i)
\]
where $c(v) = A$ if $v$ is blue and $c(v) = \K$ if $v$ is green,
for some $\K$--algebra $A$.

For every unique order- and colour-preserving bijection $f : V \to V'$ in
$\TG(M)_0$, we set $F(f) := \id_{F(V)} = \id_{F(V')}$.

For a transport graph $(G, \gamma)$ in $M$ --- an object in $\TG(M)_1$ --- and an
edge $(u, v) \in G$, we denote $\nabla^{\gamma_{u, v}}$ to be the linear map
$A \to A$ obtained by parallel transport along $\gamma_{u, v}$, with respect to
$\nabla$. Fixing some element $a_{u, v} \in A$, we then define:
\[
  F(u, v) := \begin{cases}
    \id_A
      & u = v \text{ is blue} \\
    \id_{\K} 
      & u \text{ is green and } v \text{ is green} \\
    \nabla^{\gamma_{u, v}}
      & u \text{ is blue and } v \text{ is blue} \\
    1 \mapsto \nabla^{\gamma_{u, v}}(a_{u, v})
      & u \text{ is green and } v \text{ is blue} \\
    \text{trace} \circ \nabla^{\gamma_{u, v}}
      & u \text{ is blue and } v \text{ is green}
  \end{cases}
\]
We then obtain a linear map:
\[
  F(G, \gamma) := \Exp{G}[F(u, v)]
               :  F(S(G)) \to F(T(G))
\]
For a $2$--morphism
$(f_0, f_1, h) : (G_1, \gamma_1) \to (G_2, \gamma_2)$ in $\TG(M)_1$, we consider
the path $(r_t, s_t)$ in the isomorphism group of the connection $\nabla$ such
that $(r_0, s_0) = (\id_M, \id_E)$ and $(r_1, s_1) = (f_0, f_1)$. We then have a
smoothly varying family of functions $s_t\gamma : E(G) \to C^0(I, M)$
where $(s_t\gamma)_{u, v} = s_t \circ \gamma_{u, v}$. This yields a smoothly
varying family of linear maps, which we write as:
\[
  F(f_0, f_1, h) := \Exp{G}[s_t\gamma]
    : F(S(G)) \to F(T(G)), t \in [0, 1]
\]
We call $F$ a parallel transport calculus on $M$.
\end{defn}

\begin{rmk}\label{rmk:any_vect_space}
We note that so far we have considered $\cwedge$ and $\cvee$ to mean the product
and its dual for some algebra $A$. This is not strictly necessary. We could, in
principle, choose, for each instance of $\cwedge$ in an expression of a
transport graph, a distinct map $A \tensor A \to A$ when defining a parallel
transport calculus as above, as long as it does not disturb the monoidal double
functoriality of $F$. In fact, we could replace $A$ with an arbitrary vector
space $V$ and work with arbitrary linear maps $V \tensor V \to V$ and
$V \to V \tensor V$ in replacing $\cwedge$ and $\cvee$ in expressions arising
from transport graphs.
\end{rmk}

We note that the definition of $F$ does not specify the codomain. We will define
the codomain double category in the next section along with a notion of TQFTs
based on transport graphs in cobordisms equipped with connections.

We also note that there is some redundancy in this setup. The empty graph, the
single green vertex and paths with only green vertices are not the same
horizontal $1$--morphism but they are all morphisms
$\varnothing \to \varnothing$ and map to the identity morphism of $\K$ under
$F$. This differs from $\Cob_2$ and $2\Thick$ in that there are unique
morphisms $\varnothing \to \varnothing$ in these categories as well as their
corresponding double categories that map to the identity morphism of $\K$ under
a usual TQFT.


\section{Parallel Transport Calculus}

In definition \ref{defn:sing_man_tqft}, we associated to each object a tensor
power of some algebra $A$; to each vertical $1$--morphism, the identity on $A$;
to each horizontal $1$--morphism, a linear map from a tensor power of $A$ to
another; to each $2$--morphism, a smooth family of linear maps. This data
suggests a codomain double category of our single manifold TQFT, which we define
next.

\subsection{Double Category of Finite-Dimensional Vector Spaces}

We consider as the object or $0$--morphism category of the codomain double
category the monoidal category of finite-dimensional, real (or complex) vector
spaces $\FVect_{\K}$ for $\K = \R$ or $\C$. We then notice a categorification
of the collection of morphisms of this category as follows.

We consider the collection of all morphisms of $\FVect_{\K}$,
\[
  \s{L} := \set[f]{f \in \Hom_{\FVect_{\K}}(U, V), U, V \in \Ob \FVect_{\K}}
\]
as the object collection of the horizontal $1$--morphism category of our double
category. For each such map, we choose some $p \times q$ matrix representation
$[a_{ij}]$ which yields a mapping $\iota : \s{L} \to \M_{\N}(\C)$, where
$\M_{\N}(\C)$ is the set of infinite complex matrices indexed by $\N \times \N$,
defined by:
\[
  \iota(a)_{ij} = \begin{cases}
    a_{ij} & 1 \leq i \leq p, 1 \leq j \leq q \\
    0      & \text{otherwise}
  \end{cases}
\]

This provides a notion of $2$--morphism for our double category under
construction. We notice that the Banach space $\s{B}$ of bounded operators on
the Hilbert space space $\ell^2(\N)$ can be seen as a subset of $\M_{\N}$, such
that $\s{L} \subset \s{B} \subset \M_{\N}$. It is well-known that $\s{B}$, being
a Banach space, is a simply-connected topological space. Now, for $a$ and $a'$
in $\s{L}$, we define a $2$--morphism to be a homomotpy class of paths $\alpha$
from $\iota(a)$ to $\iota(a')$ in $\s{B}$, which we denote as
$\alpha : a \To a'$. By the fact that the fundamental groupoid
$\Pi_1(\s{B})$ is a category, our $2$--morphisms have a strictly associative and
unital composition.  We note that the fundamental groupoid $\Pi_1(\s{B})$ is not
our morphism category --- it only supplies morphisms for $\s{L}$. An instance of
why the distinction is important is that a single morphism in $\Pi_1(\s{B})$
might represent morphisms between two different pairs of objects in $\s{L}$,
depending on the chosen matrix representations.

It is worthwhile observing the action of composition and tensor products of
linear maps on $2$--morphisms. We first notice that $\iota$ can be defined so that
it is multiplicative in two ways:
\[
  \iota(b \circ a) := \iota(b) \cdot \iota(a)
\]
where the right-hand-side product is the matrix product in $\M_{\N}(\C)$, and
\[
  \iota(a \tensor b) := \iota(a) \tensor \iota(b)
\]
where the $\tensor$ on the right is given by the Kronecker product. Now,
consider pairs of homotopic paths
$\alpha_1, \alpha_2 : a \To a'$ and
$\beta_1, \beta_2 : b \To b'$, where $b, a$ and $b', a'$ are
composable pairs of linear maps. We consider the pointwise composites:
\begin{equation}\label{eqn:pointwise_comp}
  (\beta_i \circ \alpha_i)(t) := \beta_i(t) \circ \alpha_i(t), t \in [0, 1],
    i \in \set{1, 2}
\end{equation}
Now, the $\beta_i \circ \alpha_i$ are clearly paths
$b \circ a \To b' \circ a'$ in $\s{B}$ --- we wish to show that
they are homotopic, making the operation well-defined on homomotopy classes of
paths. It suffices to observe that $\s{B}$ is simply connected, so that there is
exactly one class of homotopic paths between two points in $\Pi_1(\s{B})$.

Consider again elements
$a : U \to V, a' : U' \to V', b : X \to Y, b' : X' \to Y'$ of $\s{L}$. We define
\[
  n_x := \dim \dom x, m_x := \dim \codom x, x \in \set{a, a', b, b'}
\]
and
\[
  N_{x} = N_{x'} := \max\set{n_x, n_{x'}},
  M_{x} = M_{x'} := \max\set{m_x, m_{x'}}, x \in \set{a, b}
\]
We then have matrix representations of each $x \in \set{a, a', b, b'}$:
\[
  \iota'(x) := [\iota(x)_{ij}] \in \M_{M_x \times N_x}(\C)
\]
Now, $\M_{M_x \times N_x}(\C)$, being simply connected, has a path
$\gamma_x$ from $\iota'(x)$ to $\iota'(x')$ for each $x \in \set{a, b}$. In
fact, using the inclusion $\M_{M_x \times N_x}(\C) \hto \s{B}$ induced by
$\iota$, each $\gamma_x$ yields a path $\wh{\gamma_x}$ in $\s{B}$ from
$\iota(x)$ to $\iota(x')$, with its image contained in $\s{L}$. Since $\s{B}$ is
simply connected, every path in $\s{B}$ from $\iota(x)$ to $\iota(x')$ is
homotopic to $\wh{\gamma_x}$. Furthermore,
$\wh{\gamma_a}(t) \tensor \wh{\gamma_b}(t)$ is also a path from
$\iota(a \tensor a') = \iota(a) \tensor \iota(a')$ to
$\iota(b \tensor b') = \iota(b) \tensor \iota(b')$ in $\s{B}$, to which all
other paths with the same endpoints are homotopic. This shows that
\[
  (\wh{\gamma_a} \tensor \wh{\gamma_b})(t) :=
    \wh{\gamma_a}(t) \tensor \wh{\gamma_b}(t)
\]
is well-defined on homotopy classes of paths in $\Pi_1(\s{B})$. For
associativity of $\tensor$, we observe that
\[
  \iota(a \tensor (b \tensor c))
    = \iota(a) \tensor \iota(b) \tensor \iota(c)
    = \iota((a \tensor b) \tensor c)
\]
so that the constant path on $\iota(a) \tensor \iota(b) \tensor \iota(c)$
functions as the associator
\[
  a \tensor (b \tensor c) \to (a \tensor b) \tensor c
\]
For unitality of $\tensor$, we take the matrix $1_{\tensor} \in \s{B}$ whose
$(i, j)$ entry is $1$ if $i = j = 1$ and is $0$ otherwise, and then we observe:
\[
  \iota(a \tensor \id_{\K}) = \iota(a) \tensor 1_{\tensor} = \iota(a)
    = \iota(\id_{\K} \tensor a)
\]
so that the constant path on $\iota(a)$ functions as a left and right unitor.

We then take horizontal composition to be given by composition of linear maps
which is strictly associative and unital, with coherence following from that in
the category of vector spaces. The source and target functors are obvious --- we
send each linear map to its domain and codomain respectively. The unit functor
is also obvious --- we send each object to its identity linear map.

We compile these results into the following definition:
\begin{defn}[Double Category of Finite Dimensional Vector Spaces]
The following data form a monoidal double category:
\begin{enmrt}
\li Object category: $\FVect_{\K}$
\li Morphism category: $\s{L}$
\li Source functor: $\dom : (f : X \to Y) \mapsto X$
\li Target functor: $\codom : (f : X \to Y) \mapsto Y$
\li Unit functor: $V \mapsto \id_V$
\li Horizontal composition: $(g, f) \mapsto g \circ f$
\li Horizontal composition associator: constant path on
$\iota(a \circ b \circ c)$
\li Horizontal composition unitor: constant path on $\iota(\id_V)$, for a vector
space $V$
\li Monoidal product: $\tensor$ in appropriate contexts defined above
\li Monoidal unit: $\K$ for the object category and $\id_{\K}$ for the morphism
category
\li Monoidal associators: constant path on
$\iota(a) \tensor \iota(b) \tensor \iota(c)$ as an associator
\[
  a \tensor (b \tensor c) \to (a \tensor b) \tensor c
\]
\li Monoidal unitor: constant path on $\iota(\id_{\K}) = 1_{\tensor}$
\end{enmrt}
This is called the monoidal double category of finite dimensional $\K$--vector
spaces and is denoted $\FFVect_{\K}$.
\end{defn}

We finally note that if the choice of matrix representations poses foundational
problems, we can easily switch to the skeleton of $\FVect_{\K}$ consisting of
the spaces $\K^n$ for all $n \in \N$.

\subsection{Thick Tangles and Transport Graphs}

Having established a codomain double category, we look towards extending our
notion of TQFTs based on transport graphs in a single manifold to transport
graphs in cobordisms equipped with connections. For simplicity, we only consider
thick tangles equipped with connections, at the moment. Recall that a thick
tangle $X \to Y$ is a smooth surfaces $M$ with boundary $W_0 \amalg W_1$ with
equipped with smooth maps $a_M : X \to M, b_M : Y \to M$ such that $a_M$ and
$b_M$ are diffeomorphisms onto $W_0$ and $W_1$ respectively, and an embedding
$d_M : M \to \R \times [0, 1]$ (satisfying some additional properties, which we
will not need at the moment).

We consider the generating thick tangles equipped with transport graphs.
The following examples of transport graphs in the pair-of-pants, cap and their
duals gives us an idea of the structures we are dealing with:
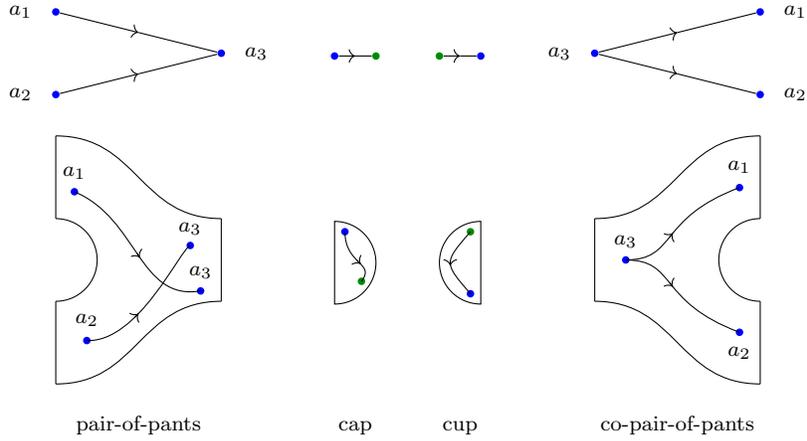
\[
\begin{tikzpicture}[scale=0.55]
\pants{0, 0}
\lblvert{0, -4}{pants}{\footnotesize pair-of-pants}
\colvert{blue}{-2, 6}{a1}
\lblvert{-2.85, 6}{a1l}{\footnotesize $a_1$}
\colvert{blue}{-2, 4}{a2}
\lblvert{-2.85, 4}{a2l}{\footnotesize $a_2$}
\colvert{blue}{2, 5}{a3}
\lblvert{2.85, 5}{a3l}{\footnotesize $a_3$}
\midarrow{a1}{a3}
\midarrow{a2}{a3}
\colvert{blue}{-1.55, 1.65}{a1p}
\lblvert{-1.55, 2.1}{a1pl}{\footnotesize $a_1$}
\colvert{blue}{1.5, -0.75}{a3p1}
\lblvert{1.5, -0.3}{a3p1l}{\footnotesize $a_3$}
\midarrowc{a1p}{0, 1}{0, -1}{a3p1};
\colvert{blue}{-1.25, -1.95}{a2p}
\lblvert{-1.25, -1.45}{a2pl}{\footnotesize $a_2$}
\colvert{blue}{1.25, 0.35}{a3p2}
\lblvert{1.25, 0.75}{a3p2l}{\footnotesize $a_3$}
\midarrowc[0.35]{a2p}{0, -2}{1, 0.25}{a3p2}
\end{tikzpicture}
\qquad
\begin{tikzpicture}[scale=0.55]
\capcob{0, 0}
\lblvert{-1.5, -4}{caplbl}{\footnotesize cap}
\colvert{blue}{-2, 5}{a1}
\colvert{green!55!black}{-1, 5}{a2}
\midarrow{a1}{a2}
\colvert{blue}{-1.75, 0.75}{a1p}
\colvert{green!55!black}{-1.35, -0.45}{a2p}
\midarrowc{a1p}{-1.75, -0.05}{-1, 0}{a2p}
\end{tikzpicture}
\qquad
\begin{tikzpicture}[scale=0.55]
\cupcob{0, 0}
\lblvert{1.5, -4}{caplbl}{\footnotesize cup}
\colvert{green!55!black}{1, 5}{a1}
\colvert{blue}{2, 5}{a2}
\midarrow{a1}{a2}
\colvert{green!55!black}{1.75, 0.75}{a1p}
\colvert{blue}{1.75, -0.75}{a2p}
\midarrowc{a1p}{1.1, 0}{1.1, 0}{a2p}
\end{tikzpicture}
\qquad
\begin{tikzpicture}[scale=0.55]
\copants{0, 0}
\lblvert{0, -4}{copants}{\footnotesize co-pair-of-pants}
\colvert{blue}{2, 6}{a1}
\lblvert{2.85, 6}{a1l}{\footnotesize $a_1$}
\colvert{blue}{2, 4}{a2}
\lblvert{2.85, 4}{a2l}{\footnotesize $a_2$}
\colvert{blue}{-2, 5}{a3}
\lblvert{-2.85, 5}{a3l}{\footnotesize $a_3$}
\midarrow{a3}{a1}
\midarrow{a3}{a2}
\colvert{blue}{-1.25, 0}{a3p}
\lblvert{-1.25, 0.45}{a3pl}{\footnotesize $a_3$}
\colvert{blue}{1.5, 1.75}{a1p}
\lblvert{1.5, 2.25}{a1pl}{\footnotesize $a_1$}
\midarrowc{a3p}{0, 0}{-0.5, 1}{a1p}
\colvert{blue}{1.5, -1.75}{a2p}
\lblvert{1.5, -2.25}{a2pl}{\footnotesize $a_2$}
\midarrowc{a3p}{0, 0}{-0.5, -1}{a2p}
\end{tikzpicture}
\]
We recall that edges that share an end-point need not map to paths that share an
end-point, as we see in the left diagram. At the same time, paths are allowed to
intersect. We also note that we have chosen the pretransport graphs so as to
match their sources and targets with the sources and targets of their realizing
cobordisms. We now turn our attention to the cylinder. We observe that the
cylinder is a cobordism $I \to I$. On the transport graph side, morphisms from a
single blue vertex to another can be any path with blue end-points including the
path consisting of a single blue vertex. However, the gluing unit for the single
blue vertex, on either side, is the single blue vertex itself. Hence, we will
consider the following transport graphs in the cylinder:
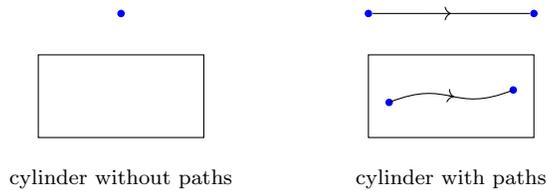
\[
\begin{tikzpicture}[scale=0.55]
\idcob{0, 0}
\colvert{blue}{0, 2}{a}
\lblvert{0, -2}{lbl}{\footnotesize cylinder without paths}
\end{tikzpicture}
\qquad \qquad
\begin{tikzpicture}[scale=0.55]
\idcob{0, 0}
\lblvert{0, -2}{lbl}{\footnotesize cylinder with paths}
\colvert{blue}{-2, 2}{a1}
\colvert{blue}{2, 2}{a2}
\midarrow{a1}{a2}
\colvert{blue}{-1.5, -0.15}{a1p}
\colvert{blue}{1.5, 0.15}{a2p}
\midarrowc{a1p}{0, 0.45}{0, -0.45}{a2p}
\end{tikzpicture}
\]

\begin{exm}
We take the example from our first description of the double categorical
approach and adapt it to this setting. The following is one possible diagram:
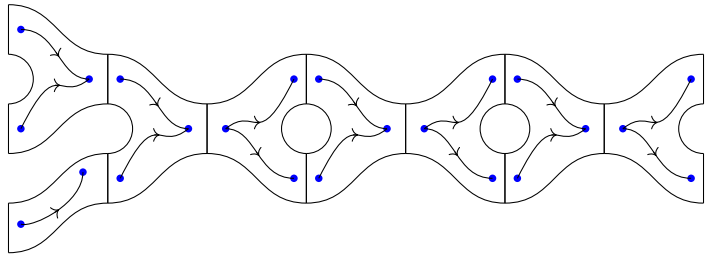
\[\begin{tikzpicture}[scale=0.33]

\idcobup{1, 0}
\colvert{blue}{-0.5, -1.85}{s}
\colvert{blue}{2, 0.25}{t}
\midarrowc{s}{0, -2}{2, -1}{t}

\coordinate (cntr) at (1, 4);
\pants{cntr}
\colvert{blue}{$(cntr) + (-1.5, 2)$}{s1}
\colvert{blue}{$(cntr) + (-1.5, -2)$}{s2}
\colvert{blue}{$(cntr) + (1.25, 0)$}{t}
\midarrowc{s1}{$(cntr) + (0, 2)$}{$(cntr) + (0, 0)$}{t}
\midarrowc{s2}{$(cntr) + (0, 1)$}{$(cntr) + (0, -1)$}{t}

\coordinate (cntr) at (5, 2);
\pants{cntr}
\colvert{blue}{$(cntr) + (-1.5, 2)$}{s1}
\colvert{blue}{$(cntr) + (-1.5, -2)$}{s2}
\colvert{blue}{$(cntr) + (1.25, 0)$}{t}
\midarrowc{s1}{$(cntr) + (0, 2)$}{$(cntr) + (0, 0)$}{t}
\midarrowc{s2}{$(cntr) + (0, 1)$}{$(cntr) + (0, -1)$}{t}

\coordinate (cntr) at (9, 2);
\begin{scope}[rotate around={180:(cntr)}]
\pants{cntr}
\colvert{blue}{$(cntr) + (-1.5, 2)$}{s1}
\colvert{blue}{$(cntr) + (-1.5, -2)$}{s2}
\colvert{blue}{$(cntr) + (1.25, 0)$}{t}
\midarrowc{s1}{$(cntr) + (0, 2)$}{$(cntr) + (0, 0)$}{t}
\midarrowc{s2}{$(cntr) + (0, 1)$}{$(cntr) + (0, -1)$}{t}
\end{scope}

\coordinate (cntr) at (13, 2);
\pants{cntr}
\colvert{blue}{$(cntr) + (-1.5, 2)$}{s1}
\colvert{blue}{$(cntr) + (-1.5, -2)$}{s2}
\colvert{blue}{$(cntr) + (1.25, 0)$}{t}
\midarrowc{s1}{$(cntr) + (0, 2)$}{$(cntr) + (0, 0)$}{t}
\midarrowc{s2}{$(cntr) + (0, 1)$}{$(cntr) + (0, -1)$}{t}

\coordinate (cntr) at (17, 2);
\begin{scope}[rotate around={180:(cntr)}]
\pants{cntr}
\colvert{blue}{$(cntr) + (-1.5, 2)$}{s1}
\colvert{blue}{$(cntr) + (-1.5, -2)$}{s2}
\colvert{blue}{$(cntr) + (1.25, 0)$}{t}
\midarrowc{s1}{$(cntr) + (0, 2)$}{$(cntr) + (0, 0)$}{t}
\midarrowc{s2}{$(cntr) + (0, 1)$}{$(cntr) + (0, -1)$}{t}
\end{scope}

\coordinate (cntr) at (21, 2);
\pants{cntr}
\colvert{blue}{$(cntr) + (-1.5, 2)$}{s1}
\colvert{blue}{$(cntr) + (-1.5, -2)$}{s2}
\colvert{blue}{$(cntr) + (1.25, 0)$}{t}
\midarrowc{s1}{$(cntr) + (0, 2)$}{$(cntr) + (0, 0)$}{t}
\midarrowc{s2}{$(cntr) + (0, 1)$}{$(cntr) + (0, -1)$}{t}

\coordinate (cntr) at (25, 2);
\begin{scope}[rotate around={180:(cntr)}]
\pants{cntr}
\colvert{blue}{$(cntr) + (-1.5, 2)$}{s1}
\colvert{blue}{$(cntr) + (-1.5, -2)$}{s2}
\colvert{blue}{$(cntr) + (1.25, 0)$}{t}
\midarrowc{s1}{$(cntr) + (0, 2)$}{$(cntr) + (0, 0)$}{t}
\midarrowc{s2}{$(cntr) + (0, 1)$}{$(cntr) + (0, -1)$}{t}
\end{scope}

\end{tikzpicture}\]
\end{exm}

\begin{exm}
We need not consider transport graphs that match the pair-of-pants or the
cylinder (or their duals) exactly. For instance, we could consider the following
graph:
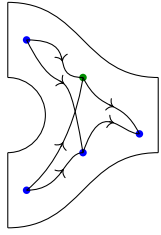
\[\begin{tikzpicture}[scale=0.5]
\coordinate (cntr) at (1, 4);
\pants{cntr}
\colvert{blue}{$(cntr) + (-1.5, 2)$}{s1}
\colvert{blue}{$(cntr) + (-1.5, -2)$}{s2}
\colvert{green!55!black}{$(cntr) + (0, 1)$}{t1}
\colvert{blue}{$(cntr) + (0, -1)$}{t2}
\colvert{blue}{$(cntr) + (1.5, -0.5)$}{t3}
\midarrowc{s1}{$(cntr) + (0, 2)$}{$(cntr) + (-1, 1)$}{t1}
\midarrowc[0.33]{s2}{$(cntr) + (-0.25, -1)$}{$(cntr) + (0, 1)$}{t1}
\midarrowc{s1}{$(cntr) + (-0.5, 0)$}{$(cntr) + (-0.5, 2)$}{t2}
\midarrowc{s2}{$(cntr) + (-0.5, -2)$}{$(cntr) + (-0.5, -1)$}{t2}
\midarrowc{t1}{$(cntr) + (0.5, 0)$}{$(cntr) + (1, 0.5)$}{t3}
\midarrowc{t2}{$(cntr) + (0.5, 0.5)$}{$(cntr) + (1, -0.5)$}{t3}
\end{tikzpicture}\]
Notice that the source and target of the graph matches the source and target of
the pair-of-pants event though the graph is not a pair-of-pants graph. Also
notice that this graph has an internal green vertex.
\end{exm}

From these examples, we are motiviated to note the following definition.
\begin{defn}
Given a $2$--dimensional thick tangle, consider a transport graph in
the tangle such that the source of the graph has the same number of vertices
as the number of boundary components of the source of the tangle and the
colouring and ordering of the source of the graph is such that a blue vertex
corresponds to an interval boundary component and a green vertex corresponds to
an empty boundary component. We further assume that the analogous statement
holds for the targets of the graph and the tangle. We then call the given
transport graph admissible for the given tangle.
\end{defn}

An immediate corollary of this definition is:
\begin{cor}
Admissible transport graphs in gluable thick tangles are gluable.
\end{cor}

The examples we have seen so far are all admissible. We also show graphs that
are not admissible in the following example.

\begin{exm}
The following are not admissible transport graphs:
\[\begin{tikzpicture}[scale=0.33]

\coordinate (cntr) at (1, 4);
\pants{cntr}
\colvert{blue}{$(cntr) + (-1.5, 2)$}{s}
\colvert{blue}{$(cntr) + (1.5, 0)$}{t}
\midarrow{s}{t}
\lblvert{$(cntr) + (0, -5)$}{lbl}{\small not enough}
\lblvert{$(cntr) + (0, -6)$}{lbl}{\small source vertices}

\coordinate (cntr) at (11, 4);
\idcob{cntr}
\colvert{green!55!black}{$(cntr) + (-1.5, 0)$}{s}
\colvert{blue}{$(cntr) + (1.5, 0)$}{t}
\midarrow{s}{t}
\lblvert{$(cntr) + (0, -5)$}{lbl}{\small source should}
\lblvert{$(cntr) + (0, -6)$}{lbl}{\small be blue}

\coordinate (cntr) at (21, 4);
\copants{cntr}
\capcob{$(cntr) + (4, 2)$}
\colvert{blue}{$(cntr) + (-1.5, 0)$}{s}
\colvert{blue}{$(cntr) + (2.5, 2)$}{t1}
\colvert{blue}{$(cntr) + (1.5, -2)$}{t2}
\midarrow{s}{t1}
\midarrow{s}{t2}
\lblvert{$(cntr) + (0, -5)$}{lbl}{\small top traget}
\lblvert{$(cntr) + (0, -6)$}{lbl}{\small should be}
\lblvert{$(cntr) + (0, -7.25)$}{lbl}{\small green}

\end{tikzpicture}\]
\end{exm}

So far, we have only shown the diagrams of surfaces with paths in these examples
but what we need to work with are surfaces equipped with bundles, connections
and transport graphs. We will now construct a monoidal double category
consisting of these structures.

Consider the monoidal double category $\CConn^V_{\DThick}$ of gluable
connections on gluable $V$--fibred (smooth or complex) bundles on
$2$--dimensional thick tangles.
For each horizontal $1$--morphism (bundle with connection) in this category, we
take all admissible transport graphs in the base of the bundle such that all
paths consist only of points internal to the base space and away from the
boundary collar over which the bundle has been made trivial.

For each gluable bundle equipped with a gluable connection and an admissible
transport graph in this manner, we take its source to be the source of the
bundle in $\CConn^V_{\DThick}$ along with the source of the transport graph.
Targets are defined similarly. The gluing units are the disjoint unions of
cylinders of the following form shown before:
\[\begin{tikzpicture}[scale=0.55]
\idcob{0, 0}
\colvert{blue}{0, 2}{a}
\lblvert{0, -2}{lbl}{\footnotesize cylinder without paths}
\end{tikzpicture}\]

Naturally, the object category consists of the sources and targets of the
horizontal $1$--morphisms and transport isomorphisms between them that are also
connection isomorphisms --- similar to the monoidal double category $\TG(M)$ of
transport graphs in a manifold $M$ defined in subsection
\ref{subsec:sing_man_tqft}. The morphism category consists of the horizontal
$1$--morphisms along with transport isomorphisms that are also connection
isomorphisms. Finally, monoidal structure is given by disjoint union, as
expected.

Having the developed the basic idea of such a monoidal double category, we avoid
going into further detail because it does not provide any additional insight. We
simply note that the structure outlined here is a monoidal double category that
can serve as the domain for a notion of (double) functorial quantum field
theory built on parallel tranport calculi over cobordisms.
Hence, we end this subsection with the following definition.

\begin{defn}
The monoidal double category outlined above is called the double category of
transport graphs in $2$--dimensional thick tangles over $V$ and is
denoted $\TG\br{\CConn^V_{\DThick}}$.
\end{defn}

\subsection{Parallel Transport Calculus}

We are now equipped with all the machinery to define our desired modification
of topological quantum field theory.

\begin{defn}[Parallel Transport Calculus]
Let $A$ be a $\K$--algebra for $\K = \R$ or $\C$. Then, we define the data of a
monoidal double functor
\[
  F : \TG\br{\CConn^V_{\DThick}} \to \FFVect_{\K}
\]
Each horizontal $1$--morphism in the domain is identical to one in the
double category of transport graphs in a manifold. $F$ is thus defined on
horizontal $1$--morphisms identically to the functor in definition
\ref{defn:sing_man_tqft}.

We observe that each object can be reduced to a source or target of a
pretransport graph since, by definition the copies of $I$ and $\varnothing$ are
matched up with blue and green vertices respectively. This allows us to define
$F$ on objects identically to \ref{defn:sing_man_tqft} again. It is then easy to
see that the action of $F$ on vertical $1$--morphisms and $2$--morphisms can be
adapted similarly.

A monoidal double functor $F$ defined in this way is called a parallel transport
calculus (on $2$--dimensional thick tangles or $\DThick$).
\end{defn}

\begin{rmk}
The ending parenthetical remark in the previous definition suggests that we can
easily consider such parallel transport calculi over other cobordism categories
but we will not pursue this idea for now.
\end{rmk}

We have not yet discussed if enough useful elements of the algebra $A$ can be
accessed with a parallel transport calculus of this form.
After all, this was the original
issue with $1$--categorical TQFTs. We will not treat this issue in full in this
paper but we will note that the machinery we have developed so far puts no
serious restrictions on the bundles or connections we can choose over our
manifolds. What we mean by this is that given an arbitray bundle with a
connection, we can make it gluable by only modifying it in a small collar of the
boundary. The bundle and connection behave as usual over rest of the base
manifold. We hope that this will provide enough structure to ensure that enough
useful algebra elements become accessible with a parallel transport calculus.
Nevertheless, we will make this problem precise for future work.
One of the main questions to be answered here the following:

\begin{qstn}\label{qstn:elem_from_pt}
Given a manifold $M$ and some fixed element $a$ of a complex algebra
$A^{\tensor n}$, are there
\begin{enmrt}
\li an $A$--fibred complex bundle $E \to M$,
\li a complex linear connection $\nabla$ on $E$,
\li a transport graph $G$ in $M$ with $S(G)$ consisting of only green vertices,
$T(G)$ having $n$ blue vertices and with the paths in the geometric realization
of $G$ away from some small neighbourhood of $\partial M$,
\end{enmrt}
such that $a$ can be obtained as a linear map $\C \to A^{\tensor n}$ in
the process described in definition \ref{defn:sing_man_tqft} for horizontal
$1$--morphisms?
\end{qstn}

\begin{qstn}\label{qstn:map_from_pt}
Given a manifold $M$ and some linear map $f : V \to V$, are there
\begin{enmrt}
\li a $V$--fibred complex bundle $E \to M$
\li a complex linear connection $\nabla$ on $E$,
\li a transport graph $G$ in $M$ with $S(G)$ and $T(G)$ both consisting of only
blue vertices, with its paths away from some neighbourhood of $\partial M$, as
before,
\end{enmrt}
such that $f$ can be obtained in the process described in definition
\ref{defn:sing_man_tqft} for horizontal $1$--morphisms?
\end{qstn}

If we can answer these questions in the affirmative for an appreciable
collection of elements $a \in A$, then a parallel transport calculus gives us a
concrete way to perform computations involving the multiplication of $A$ and
automorphisms of $A$ and tensor products of these maps, using the geometry of
manifolds. Hence, we make the following definition:

\begin{defn}
If the answers to question \ref{qstn:elem_from_pt} (or \ref{qstn:map_from_pt})
is yes for some manifold $M$, then we say that $a$ (or $f$) is accessible from
$M$.
\end{defn}

If an element $a$ is accessible from a manifold $M$, it is in the image of a
parallel transport calculus on $M$.
If an element $a$ is accessible from a thick tangle
$M : \varnothing \to I^{\amalg n}$ where the chosen transport graph is
admissible, then $a$ is in the image of a parallel transport calculus on $\DThick$.

\begin{defn}
In the latter case, we say that $a$ is accessible from $\DThick$.
\end{defn}

After this, it is easy to see
that we can compute with accessible elements using the machinery of a parallel
transport calculus. This picture will become clearer as we treat quantum information
and computing in the next subsection.

\subsection{Quantum Computing with Parallel Transport}

Take $A$ to be the complex matrix algebra $\M_2(\C)$ of $2 \times 2$ complex
matrices with the usual multiplication. Let
$\mathcal{U} = \set{u_1, \dots, u_n}$ be a set
of single qubit quantum gates --- unitary matrices --- in $\M_2(\C)$ such that
each $u_k$ is accessible from $\DThick$. Let the thick tangle with a bundle,
a connection and an admissible transport graph that realizes the accessibility
of the $u_k$ be $M_k$, for each $k \in \set{1, \dots, n}$.

We recall that, in the simplest terms, a quantum circuit is a sequence of
composable complex linear unitary maps
$U_i : (\C^2)^{\tensor N} \to (\C^2)^{\tensor N}$, $i = 1, \dots, p$.
For simplicity, we assume that
each $U_i$ arises as a tensor product $\bigotimes_{j = 1}^{N} g_{i, j}$ for
gates $g_{i, j} \in \mathcal{U}$. Since each $g_{i, j}$ is accessible, we have
a $G_{i, j} \in \set[M_k]{k = 1, \dots, n}$ such that $F(G_{i, j}) = g_{i, j}$
for a parallel transport calculus $F$. Then, we have:
\[
 U := U_p \circ \cdots \circ U_1
  = \bigotimes_{j = 1}^{N} g_{p, j} \circ \cdots
    \circ \bigotimes_{j = 1}^{N} g_{1, j}
  = \bigotimes_{j = 1}^{N} (g_{p, j} \circ \cdots \circ g_{1, j})
\]
where $\circ$ is matrix multiplication (composition of linear maps). Consider
the following transport graph in the pair-of-pants:

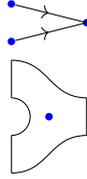
\[\begin{tikzpicture}[scale=0.25]
\pants{1, 0}
\colvert{blue}{-1, 6}{s1}
\colvert{blue}{-1, 4}{s2}
\colvert{blue}{3, 5}{t}
\colvert{blue}{1, 0}{a}
\midarrow{s1}{t}
\midarrow{s2}{t}
\end{tikzpicture}\]
where each edge is geometrically realized as a constant path on a single point
--- shown as a blue dot --- in the pair-of-pants. This graph is admissible and
provides a binary operation on thick tangles with target $I$ in the obvious way.
We denote this operation as $\wedge$. It is then easy to see that
\[
  U = F\br{\coprod_{j = 1}^{N} G_{p, j} \wedge \cdots \wedge G_{1, j}}
\]
We should stress that $\wedge$ is not associative, even up to isomorphism, in
$\DThick$ but its image under $F$ is. The reason for failure of associativity is
that graphs do not have the smooth structure needed for associator isomorphisms
for $\wedge$.

This setup provides a very basic formalism for expressing quantum
circuits in the language of thick tangles equipped with bundles, connections and
transport graphs.
We note, however, that it is not clear what structure plays the role of
qubits or registers in this picture. One easy way to get around this is to
consider the following embedding of $\C^2$ into $\M_2(\C)$:
\[
  \bmat{a \\ b} \mapsto \bmat{a & 0 \\ b & 0}
\]
If we then require that the following matrices are accessible:
\[
  \bmat{1 & 0 \\ 0 & 0} \text{ and } \bmat{0 & 0 \\ 1 & 0}
\]
we can model classical inputs with thick tangles just like gates. Multiplying
inputs with circuits using the pair-of-pants then models the application of
the circuit to the input.

One issue with this approach is that there might be gates acting on more
than one qubit that are not elementary tensor products of single qubit gates.
For instance, the controlled not gate is one such example. In order to obtain
non-elementary tensors, we will require addition of vectors. To capture this
in the language of parallel transport calculi, we need a notion of addition for
cobordisms. We will return to this idea in the next subsection. For now, we
observe another approach to quantum computing using parallel transport calculi.

We now take the fibres of our bundles to be $\C^2$ --- the space where a single
qubit lives (recall \ref{rmk:any_vect_space}).
We view our previous collection of quantum gates
$g_{i, j} \in \M_2(\C)$ as a collection of linear maps
$g_{i, j} : \C^2 \to \C^2$ and assume that the $g_{i, j}$ are accessible from a
collection of $2$--dimensional thick
tangles $G'_{i, j} : I \to I$ (not $\varnothing \to I$, this time) such that
$F(G'_{i, j}) = g_{i, j}$. Then, the quantum circuit $U$ can be expressed as:
\[
  U = F\br{\coprod_{j = 1}^{N} G'_{p, j} * \cdots
           * \coprod_{j = 1}^{N} G'_{1, j}}
\]
recalling that $*$ is gluing of thick tangles. Let $G$ be the input to $F$ in
the above equation.

Notice that single qubit inputs are linear maps $\C \to \C^2$. Hence, we now
assume that linear maps
\[
  \ket{0} = z \mapsto z\bmat{1 \\ 0} \text{ and }
  \ket{1} = z \mapsto z \bmat{0 \\ 1}
\]
are accessible. That is, a quantum register expressed as thick tangles is a
disjoint union of thick tangles (of course, equipped with transport graphs)
$\mathbf{0}, \mathbf{1} : \varnothing \to I$ realizing the accessibility of
$\ket{0}$ and $\ket{1}$ respectively. In this case, we will have
$F(\mathbf{0}) = \ket{0}$ and $F(\mathbf{1}) = \ket{1}$.

It is then easy to see that an application of a quantum circuit to a quantum
register is given by composition of thick tangles. That is, let $R$ be a thick
tangle formed from the disjoint union of a sequence of $\mathbf{0}$ and
$\mathbf{1}$. This represents the input register. We can then glue $R$ on the
source end of $G$ to obtain a thick tangle $H$. Then,
$F(H) = F(G * R) = F(G) \circ F(R)$ is the result of giving the circuit
$F(G)$ the input from the register $F(R)$.

\subsection{Addition of Cobordisms}

We will now to make precise an addition operation for $2$--dimensional thick tangles
mentioned in the previous subsection, so that non-elementary tensors become
accessible with parallel transport calculi.

Consider $2$--dimensional thick tangles $M : I^{\amalg n} \to I^{\amalg m}$ and
$N : I^{\amalg n'} \to I^{\amalg m'}$. Suppose $S_1 \in \set{S(M), S(N)}$ is the
source of either $M$ or $N$ with the most copies of $I$ and $S_0$ is that with
the least copies of $I$. $T_0, T_1 \in \set{T(M), T(N)}$ are defined
analogously. That is, $S_0 = I^{\amalg \min\set{m, m'}}$,
$T_0 = I^{\amalg \min\set{n, n'}}$, $S_1 = I^{\amalg \max\set{m, m'}}$ and
$T_1 = I^{\amalg \max\set{n, n'}}$. Then, $M + N$ is defined to be the shape
obtained by gluing $S_0$ to the first $\min\set{m, m'}$ copies of $I$ in
$S_1$ and $T_0$ to the first $\min\set{n, n'}$ copies of $I$ in $T_1$.

For instance, let $M$ be the pair-of-pants and $N$ the cylinder $I \times I$.
Then, pictorially, $M + N$ would look like:

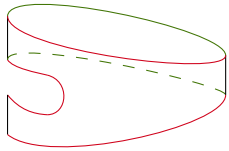
\begin{figure}[H]\label{fig:cobaddexm}
\begin{tikzpicture}[x=0.75pt,y=0.75pt,yscale=-1,xscale=1]
\draw [color={rgb, 255:red, 208; green, 2; blue, 27 }  ,draw opacity=1 ]
(10,62.25) .. controls (12.6,66.54) and (25,68.8) .. (30,70) ;
\draw [color={rgb, 255:red, 208; green, 2; blue, 27 }  ,draw opacity=1 ]
(10,80) .. controls (16.2,87.6) and (23,89.6) .. (30,90) ;
\draw [color={rgb, 255:red, 208; green, 2; blue, 27 }  ,draw opacity=1 ]
(30,70) .. controls (41.8,72.8) and (40.6,90) .. (30,90) ;
\draw (10,40.25) -- (10,62.25) ;
\draw (10,80) -- (10,100) ;
\draw [color={rgb, 255:red, 208; green, 2; blue, 27 }  ,draw opacity=1 ]
(10,100) .. controls (30.6,116.7) and (120.2,96) .. (120,80) ;
\draw [color={rgb, 255:red, 65; green, 117; blue, 5 }  ,draw opacity=1 ]
(10,40.25) .. controls (11,23.2) and (119.8,46.8) .. (120,60) ;
\draw [color={rgb, 255:red, 65; green, 117; blue, 5 }  ,draw opacity=1 ]
[dash pattern={on 4.5pt off 4.5pt}]  (10,62.25) .. controls (9.8,50) and
(119.8,74.4) .. (120,80) ;
\draw    (120,60) -- (120,80) ;
\draw [color={rgb, 255:red, 208; green, 2; blue, 27 }  ,draw opacity=1 ]
(10,40.25) .. controls (10.2,57.6) and (120.6,71.6) .. (120,60) ;
\end{tikzpicture}
\end{figure}
We note that this operation is distinct from both the disjoint union and the
gluing of thick tangles end-to-end. In particular, the results of this operation
can not, in general, be embeded into the infinite strip $\R \times I$. However,
we observe that we can unambiguously define the sources and targets of these
shapes as sources $S_1$ and $T_1$ respectively. We also observe that when
the domains and codomains of the summands are the same, the picture of a sum is
simply a diagram in $\DThick$ involving parallel cobordisms. This motivates us
to name these shapes as follows.

\begin{defn}
For thick tangles $M, N, S_0, T_0, S_1, T_1$ as above, $M + N$ is called a
multitangle $S_1 \to T_1$. In particular, we take the empty manifold to be
the empty sum --- a multitangle between any pair of objects in $\DThick$.
\end{defn}

\begin{rmk}
We observe that this definition carries over as is to transport graphs, bundles
and connections, even though the results of this kind of addition need not be
transport graphs, bundles or connections. In fact, if we take sums of three
cobordisms, then the result is not a manifold, in the usual sense because not
point in the source or target ends has a neighbourhood hormeomorphic to a
Euclidean open set or half-plane. Nevertheless, we will carry on with this
construction, being aware that we might start to lose some of the double
categorical structures.
\end{rmk}

However, we will observe that many of the useful structures in $\DThick$ are not
disturbed by this operation. We first define the starting data of yet another
double category:

\begin{enmrt}
\li Object category: same as $\TG\br{\CConn^V_{\DThick}}$
\li Morphism category: objects are sums of objects in
$\TG\br{\CConn^V_{\DThick}}$, including the empty sum, and morphisms are
piecewise isomorphisms --- that is, they are tuples of $2$--morphisms in
$\TG\br{\CConn^V_{\DThick}}$, one for each summand
\end{enmrt}

We would like to make our structure resemble an additive category. For this, we
will modify the definition of end-to-end gluing or horizontal composition so as
to make it $\Z$--bilinear as follows. First, we will reduce diagrams involving
multibordisms into line diagrams of the following form, for simplicity:
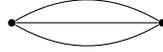
\[\begin{tikzpicture}
\colvert{black}{1, 0}{s}
\colvert{black}{3, 0}{t}
\path
  (s) edge[bend left] (t)
  (s) edge (t)
  (s) edge[bend right] (t)
  ;
\end{tikzpicture}\]
where we shrink sources and targets to single vertices. Now, suppose that we
have thick tangles $M_i : X \to Y, i \in \set{1, \dots, m}$ and
$N_j : Y \to Z, j \in \set{1, \dots, n}$. Then, we can define
\[
  \sum_{j = 1}^n N_j * \sum_{i = 1}^m M_i
  := \sum_{i = 1}^n \sum_{j = 1}^m N_j * M_i
\]
since the relevant composites all exist. For $m = 3, n =2$, in line diagrams,
this equation is:
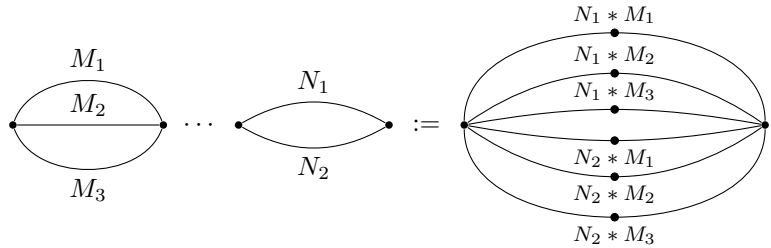
\[\begin{tikzpicture}
\colvert{black}{1, 0}{X}
\colvert{black}{3, 0}{Y}
\lblvert{3.5, 0}{glue}{$\cdots$}
\colvert{black}{4, 0}{YY}
\colvert{black}{6, 0}{Z}
\path
  (X)  edge[bend left=70]   node[above]{$M_1$}   (Y)
  (X)  edge                 node[above]{$M_2$}   (Y)
  (X)  edge[bend right=70]  node[below]{$M_3$}   (Y)
  (YY) edge[bend left]      node[above]{$N_1$}   (Z)
  (YY) edge[bend right]     node[below]{$N_2$}   (Z)
  ;
\lblvert{6.5, 0}{equals}{$:=$}
\colvert{black}{7, 0}{XX}
\colvert{black}{11, 0}{ZZ}
\draw
  (XX) to[bend left=90]
  node[midway, draw=black, fill=black, shape=circle, inner sep=\vertinnersep]{}
  node[midway, above]{{\footnotesize $N_1 * M_1$}}
  (ZZ);
\draw
  (XX) to[bend left=35]
  node[midway, draw=black, fill=black, shape=circle, inner sep=\vertinnersep]{}
  node[midway, above]{{\footnotesize $N_1 * M_2$}}
  (ZZ);
\draw
  (XX) to[bend left=10]
  node[midway, draw=black, fill=black, shape=circle, inner sep=\vertinnersep]{}
  node[midway, above]{{\footnotesize $N_1 * M_3$}}
  (ZZ);
\draw
  (XX) to[bend right=10]
  node[midway, draw=black, fill=black, shape=circle, inner sep=\vertinnersep]{}
  node[midway, below]{{\footnotesize $N_2 * M_1$}}
  (ZZ);
\draw
  (XX) to[bend right=35]
  node[midway, draw=black, fill=black, shape=circle, inner sep=\vertinnersep]{}
  node[midway, below]{{\footnotesize $N_2 * M_2$}}
  (ZZ);
\draw
  (XX) to[bend right=90]
  node[midway, draw=black, fill=black, shape=circle, inner sep=\vertinnersep]{}
  node[midway, below]{{\footnotesize $N_2 * M_3$}}
  (ZZ);
\end{tikzpicture}\]
In other words, during horizontal composition of multitangles, we first undo
the gluing resulting from addition at the composition site, duplicate the
branches as needed and then perform the gluing for horizontal composition. This
new composition operation is easily seen to be associative up to piecewise
isomorphisms. Furthermore, we have
\[
  (N_1 + N_2) * M_1 = (N_1 * M_1) + (N_2 * M_1)
\]
In pictures:
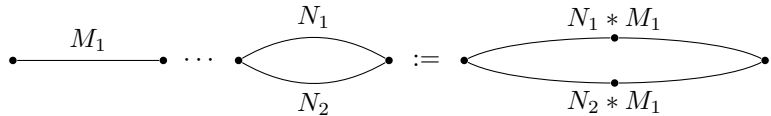
\[\begin{tikzpicture}
\colvert{black}{1, 0}{X}
\colvert{black}{3, 0}{Y}
\lblvert{3.5, 0}{glue}{$\cdots$}
\colvert{black}{4, 0}{YY}
\colvert{black}{6, 0}{Z}
\path
  (X)  edge   node[above]{$M_1$}   (Y)
  (YY) edge[bend left]   node[above]{$N_1$}  (Z)
  (YY) edge[bend right]  node[below]{$N_2$} (Z)
  ;
\lblvert{6.5, 0}{equals}{$:=$}
\colvert{black}{7, 0}{XX}
\colvert{black}{9, 0.30}{YYY}
\colvert{black}{9, -0.30}{YYYY}
\colvert{black}{11, 0}{ZZ}
\draw (XX) .. controls (7.5, 0.25) and (8.5, 0.30) .. (YYY);
\draw (YYY) .. controls (9.5, 0.30) and (10.5, 0.25) .. (ZZ);
\draw (XX) .. controls (7.5, -0.25) and (8.5, -0.30) .. (YYYY);
\draw (YYYY) .. controls (9.5, -0.30) and (10.5, -0.25) .. (ZZ);

\lblvert{9, 0.55}{NM}{$N_1 * M_1$}
\lblvert{9, -0.55}{NNMM}{$N_2 * M_1$}
\end{tikzpicture}\]
If, in addition, $M_1 = Y \times I$, then we easily see that there is a
piecewise isomorphism $(N_1 + N_2) * M_1 \to N_1 + N_2$. It also easy to see for
arbitrarily many $N_1, \dots, N_n$ --- that is, horizontal
composition is right unital up to isomorphism. Left unitality is similar ---
consider $N_1 * (M_1 + M_2)$ and take $N_1 = Y \times I$. We should note that
the addition of (multi-)cobordisms does not have immediate inverses but it is
commutative up to piecewise isomorphism. Hence, the morphism category of
multitangles is an up-to-isomorphism commutative monoid.

So far, we have defined addition on horizontal $1$--morphisms. We will define
addition on objects to fit this picture. Given $I^{\amalg m}$ and
$I^{\amalg m'}$, we define
\[
  I^{\amalg m} + I^{\amalg m'} := I^{\amalg \max\set{m, m'}}
\]
Without strictly verifying (or even defining) axioms further, we propose that
multitangles form a structure akin to an additive double category with
monoidal structure (given by disjoint union) --- in some loose sense, at the very
least. We invite the reader to formulate this notion of double category with
fully defined axioms to make our structure fit the definition.

We then move on to show that with this structure we have a more robust notion of
parallel transport calculus capable of handling non-elementary tensors in the
context of quantum computing. To this end, we define:

\begin{defn}
We denote the ``additive monoidal double category'' of multitangles constructed
so far as $\TG^+(\CConn^V_{\DThick})$.
\end{defn}

\subsection{Quantum Computing Revisited}

Given a parallel transport calculus as defined so far, we can extend the domain of
$F$ to the structure $\TG^+\br{\CConn^V_{\DThick}}$ in an obvious way to obtain
a ``functor'' $F^+$ as follows. $F^+$ is identical to $F$ on the object category
--- there are no issues here since the object category was not modified in
constructing $\TG^+\br{\CConn^V_{\DThick}}$. A horizontal $1$--morphism in this
``double category'', however, is of the form
\[
  M = M_1 + \cdots + M_k
\]
for horizontal $1$--morphisms $M_i, i \in \set{1, \dots, k}$ in
$\TG\br{\CConn^V_{\DThick}}$. $F^+(M)$ is defined to be
\[
  F^+(M) := F(M_1) + \cdots + F(M_k)
\]
where the addition on the right is not well-defined as is. To define this we
observe that each $F^+(M_i)$ is a linear map
$F(I)^{\tensor n_i} \to F(I)^{\tensor m_i}$. Then, we choose a sensible
embedding of each $F(I)^{\tensor n_i}$ in $F(I)^{\tensor \max_i n_i}$ and of
each $F(I)^{\tensor m_i}$ in $F(I)^{\tensor \max_i m_i}$. After this, the
addition is taken within the space of linear maps
$F(I)^{\tensor \max_i n_i} \to F(I)^{\tensor \max_i m_i}$. We also note that if
$M = \varnothing$, then we define:
\[
  F(M)(x) := 0, \forall x
\]

\begin{defn}
An additive parallel transport calculus is an ``additive monoidal double functor''
\[
  F : \TG^+(\CConn^V_{\DThick}) \to \FFVect_{\C}
\]
defined using the above construction.
Accessibility is defined similarly for additive parallel transport calculi.
\end{defn}

Given a collection of $1$--qubit gates, we can take sums of tensor products of
these gates to obtain multi-qubit gates that are not elementary tensors.  Thus,
if we can solve the accessibility problem for $1$--qubit gates in the first
sense of parallel transport calculi, we can express sums of tensor products of
these gates using multitangles. Thus, we have a concrete way to express both
quantum registers and circuits using structures in $\TG^+(\CConn^V_{\DThick})$
which finally yield usual linear algebraic quantum registers and circuits under
additive parallel transport calculi. In fact, both approaches to quantum
computing discussed before can be adapted to this framework. It is then also
easy to adapt this notion back to the single manifold case of
\ref{subsec:sing_man_tqft}.

\subsection{Connections to Operads and PROPs}

Our goal has been to present a framework for quantum
computing in the language of TQFTs and of geometric structures supported on them.
Hence, our constructions are mostly based on combinatorial and geometric data
that can be associated with cobordisms. Nevertheless, one cannot help but notice
the similarity of transport graphs and, more visibly, cobordisms themselves with
operads. As a start, consider the operad of rooted trees whose non-root vertices
are all leaves and where composition is given by gluing the root of a tree with
one of the leaves of another tree \cite{WhatOp}. Of course, for each leaf, we
get a different composite which is different from gluing the operation for
cobordisms, at first sight. An example is shown below, where the subscript of
the composition sign indicates the leaf chosen for the gluing.
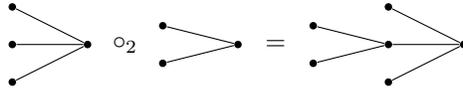
\[\begin{tikzpicture}

\colvert{black}{-2, 0}{a}
\colvert{black}{-3, 0.5}{a1};
\colvert{black}{-3, 0}{a2};
\colvert{black}{-3, -0.5}{a3};
\draw (a1) -- (a);
\draw (a2) -- (a);
\draw (a3) -- (a);

\lblvert{-1.5, 0}{comp}{$\circ_2$}

\colvert{black}{0, 0}{a}
\colvert{black}{-1, 0.25}{a1};
\colvert{black}{-1, -0.25}{a2};
\draw (a1) -- (a);
\draw (a2) -- (a);

\lblvert{0.5, 0}{eq}{$=$}

\colvert{black}{3, 0}{a}
\colvert{black}{2, 0.5}{a1};
\colvert{black}{2, 0}{a2};
\colvert{black}{2, -0.5}{a3};
\draw (a1) -- (a);
\draw (a2) -- (a);
\draw (a3) -- (a);

\colvert{black}{1, 0.25}{a4};
\colvert{black}{1, -0.25}{a5};
\draw (a4) -- (a2);
\draw (a5) -- (a2);

\end{tikzpicture}
\]

However, we notice that the difference of this situation with transport graphs
or cobordisms is artificial for we could define composition operations for
transport graphs and cobordisms parametrized by their inputs and outputs similar
to the case of operads. Note, however, that we need to handle the gluing of
multiple outputs to inputs in various combinations. For this, a relevant
modification of operads is as follows:

\begin{defn}[PROP]
Let $\s{C}$ be any (not necessarily symmetric) monoidal category and
$P = \set{P(n, m)}_{n, m \in \N}$ be a
collection of objects of $\s{C}$. For any $n \in \N$ and $1 \leq i \leq n$,
let $I = \set{i, i + 1, \dots, i + m} \subset \set{1, \dots, n}$. For each
such $n$ and $I$ as well as some $n' \in \N$, suppose there is a morphism in
$\s{C}$ as follows, called composition at $I$:
\[
  \circ_I : P(n, m) \tensor P(n', n) \to P(n', m)
\]
For each $n, m, p, q$, suppose there is a morphism in $\s{C}$ as follows,
called tensor product:
\[
  \square : P(n, m) \otimes P(q, p) \to P(n + q, m + p)
\]
Then, with some coherence conditions we will not describe, we will call
$P$ a PROP in $\s{C}$.
\end{defn}

\begin{rmk}
This is an instance of a much more intricate structure called a pasting scheme
\cite{JY15} along the lines of the first formulations of a PACT given in
\cite{MacLane65}. We also note that PROP is an abbreviation of ``Products and
Permuations Category''.
\end{rmk}

\begin{exm}
Consider the same diagram above giving an example of operads of trees. This
diagram is an example of a (pre)transport graph if we direct the edges and
colour the vertices! Take $\s{C}$ to be the monoidal category of transport
graphs and transport homomorphisms and composition to be gluing of targets of
transport graphs to a contiguous subset of the source of another
transport graph.
We note that the only difference with the definition given
above with an operad is that the composition operation
here is parametrized by a ``segment'' of the gluing site as opposed to a single
input, as is the case with operads. We give another example below to clarify the
difference:
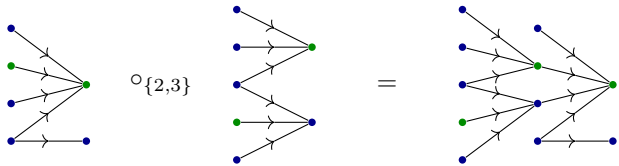
\[\begin{tikzpicture}

\colvert{green!55!black}{1, 2.5}{c}
\colvert{blue!55!black}{1, 1.75}{cc}
\colvert{blue!55!black}{0, 3.25}{c1}
\colvert{green!55!black}{0, 2.75}{c2}
\colvert{blue!55!black}{0, 2.25}{c3}
\colvert{blue!55!black}{0, 1.75}{c4}
\midarrow{c1}{c}
\midarrow{c2}{c}
\midarrow{c3}{c}
\midarrow{c4}{c}
\midarrow{c4}{cc}

\lblvert{2, 2.5}{comp}{$\circ_{\set{2, 3}}$}

\colvert{green!55!black}{4, 3}{a}
\colvert{blue!55!black}{3, 3.5}{a1}
\colvert{blue!55!black}{3, 3}{a2}
\colvert{blue!55!black}{3, 2.5}{a3}
\midarrow{a1}{a}
\midarrow{a2}{a}
\midarrow{a3}{a}

\colvert{blue!55!black}{4, 2}{b}
\colvert{green!55!black}{3, 2}{b1}
\colvert{blue!55!black}{3, 1.5}{b2}
\midarrow{b1}{b}
\midarrow{b2}{b}
\midarrow{a3}{b}

\lblvert{5, 2.5}{comp}{$=$}

\colvert{green!55!black}{8, 2.5}{c}
\colvert{blue!55!black}{8, 1.75}{cc}
\colvert{blue!55!black}{7, 3.25}{c1}
\colvert{green!55!black}{7, 2.75}{c2}
\colvert{blue!55!black}{7, 2.25}{c3}
\colvert{blue!55!black}{7, 1.75}{c4}
\midarrow{c1}{c}
\midarrow{c2}{c}
\midarrow{c3}{c}
\midarrow{c4}{c}
\midarrow{c4}{cc}

\colvert{blue!55!black}{6, 3.5}{a1}
\colvert{blue!55!black}{6, 3}{a2}
\colvert{blue!55!black}{6, 2.5}{a3}
\colvert{green!55!black}{6, 2}{b1}
\colvert{blue!55!black}{6, 1.5}{b2}
\midarrow{a1}{c2}
\midarrow{a2}{c2}
\midarrow{a3}{c2}
\midarrow{b1}{c3}
\midarrow{b2}{c3}
\midarrow{a3}{c3}

\end{tikzpicture}\]
Notice that the composition glues the output vertices of the right operand
with the second and third vertices of the left operand, as indicated by
$\circ_{\set{2, 3}}$.
In this example, we also considered coloured vertices and hence what we really
require is a coloured variant of PROPs, in the same spirit of introducing
colours to ordinary operads. Gluing different coloured
vertices requires some intermediary --- say, an arrow from green to blue --- but
the choice of this intermediary is not unique. Hence, it is easier to
require gluing of vertices of the same colour only.
\end{exm}

\begin{exm}
Of course, thick tangles admit the same structure. Given the output boundary
components of one thick tangle, we can choose a matching collection of input
boundary components of another and glue accordingly.
\end{exm}

Our theory of transport graphs is neatly captured in the formalism of PROPs.
That is, transport graphs and transport homomorphisms loosely form a coloured
PROP arising from a the pasting scheme of wheel-free graphs
\cite[xxiii]{JY15}.

We end this section by speculating on some possible connections with other
operad--like structures that feature prominently in geometry, topology and
physics.
Note, in particular, the similarity with cyclic operads
\cite{ModOp}. Recall that a cyclic operad $P$ is roughly an operad with the
action of the symmetric group $S_n$ on $P(n)$ replaced by an action of
$S_{n + 1}$. This effectively conflates the ``inputs'' and ``outputs'' of an
operad ``element''. Here, however, we take a different approach ---
simply allow multiple inputs and multiple outputs. The other difference is that
a PROP should not require symmetry in a monoidal category for our
categories of thick tangles and transport graphs are not equipped with symmetry.

Next, we comment on a possible connection to modular operads \cite{ModOp}.
A modular operad is roughly a cyclic operad which allows for the gluing of
``inputs'' and ``outputs'' of the same object \cite{Giansiracusa}.
This formalism is not present in our setting, but we can try to introduce
something similar, although in a non-unique way. Consider the simple example
below. Here, we wish to glue the first and only output to the first input, say.
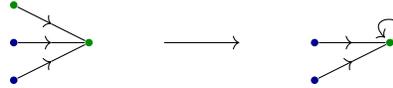
\[\begin{tikzpicture}
\colvert{green!55!black}{1, 3}{a}
\colvert{green!55!black}{0, 3.5}{a1}
\colvert{blue!55!black}{0, 3}{a2}
\colvert{blue!55!black}{0, 2.5}{a3}
\midarrow{a1}{a}
\midarrow{a2}{a}
\midarrow{a3}{a}

\draw[->] (2, 3) to (3, 3);

\colvert{green!55!black}{5, 3}{a}
\colvert{blue!55!black}{4, 3}{a2}
\colvert{blue!55!black}{4, 2.5}{a3}
\midarrow{a2}{a}
\midarrow{a3}{a}

\draw[->] (a) edge[loop] node {} (a);

\end{tikzpicture}\]
If we were to do this directly, it would introduce a loop in the graph which
contradicts our condition on transport graphs. In more general situations, we
would introduce cycles. Introducing cycles in a graph upsets the level ordering
of our graphs and hence the algorithm for extracting linear maps by parallel
transport. To remedy this, we must modify our definition of transport
graph and our algorithm to be able to interpret and process cycles in a
satisfactory way.
Furthermore, it is not only interesting, but also necessary to fit the addition
of cobordisms to the operadic formalism for a full translation of our framework
to operad theory.

We observe that the gluing constructions that have worked so far for transport
graphs, also work for bundles with connection as well so that our framework for
quantum computing could be treated in this operadic formalism. However, making
all the details precise and handling the issues discussed in the previous
paragraph is beyond the scope of the current work. Nevertheless,
these preliminary observations hint at various possibilities.
Modular operads as developed by Getzler and Kapranov generalize
Kontsevich's graph complexes \cite{ModOp}. Perhaps, we can find an
interpretation of these graph complexes or of the graphs defined in
\cite{ModOp} as parallel transport machinery, leading
to a solution of the self-gluing problem for transport graphs.
It if of interest to fully flesh out such connections to
modular operads for this could lead to the
application of Chern-Simons theory to the quantum computing framework
developed in this paper. On the other hand, an approach to cobordism
theory using coloured operads \cite{CobOp} has been used to categorify $sl_n$
quantum invariants. The operadic approach to transport graphs on manifolds
with connections might be one way to relate categorification in
representation theory with quantum computing.

\section{Further Directions: Graphs, Categorification, and Hyperbolic Matter}

We note that it is also possible to reformulate $2$-dimensional TQFTs in terms of graph-theoretic data via (dual) ribbon graphs and edge-contraction operations on these graphs, as formulated in \cite{DM}.  While \cite{DM} is restricted to ordinary TQFTs and does not treat thick ones, the notion of transport graph in our work and those of ribbon and cell graphs in \cite{DM} bear a resemblance that is worth exploring in further detail.  In particular, the edge-contraction operations may carry meaning in terms of quantum information.

The point of view in our article can be regarded as a particular instance in a much larger program of \emph{categorification} in mathematics.  While we treat this as a general theme and mantra in our setting, we may speculate on concrete ways in which our program interacts with categorification in, say, geometric representation theory.  We observe that our transport graphs are, in fact, quivers and thus give rise to well-defined Nakajima quiver varieties in the sense of \cite{HN1}.  These quiver varieties carry categorical actions of Lie algebras, both at the level of $K$-theory \cite{HN2} and at a geometric level \cite{CKL}.  These actions, especially the fully geometric version, are in turn a potential source of gates.  Again, we leave such speculations for future work.

Finally, one may ask how the constructions in this paper might be realized in a physical system of qubits, or whether there is a particular type of qubit that is better suited to these constructions than others.  One candidate is provided by the synthetic hyperbolic lattices studied experimentally by Koll\'ar, Fitzpatrick, and Houck in \cite{KFH} and theoretically by Maciejko and the second-named author in \cite{MR1}.  The effective quantum electrodynamics of these materials takes place on a Riemann surface of genus $g\geq2$.  Mathematically, the curve arises from the quotient of the hyperbolic plane $\mathbb H$ by the discrete Fuchsian group $\Gamma\subset\mbox{PSL}(2)$ of translations of the lattice.  Physically, one engineers a finite lattice but the Riemann surface still arises by considering a normal subgroup of $\Gamma$.  The experiments of \cite{KFH} realize these lattices artificially by tuning resonators in successive rings of the device so as to mimic the Poincar\'e metric. These negatively-curved materials have a well-defined electronic band theory \cite{MR1} and a complete Bloch wave decomposition \cite{MR2} in terms of the representations (of all ranks) of the fundamental group of the surface, which makes it possible to replace the eigenvalue problem for the Laplace-Beltrami operator and a periodic potential with purely topological data coming from the surface.  One can therefore imagine simulating thick, tangled TQFTs --- and thereby manipulations of quantum information, via our correspondence --- through a physical circuit of hyperbolic lattices, with genus $0$ and $1$ sites given by ordinary topological lattices.  Another recent article \cite{KR} in the theme of hyperbolic lattices presents a web of speculations and correspondences involving hyperbolic band theory and quantum field theories.  Our proposal expands this web, in effect, to include quantum information.

\pagebreak


\bibliography{TQFT_QC}{}
\bibliographystyle{acm}

\end{document}